%% file: main-ICECCS.tex
\PassOptionsToPackage{dvipsnames,svgnames,table}{xcolor}

\documentclass[conference]{IEEEtran}
\IEEEoverridecommandlockouts
\usepackage{cite}
\usepackage{amsmath,amssymb,amsfonts}
\usepackage{algorithmic}
\usepackage{graphicx}
\usepackage{textcomp}
\usepackage{xcolor}
\def\BibTeX{{\rm B\kern-.05em{\sc i\kern-.025em b}\kern-.08em
    T\kern-.1667em\lower.7ex\hbox{E}\kern-.125emX}}

\usepackage{balance} 
\usepackage{xspace}
\usepackage[super]{nth}  
\usepackage{tikz}
\usepackage{subcaption}
\usepackage{eurosym}
\usepackage[normalem]{ulem}
\usepackage{diagbox}
\usepackage{pifont}
\usepackage{mathtools}
\usepackage[fillcomment,nosemicolon,ruled,vlined,linesnumbered]{algorithm2e}
\usepackage{etoolbox}
\usepackage[hidelinks]{hyperref}
\usepackage[capitalise,english,nameinlink]{cleveref}
\usepackage{amsthm}
\usepackage[inline,shortlabels]{enumitem}

\newtheorem{example}{Example}

\newtheorem{proposition}{Proposition}

\crefname{algorithm}{alg.}{algs.}
\Crefname{algorithm}{Algorithm}{Algorithms}

\SetKwComment{Comment}{$\triangleright$\,}{}
\SetVlineSkip{0pt}

\SetCommentSty{algcomment}

\makeatletter
\patchcmd{\@algocf@start}
  {-1.5em}
  {0pt}
  {}{}
\makeatother

\usetikzlibrary{calc}
\usetikzlibrary{arrows,shapes,automata,backgrounds,petri,intersections}
\usetikzlibrary{decorations}
\usetikzlibrary{decorations.pathmorphing}
\usetikzlibrary{shadows}
\usetikzlibrary{shapes.symbols}
\usetikzlibrary{shapes.callouts}
\usetikzlibrary{shapes.geometric}
\usetikzlibrary{patterns}
\usetikzlibrary{positioning}
\usetikzlibrary{backgrounds}
\usetikzlibrary{fit,patterns,shapes.multipart,calc,matrix,shapes.gates.logic,shapes.gates.logic.US,circuits}

\tikzstyle{place}=[circle,thick,draw=blue!75,fill=blue!20,minimum size=6mm]
\tikzstyle{contour place}=[place,draw=red!100]
\tikzstyle{red place}=[place,draw=red!75,fill=red!20]
\tikzstyle{gray place}=[place,draw=black!100,fill=black!30]
\tikzstyle{transition}=[ rectangle,thick, fill=black, minimum width=8mm, inner ysep=2pt]
\tikzstyle{red transition}=[ rectangle,thick, fill=red!75, minimum width=8mm, inner ysep=2pt]
\tikzstyle{gray transition}=[ rectangle,thick, draw=black!100,fill=black!30, minimum width=8mm, inner ysep=2pt]
\tikzstyle{blue transition}=[ rectangle,thick, fill=blue!75, minimum width=8mm, inner ysep=2pt]
\tikzstyle{every label}=[black]
\tikzstyle{gateSEQ}=[diamond]
\tikzstyle{gateNULL}=[trapezium,trapezium left angle=120, trapezium right angle=120]
\tikzset{depth label/.style={
  label={[draw=none,green!60!black,label distance=-2pt]left:#1}}}
\tikzset{level label/.style={
  label={[draw=none,blue,label distance=-4pt]right:#1}}}
\tikzstyle{every node}=[initial text=]
\tikzstyle{location}=[rectangle, rounded corners, minimum size=12pt, draw=black, fill=blue!10, inner sep=2pt]
\tikzstyle{location10}=[location, minimum size=10pt]
\tikzstyle{invariant}=[draw=black, dotted, inner sep=1pt] 

\tikzstyle{goodtile} = [draw=green!50!black, fill=green]
\tikzstyle{badtile} = [draw=red!50!black, fill=red]

\tikzstyle{invisible}=[draw=none]
\tikzstyle{final}=[double]
\tikzstyle{urgent}=[fill=yellow!50]
\tikzstyle{bad}=[fill=red!50]

\tikzstyle{every node}=[initial text=]
\tikzstyle{final}=[double]
\tikzstyle{sync}=[draw=blue,thick]
\tikzstyle{square}=[regular polygon,regular polygon sides=4]

\tikzstyle{seq}=[path picture={
	\draw[->]  ([xshift=#1, yshift=6pt]path picture bounding box.south)
					-- ([xshift=#1, yshift=-7pt]path picture bounding box.south);}]

\input{macros}

\begin{document}

\title{Minimal Schedule with Minimal Number of Agents in Attack-Defence Trees
\thanks{The authors acknowledge the support of CNRS and PAN, under the IEA project MoSART,
and of NCBR Poland and FNR Luxembourg, under the PolLux/FNR-CORE project STV (POLLUX-VII/1/2019).}
}

\author{\IEEEauthorblockN{Jaime Arias\textsuperscript{1}, Laure Petrucci\textsuperscript{1}}
\IEEEauthorblockA{\textsuperscript{1}\textit{LIPN, CNRS UMR 7030,} \\
\textit{Université Sorbonne Paris Nord}\\
Villetaneuse, France \\
\{arias, petrucci\}@lipn.univ-paris13.fr}
\and
\IEEEauthorblockN{{\L}ukasz Maśko\textsuperscript{2}, Wojciech Penczek\textsuperscript{2}, Teofil Sidoruk\textsuperscript{2, 3}}
\IEEEauthorblockA{\textsuperscript{2}\textit{Institute of Computer Science, Polish Academy of Sciences} \\
\textsuperscript{3}\textit{Faculty of Math. and Inf. Science, Warsaw University of Technology}\\
Warsaw, Poland \\
\{masko, penczek, t.sidoruk\}@ipipan.waw.pl}
}

\maketitle

\begin{abstract}
  Expressing attack-defence trees in a multi-agent setting allows for studying a new
  aspect of security scenarios,
	namely how the number of agents and their task assignment impact the performance,
	\eg attack time, of strategies executed by opposing coalitions.
	Optimal scheduling of agents' actions, a non-trivial problem, is thus vital.
	We discuss associated caveats and	propose an algorithm that synthesises such an assignment,
	targeting minimal attack time and using minimal number of agents for a given attack-defence tree.
\end{abstract}


\section{Introduction}
\label{sec:intro}
\input{introduction}

\section{Attack-Defence Trees}
\label{sec:ADT}
\input{adt}

\section{Preprocessing the tree}
\label{sec:preprocess}
\input{preprocess}

\section{Best minimal agent assignment}
\label{sec:algo}
\input{algo}

\section{Experiments}
\label{sec:expe}
\input{experiments}

\section{Conclusion}
\label{sec:conclusion}
\input{conclusion}

\bibliographystyle{IEEEtran}
\balance
\bibliography{refs}


\clearpage
\appendices

\section{forestall}
\label{app:forestall}
\input{examples/forestall.tex}

\clearpage
\section{iot-dev}
\label{app:iotdev}
\input{examples/iot-dev.tex}

\clearpage
\section{gain-admin}
\label{app:gainadmin}
\input{examples/gain-admin.tex}

\clearpage
\section{interrupted}
\label{app:interrupted}
\input{examples/interrupted.tex}

\newpage
\section{example from~\protect\cite{ICFEM2020}}
\label{app:icfem2020}
\input{examples/icfem2020_example.tex}

\clearpage
\section{last}
\label{app:last}
\input{examples/last.tex}

\newpage
\section{scaling example}
\label{app:scaling}
\input{examples/scaling_example.tex}
\input{examples/scaling_table.tex}

\end{document}

%% file: macros.tex
\usepackage{amsmath} 
\usepackage{amssymb} 

\ifdefined \VersionWithComments
  \usepackage{marginnote}
  \newcommand{\marginX}{\marginnote{\huge{\quad\textbf{!}\quad}}}
  \newcommand{\wop}[1]{\textcolor{red!50!blue}{\marginX{}[\textbf{Wojciech}: #1
      ]}}
  \newcommand{\lp}[1]{\textcolor{blue!50}{\marginX{}[\textbf{Laure}: #1 ]}}
  \newcommand{\teo}[1]{\textcolor{green!50!black}{\marginX{}[\textbf{Teofil}: #1]}}
  \newcommand{\ja}[1]{\textcolor{DarkOrchid}{\marginX{}[\textbf{Jaime}: #1 ]}}
  \newcommand{\todo}[1]{\textcolor{red}{\marginX{}TODO: #1}}
\else
  \newcommand{\wop}[1]{}
  \newcommand{\lp}[1]{}
  \newcommand{\ja}[1]{}
  \newcommand{\teo}[1]{}
  \newcommand{\todo}[1]{}
\fi

\newcommand{\eg}{\emph{e.g.}\xspace}
\newcommand{\ie}{\emph{i.e.}\xspace}

\newcommand{\wrt}{\emph{w.r.t.}\xspace}

\newcommand{\tool}{\texttt{ADT2AMAS}\xspace}

\def\ADT/{ADTree}
\def\DAG/{DAG}

\newcommand{\actAttack}{\ensuremath{a}}
\newcommand{\ActAttack}{\ensuremath{A}}

\newcommand{\actDefence}{\ensuremath{d}}

\newcommand{\leaf}[1]{\ensuremath{\mathtt{#1}}\xspace}
\newcommand{\gate}[1]{\ensuremath{\mathtt{\MakeUppercase{#1}}}\xspace}
\newcommand{\gateAND}{\gate{and}}
\newcommand{\gateOR}{\gate{or}}
\newcommand{\gateCAND}{\gate{cand}}

\newcommand{\gateLEAF}{\gate{leaf}}
\newcommand{\gateNULL}{\gate{null}}
\newcommand{\gateNODEF}{\gate{nodef}}
\newcommand{\gateSAND}{\gate{sand}}
\newcommand{\gateSEQ}{\gate{seq}}

\newcommand{\gateSCAND}{\gate{scand}}

\newcommand{\nodeLabel}[1]{\ensuremath{\mathtt{#1}}\xspace}

\newcommand{\initTime}[1]{\ensuremath{\mathit{init\_time}(#1)}}
\newcommand{\Time}[1]{\ensuremath{\mathit{time}(#1)}}

\newcommand{\tunit}{\ensuremath{t_{\mathit{unit}}}}
\newcommand{\DAGset}{\ensuremath{\mathit{DAG\_set}}}
\newcommand{\inputDAG}{\ensuremath{\mathit{DAG}}}
\newcommand{\workingSet}{\ensuremath{S}}
\newcommand{\N}{\ensuremath{N}}

\newcommand{\varAgent}{\ensuremath{\mathit{agent}}}
\newcommand{\varSlot}{\ensuremath{\mathit{slot}}}
\newcommand{\varSlots}{\ensuremath{\mathit{slots}}}
\newcommand{\varChildNode}{\ensuremath{\mathit{child\_node}}}
\newcommand{\varCurrentNode}{\ensuremath{\mathit{current\_node}}}
\newcommand{\varParentAgent}{\ensuremath{\mathit{par\_agent}}}
\newcommand{\varParentNode}{\ensuremath{\mathit{parent\_node}}}
\newcommand{\varNodesLeft}{\ensuremath{\mathit{n\_remain}}}
\newcommand{\varBound}{\ensuremath{\mathit{lower\_bound}}}
\newcommand{\varUpperBound}{\ensuremath{\mathit{upper\_bound}}}
\newcommand{\varMaxAgents}{\ensuremath{\mathit{max\_agents}}}
\newcommand{\varNumAgents}{\ensuremath{\mathit{agents}}}
\newcommand{\varOutput}{\ensuremath{\mathit{output}}}
\newcommand{\varCurrentOutput}{\ensuremath{\mathit{curr\_output}}}
\newcommand{\varCandidate}{\ensuremath{\mathit{candidate}}}

\newcommand{\numLevels}{\ensuremath{\mathit{L}}}

\newcommand{\agent}{\ensuremath{\mathit{agent}}}

\newcommand{\children}{\ensuremath{\mathit{child}}}
\newcommand{\depth}{\ensuremath{\mathit{depth}}}

\newcommand{\keep}{\ensuremath{\mathit{keep}}}
\newcommand{\level}{\ensuremath{\mathit{level}}}
\newcommand{\node}{\ensuremath{\mathit{node}}}
\newcommand{\parent}{\ensuremath{\mathit{parent}}}
\newcommand{\slot}{\ensuremath{\mathit{slot}}}
\newcommand{\type}{\ensuremath{\mathit{type}}}

\def\picADT#1{\scalebox{.5}{\input{examples/#1}}}
\def\picTable#1{\scalebox{.555}{\input{examples/#1}}}
\newcommand{\casestudy}[1]{\textnormal{\small\textsf{\mdseries #1}}\xspace}
\newcommand{\csfs}{\casestudy{forestall}}
\newcommand{\csiot}{\casestudy{iot-dev}}
\newcommand{\csadmin}{\casestudy{gain-admin}}

%% file: introduction.tex


Security of safety-critical multi-agent systems \cite {Wooldridge02a} is a major challenge. 
Attack-defence trees \cite{KordyMRS10,Zaruhi_pareto_2015} have been developed 
to evaluate the safety of systems and to study interactions between attacker and defender parties.
They provide a simple graphical formalism of possible attacker's actions to be taken in order
to attack a system and the defender's defenses employed to protect the system. 
Recently, it has been proposed to model attack-defence trees (\ADT/s) in the formalism of
asynchronous multi-agent systems (AMAS) extended with certain \ADT/ 
characteristics~\cite{ICECCS2019,ICFEM2020}.
In this setting, one can reason about attack/defence scenarios 
considering agent distributions over tree nodes and their impact on the feasibility and performance
(quantified by metrics such as time and cost) of attacking and defending strategies executed by specific coalitions.

\subsection{Minimal schedule with minimal number of agents}

The time metric, on which we focus here, is clearly affected by
both the number of available agents and their distribution over \ADT/ nodes. 
Hence, there arises the problem of optimal scheduling, \ie{} obtaining an assignment that achieves
the lowest possible time, while using the minimum number of agents required for an attack to be feasible. 
To that end, we first preprocess the input \ADT/, transforming it into a Directed Acyclic Graph (\DAG/),
where specific types of \ADT/ gates are replaced with sequences of nodes with normalised time 
(\ie{} duration of either zero, or the greatest common factor across all nodes of
the original \ADT/).
Because some \ADT/ constructs (namely, \gateOR gates and defences) induce multiple alternative outcomes,
we execute the scheduling algorithm itself on a number of independently considered \DAG/ variants.
For each such variant, we begin by iterating over its nodes, labeling them in accordance with their depth and height in the graph.
The lowest possible attack time follows from the longest sequence of nodes in the \DAG/,
because these actions cannot be handled any faster by executing them in parallel.
In order to minimise the number of agents as well, we synthesise a schedule multiple times in a divide-and-conquer strategy,
adjusting the number of agents until the lowest one that produces a valid assignment is found.
Since we preserve labels during the preprocessing step, all \DAG/ nodes are traceable back to specific gates and leaves of the original \ADT/.
Thus, in the final step we ensure that the same agent is assigned to nodes of the same origin, reshuffling the schedule if necessary.

While there are some clear parallels with multi-processor task scheduling,
the \ADT/ formalism also introduces a number of unique caveats to consider.
Consequently, our approach differs from that of classical process scheduling,
whose techniques cannot be directly applied in this case.
We expand on this comparison in \Cref{sec:intro:relatedwork}.

\subsection{Contributions}

In this paper, we:
\begin{enumerate*}[($i$)]
	\item present and prove the correctness of an algorithm for \ADT/s which finds an optimal assignment of the minimal number of agents for all possible \DAG/ variants of a given attack/defence scenario,
	\item show the scheduling algorithm's complexity to be quadratic in the number of nodes of its preprocessed input \DAG/, 
	\item implement the algorithm in our tool \tool and evaluate experimental results.
\end{enumerate*}

\subsection{Related work}
\label{sec:intro:relatedwork}

\ADT/s \cite{KordyMRS10,KMRS14}
are a popular formalism that has been implemented in a broad range
of analysis frameworks \cite{Gadyatskaya_2016,AGKS15,GIM15,ANP16},
comprehensively surveyed in \cite{KPCS13,WAFP19}.
They remain extensively studied today \cite{FilaW20}.
Of particular relevance is the \ADT/ to AMAS translation \cite{ICFEM2020},
based on the semantics from \cite{AAMASWJWPPDAM2018a}.
Furthermore, the problem discussed in this paper is clearly related to
parallel program scheduling \cite{CompSched,KwokA99}.
Due to time normalisation, it falls into the category of Unit Computational Cost (UCC) graph scheduling problems,
which can be effectively solved for tree-like structures \cite{Hu1961},
but cannot be directly applied to a set of \DAG/s.
Although a polynomial solution for interval-ordered \DAG/s was proposed by \cite{PapaYan1979},
their algorithm does not guarantee the minimal number of agents.
Due to zero-cost communication in all considered graphs, the problem can also be classified as No Communication (NC) graph scheduling.
A number of heuristic algorithms using list scheduling were proposed \cite{CompSched}, including
Highest Levels First with No Estimated Times (HLFNET),
Smallest Co-levels First with no Estimated Times (SCFNET),
and Random, where nodes in the DAG are assigned priorities randomly.
Variants assuming non-uniform node computation times are also considered, 
but are not applicable to the problem solved in this paper.
Furthermore, this class of algorithms does not aim at finding a schedule with the minimal number of processors or agents.
On the other hand, known algorithms that include such a limit, \ie for the Bounded Number of Processors (BNP) class of problems,
assume non-zero communication cost and rely on the clustering technique,
reducing communication, and thus schedule length, by mapping nodes to processing units.
Hence, these techniques are not directly applicable. 

The algorithm described in this paper can be classified as list scheduling with a fusion of HLFNET and SCFNET heuristics,
but with additional restriction on the number of agents used.
The length of a schedule is determined as the length of the critical path of a graph.
The number of minimal agents needed for the schedule is found with bisection.

Branching schedules analogous to the variants discussed in \Cref{sec:preprocess} have been previously explored, 
albeit using different models that either include probability \cite{Towsley86}
or require an additional \DAG/ to store possible executions \cite{ElRewiniA95}.
Zero duration nodes are also unique to the \ADT/ setting.

To the best of our knowledge, this is the first work dealing with agents in this context.
Rather, scheduling in multi-agent systems typically focuses on agents' \textit{choices}
in cooperative or competitive scenarios, \eg in models such as BDI \cite{NunesL14,DannT0L20}.

\subsection{Outline}

The rest of the paper is structured as follows.
The next section briefly recalls the \ADT/ formalism and its translation to multi-agent systems.
In \Cref{sec:preprocess}, several preprocessing steps are discussed,
including transforming the input tree to a DAG, normalising node attributes,
and handling different types of nodes. 
\Cref{sec:algo} presents steps performed by the main algorithm, as well as a proof of its correctness and optimality.
The algorithm, implemented in our tool \tool\cite{ADT2AMASDemoPaper},
is demonstrated in practice on three use cases in \Cref{sec:expe}.
The final section contains conclusions and perspectives.

%% file: adt.tex


To keep the paper self-contained, we briefly recall the basics of \ADT/s
and their translation to a multi-agent setting.

\subsection{Attack-defence trees}
\label{sec:adt:adt}

\ADT/s are a well-known formalism that models security scenarios
as an interplay between attacking and defending parties.
\Cref{fig:adt:constructs} depicts basic constructs
used in examples throughout the paper.
For a more comprehensive overview, we refer the reader to~\cite{ICFEM2020}.

\input{adt-constructs}

Attacking and defending actions are depicted in red and green, respectively.
Leaves represent individual actions at the highest level of granularity.
Different types of gates allow for modeling increasingly broad intermediary goals,
all the way up to the root, which corresponds to the overall objective.
\gateOR and \gateAND gates are defined analogously to their logical counterparts.
\gateSAND is a sequential variant of the latter, \ie the entire subtree $a_i$ needs to be completed before handling $a_{i+1}$.
Although only shown in attacking subtrees here, these gates may refine defending goals in the same manner.
Reactive or passive countering actions can be expressed using gates
\gateCAND (counter defence; successful iff \actAttack{} succeeds and \actDefence{} fails),
\gateNODEF (no defence; successful iff either \actAttack{} succeeds or \actDefence{} fails),
and \gateSCAND (failed reactive defence; sequential variant of \gateCAND, where \actAttack{} occurs first).
We collectively refer to gates and leaves as \textit{nodes}.

\ADT/ nodes may additionally have numerical \textit{attributes},
\eg the time needed for an attack, or its financial cost.
Boolean functions over these attributes, called \textit{conditions},
may then be associated with counter-defence nodes to serve as additional constraints
for the success or failure of a defending action.

In the following, the treasure hunters \ADT/ in \Cref{fig:adt:treasure}
will be used as a running example.
While both the gatekeeper \leaf{b} and the door \leaf{f}
need to be taken care of to steal the treasure (\gate{ST}),
just one escape route (either \leaf{h} or \leaf{e}) is needed
to flee (\gate{GA}), with \gate{TF} enforcing sequentiality.
Note the \gate{TS} counter defence with an additional constraint on the police's success \leaf{p}.


\begin{figure}[!!htb]
	\hspace{-1em}%
	\begin{subfigure}[b]{0.45\linewidth}
		\centering
		\scalebox{.75}{
			\begin{tikzpicture}
				[every node/.style={ultra thick,draw=red,minimum size=6mm}]
				\node[and gate US,point up,logic gate inputs=ni] (ca)
				{\rotatebox{-90}{\texttt{TS}}};
				\node[rectangle,draw=Green,minimum size=8mm, below = 5mm of ca.west, xshift=10mm] (d) {\texttt{p}};
				\draw (d.north) -- ([yshift=0.28cm]d.north) -| (ca.input 2);
				\node[and gate US,point up,logic gate inputs=nn, seq=4pt, below = 9mm of ca.west, yshift=14mm] (A)
				{\rotatebox{-90}{\texttt{TF}}};
				\draw (A.east) -- ([yshift=0.15cm]A.east) -| (ca.input 1);
				\node[and gate US,point up,logic gate inputs=nn, below = 9mm of A.west, yshift=14mm] (a1)
				{\rotatebox{-90}{\texttt{ST}}};
				\draw (a1.east) -- ([yshift=0.15cm]a1.east) -| (A.input 1);
				\node[state, below = 4mm of a1.west, xshift=-5mm] (a1n) {\texttt{b}};
				\draw (a1n.north) -- ([yshift=0.15cm]a1n.north) -| (a1.input 1);
				\node[state, below=4mm of a1.west, xshift=5mm] (a11) {\texttt{f}};
				\draw (a11.north) -- ([yshift=0.15cm]a11.north) -| (a1.input 2);
				\node[or gate US,point up,logic gate inputs=nn, below = 10mm of A.west, yshift=-7mm] (a2)
				{\rotatebox{-90}{\texttt{GA}}};
				\draw (a2.east) -- ([yshift=0.15cm]a2.east) -| (A.input 2);
				\node[state, below = 4mm of a2.west, xshift=-5mm] (a21) {\texttt{h}};
				\draw (a21.north) -- ([yshift=0.15cm]a21.north) -| (a2.input 1);
				\node[state, below=4mm of a2.west, xshift=5mm] (a2n) {\texttt{e}};
				\draw (a2n.north) -- ([yshift=0.15cm]a2n.north) -| (a2.input 2);
			\end{tikzpicture}
		}
		\subcaption{\ADT/}
	\end{subfigure}
	\hspace{-.5em}
	\begin{subfigure}[b]{.45\linewidth}
		\centering
		\scalebox{.75}{\parbox{\linewidth}{%
				~%
				\begin{tabular}{r@{~}l@{~\;}r@{~\;}r}
					\multicolumn{2}{l}{\bf Name} & {\bf Cost}         & {\bf Time}              \\
					\hline
					\gate{TS}                    & (treasure stolen)  &                &        \\
					\leaf{p}                     & (police)           & \EUR{100}      & 10 min \\
					\gate{TF}                    & (thieves fleeing)  &                &        \\
					\gate{ST}                    & (steal treasure)   & 							 & 2 min  \\
					\leaf{b}                     & (bribe gatekeeper) & \EUR{500}      & 1 h    \\
					\leaf{f}                     & (force arm. door) 	& \EUR{100}      & 2 h    \\
					\gate{GA}                    & (get away)         &                &        \\
					\leaf{h}                     & (helicopter)       & \EUR{500}      & 3 min  \\
					\leaf{e}                     & (emergency exit)   &                & 10 min
				\end{tabular}
				~\\[.2ex]
				\uline{\textbf{Condition for \texttt{TS}}}:\\[.2ex]
				\mbox{$\initTime{\mathtt{p}} > \initTime{\mathtt{ST}} + \Time{\mathtt{GA}}$}
			}}
		\vspace{1ex}
		\subcaption{Attributes of nodes}
	\end{subfigure}
	\caption{Running example: treasure hunters}
	\label{fig:adt:treasure}
\end{figure}

\subsection{Translation to extended AMAS}
\label{sec:adt:adt2amas}

Asynchronous multi-agent systems (AMAS) \cite{AAMASWJWPPDAM2018a} are essentially networks of automata,
which synchronise on shared transitions and interleave private ones for asynchronous execution.
An extension of this formalism with attributes and conditional constraints to model \ADT/s,
and the translation of the latter to extended AMAS, were proposed in \cite{ICFEM2020}.
Intuitively, each node of the \ADT/ corresponds to a single automaton in the resulting network.
Specific patterns, embedding reductions to minimise state space explosion \cite{ICECCS2019},
are used for different types of \ADT/ constructs.
As the specifics exceed the scope and space of this paper, we refer the reader to \cite{AAMASWJWPPDAM2018a} for the AMAS semantics,
and to \cite{ICFEM2020} for the details on the translation.


In the multi-agent setting, groups of agents working for the attacking and defending parties can be considered.
Note that the \textit{feasibility} of an attack is not affected by the number or distribution of agents over \ADT/ nodes,
as opposed to some \textit{performance} metrics, such as time 
(\eg a lone agent can handle all the actions sequentially, albeit usually much slower).

\subsection{Assignment of agents for \ADT/s}
\label{sec:adt:sched}

Consequently, the optimal distribution of agent coalitions is of vital importance for both parties,
allowing them to prepare for multiple scenarios, depending on how many agents they can afford to recruit
(thereby delaying or speeding up the completion of the main goal).
For instance, the thieves in \Cref{fig:adt:treasure}, knowing the police response time, would have to plan accordingly
by bringing a sufficiently large team and, more importantly, schedule their tasks to make the most of these numbers.
Thus, we can formulate two relevant and non-trivial scheduling problems.
%
{\em The first one}, not directly addressed here, is obtaining the assignment using a given number of agents that results 
in optimal execution time.
{\em The second one}, on which we focus in this paper, is synthesising an assignment that achieves a particular execution 
time using the least possible number of agents.
Typically, the minimum possible time is of interest here.
As we show in \Cref{sec:preprocess}, this time is easily obtainable from the structure of the input \ADT/ itself
(and, of course, the time attribute of nodes).
However, our approach can also target a longer attack time if desired.
In the next section, we discuss it in more detail as normalisation of the input tree is considered, 
along with several other preprocessing steps.

%% file: adt-constructs.tex

\begin{figure}[!!htb]
	\begingroup
		\def\nodesubtree{\node[isosceles triangle, isosceles triangle apex angle=70, shape border rotate=90, minimum size=6.5mm]}
		\def\nodeattack{\node[state,draw=red,ultra thick]}
		\def\nodedefence{\node[rectangle,draw=Green,ultra thick]}
	\captionsetup[subfigure]{justification=centering}
	\centering
	\begin{subfigure}[b]{0.2\linewidth}
		\centering
		\scalebox{0.5}{
			\begin{tikzpicture}
				\normalsize
				\nodeattack[minimum size=6mm] (LA) {$\actAttack{}$};
			\end{tikzpicture}
		}
		\subcaption{leaf (attack)}
	\end{subfigure}
	\hfill
	\begin{subfigure}[b]{0.2\linewidth}
		\centering
		\scalebox{0.5}{
			\begin{tikzpicture}[every node/.style={ultra thick,draw=red,minimum size=6mm}]
				\normalsize
				\node[and gate US,point up,logic gate inputs=nn] (A)
					{\rotatebox{-90}{\large$\ActAttack$}};
				\nodesubtree at (-1,-1.3) (a1) {$\!\actAttack_1\!\!$};
				\nodesubtree at ( 1,-1.3) (an) {$\!\actAttack_n\!\!$};
				\node[draw=none] at (0,-1.3) {$\cdots$};
				\draw (a1.north) -- ([yshift=1.5mm]a1.north) -| (A.input 1);
				\draw (an.north) -- ([yshift=1.3mm]an.north) -| (A.input 2);
			\end{tikzpicture}
		}
		\subcaption{\gateAND}
	\end{subfigure}
	\hfill
	\begin{subfigure}[b]{0.2\linewidth}
		\centering
		\scalebox{0.5}{
			\begin{tikzpicture}[every node/.style={ultra thick,draw=red,minimum size=6mm}]
				\normalsize
				\node[or gate US,point up,logic gate inputs=nn] (A)
					{\rotatebox{-90}{\large$\ActAttack$}};
				\nodesubtree at (-1,-1.3) (a1) {$\!\actAttack_1\!\!$};
				\nodesubtree at ( 1,-1.3) (an) {$\!\actAttack_n\!\!$};
				\node[draw=none] at (0,-1.3) {$\cdots$};
				\draw (a1.north) -- ([yshift=1.5mm]a1.north) -| (A.input 1);
				\draw (an.north) -- ([yshift=1.3mm]an.north) -| (A.input 2);
			\end{tikzpicture}
		}
		\subcaption{\gateOR}
	\end{subfigure}
	\hfill
		\begin{subfigure}[b]{0.2\linewidth}
		\centering
		\scalebox{0.5}{
			\begin{tikzpicture}[every node/.style={ultra thick,draw=red,minimum size=6mm}]
				\normalsize
				\node[and gate US,point up,logic gate inputs=nn, seq=4pt] (A)
					{\rotatebox{-90}{\large$\ActAttack$}};
				\nodesubtree at (-1,-1.3) (a1) {$\!\actAttack_1\!\!$};
				\nodesubtree at ( 1,-1.3) (an) {$\!\actAttack_n\!\!$};
				\node[draw=none] at (0,-1.3) {$\cdots$};
				\draw (a1.north) -- ([yshift=1.5mm]a1.north) -| (A.input 1);
				\draw (an.north) -- ([yshift=1.5mm]an.north) -| (A.input 2);
			\end{tikzpicture}
		}
		\subcaption{\gateSAND}
		\end{subfigure}
	\hfill
		\begin{subfigure}[b]{0.2\linewidth}
			\centering
			\scalebox{0.5}{
				\begin{tikzpicture}
					\normalsize
					\nodedefence[minimum size=6mm] (LD) {$\actDefence{}$};
				\end{tikzpicture}
			}
			\subcaption{leaf (defence)}
		\end{subfigure}
	\hfill
		\begin{subfigure}[b]{0.2\linewidth}
		\centering
		\scalebox{0.5}{
			\begin{tikzpicture}[every node/.style={ultra thick,draw=red,minimum size=6mm}]
				\normalsize
				\node[and gate US,point up,logic gate inputs=ni] (A)
					{\rotatebox{-90}{\large$\ActAttack$}};
				\nodesubtree             at (-1,-1.36) (a) {$\actAttack$};
				\nodesubtree[draw=Green] at ( 1,-1.32) (d) {$\!\actDefence\!$};
				\draw (a.north) -- ([yshift=1.4mm]a.north) -| (A.input 1);
				\draw (d.north) -- ([yshift=1.5mm]d.north) -| (A.input 2);
			\end{tikzpicture}
		}
		\subcaption{\gateCAND}
	\end{subfigure}
	\hfill
		\begin{subfigure}[b]{0.2\linewidth}
		\centering
		\scalebox{0.5}{
			\begin{tikzpicture}[every node/.style={ultra thick,draw=red,minimum size=6mm}]
				\normalsize
				\node[or gate US,point up,logic gate inputs=ni] (A)
					{\rotatebox{-90}{\large$\ActAttack$}};
				\nodesubtree             at (-1,-1.28) (a) {$\actAttack$};
				\nodesubtree[draw=Green] at ( 1,-1.24) (d) {$\!\actDefence\!$};
				\draw (a.north) -- ([yshift=1.4mm]a.north) -| (A.input 1);
				\draw (d.north) -- ([yshift=1.5mm]d.north) -| (A.input 2);
			\end{tikzpicture}
		}
		\subcaption{\gateNODEF}
	\end{subfigure}
	\hfill
		\begin{subfigure}[b]{0.2\linewidth}
		\centering
		\scalebox{0.5}{
			\begin{tikzpicture}[every node/.style={ultra thick,draw=red,minimum size=6mm}]
				\normalsize
				\node[and gate US,point up,logic gate inputs=ni, seq=9pt] (A)
					{\rotatebox{-90}{\large$\ActAttack$}};
				\nodesubtree             at (-1,-1.36) (a) {$\actAttack$};
				\nodesubtree[draw=Green] at ( 1,-1.32) (d) {$\!\actDefence\!$};
				\draw (a.north) -- ([yshift=1.4mm]a.north) -| (A.input 1);
				\draw (d.north) -- ([yshift=1.5mm]d.north) -| (A.input 2);
			\end{tikzpicture}
		}
		\subcaption{\gateSCAND}
	\end{subfigure}
	\endgroup

	\caption{Basic \ADT/ constructs \label{fig:adt:constructs}}
\end{figure}

%% file: preprocess.tex

In this preprocessing step, an \ADT/ is transformed into
\DAG/s (\emph{Directed Acyclic Graphs}) of actions of the same duration.
This is achieved by splitting nodes into sequences of such actions, mimicking the
scheduling enforced by \ADT/s sequential gates, and considering the different possibilities
of defences.
Therefore, we introduce a sequential node \gateSEQ, which only waits for some input,
processes it and produces some output. It is depicted as a lozenge (see nodes 
$N_1$ and $N_2$ in \Cref{fig:pre:time:ex}).


In what follows, we assume one
time unit is the greatest common factor of time durations across all nodes in the
input \ADT/, \ie{} $\tunit = \mathit{gcf}(t_{N_1} \dots t_{N_{|\ADT/|}})$.
By \textit{time slots}, we refer to fragments of the schedule whose length is $\tunit$.
That is, after normalisation, one agent can handle exactly one node of non-zero duration
within a single time slot.
%
Note that, during the preprocessing steps described in this section, node labels are
preserved to ensure backwards traceability. 
Their new versions are either primed or indexed.

\subsection{Nodes with no duration}
\label{sec:pre:zero}
It happens that several nodes have no time parameter set, and are thus considered to have a duration of $0$. 
Such nodes play essentially a structuring role.
Since they do not take any time, the following proposition is straightforward.

\begin{proposition}
	\label{prop:pre:zero}
	Nodes with duration $0$ can always be scheduled immediately before their parent node
	or after their last occurring child, using the same agent in the same time slot.
\end{proposition}

Preprocessing introduces nodes similar to \gateSEQ{} but with 0~duration, called \gateNULL and depicted as trapeziums (Fig.~\ref{fig:pre:sand:ex}).

\subsection{Normalising time}
\label{sec:pre:norm}

The first pre-processing step prior to applying the scheduling algorithm
normalises the time parameter of nodes.

\begin{proposition}
	\label{prop:pre:time}
	Any node $N$ of duration $t_N = n \times \tunit, n\neq0$ can be
	replaced with an equivalent sequence consisting of a node $N'$ 
	(differing from $N$ only in its $0$ duration) and $n$ \gateSEQ nodes $N_1$,
	\ldots, $N_n$ of duration $\tunit$.
\end{proposition}
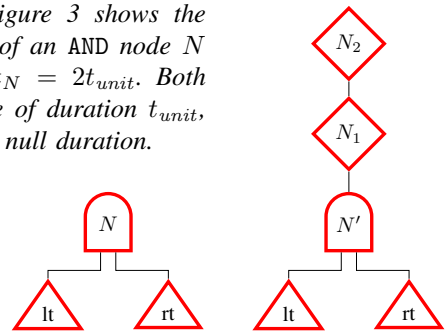
\begin{figure}[!!htb]
\centering
\begingroup
\def\nodesubtree{\node[isosceles triangle, isosceles triangle apex angle=70, shape border rotate=90, minimum size=6.5mm]}
\scalebox{0.8}{
	\begin{tikzpicture}[every node/.style={ultra thick,draw=red},
			node distance=1.5cm]
		\node[and gate US,point up,logic gate inputs=nn] (N')
		{\rotatebox{-90}{$N'$}};
		\node[gateSEQ,above of= N'] (N1) {$N_1$};
		\node[gateSEQ,above of= N1] (N2) {$N_2$};
		\nodesubtree[below of= N',xshift=-1cm] (l') {lt};
		\nodesubtree[below of= N',xshift=1cm] (r') {rt};
		\draw(N'.input 1) -- ([yshift=-3mm]N'.input 1) -| (l');
		\draw(N'.input 2) -- ([yshift=-3mm]N'.input 2) -| (r');
		\draw (N') -- (N1) --  (N2); 

		\node[and gate US,point up,logic gate inputs=nn,
			above of=N',node distance=4cm] (N)
		{\rotatebox{-90}{$N$}};
		\nodesubtree[below of= N,xshift=-1cm] (l) {lt};
		\nodesubtree[below of= N,xshift=1cm] (r) {rt};
		\draw(N.input 1) -- ([yshift=-3mm]N.input 1) -| (l);
		\draw(N.input 2) -- ([yshift=-3mm]N.input 2) -| (r);
		\node [above of= N,xshift=-1.5cm,yshift=1.2cm,text width=5cm,align=justify,scale=1/0.8,draw=none]
		{
			\begin{example}
				\Cref{fig:pre:time:ex} shows the transformation of an \gateAND node $N$ with duration
				$t_N=2\tunit{}$. Both $N_2$ and $N_1$ are of duration
				$\tunit{}$, while $N'$ has a null duration.
			\end{example}
		};
	\end{tikzpicture}
}
\endgroup

\caption{Normalising \gateAND node $N$ ($t_N=2 \tunit{}$, $t_{N'}=0$)\label{fig:pre:time:ex}}
\end{figure}




\subsection{Scheduling enforcement}
\label{sec:pre:sched}
Sequential nodes \gateSAND 
enforce some scheduling. 
These are transformed into a sequence containing their subtrees and \gateNULL nodes.

\begin{proposition}
	\label{prop:pre:sand}
	Any \gateSAND node $N$ with children subtrees $T_1$, \ldots, $T_n$ can be replaced
	with an equivalent sequence $T_1$, $N_1$, $T_2$, \ldots, $N_{n-1}$, $T_n$, $N_n$, where
	each $N_i$ is a \gateNULL node, its input is the output of $T_i$ and its outputs
	are the leaves of $T_{i+1}$ (except for $N_n$ which has the same output as $N$ if
	any).
\end{proposition}




\begingroup

\begin{figure}[!!htb]
\centering
\begingroup
\def\nodesubtree{\node[isosceles triangle, isosceles triangle apex angle=70, shape border rotate=90, minimum size=6.5mm]}

\vspace*{-0.85\baselineskip}

\scalebox{0.57}{
	\begin{tikzpicture}
		[every node/.style={ultra thick,draw=red,minimum size=6mm},
			node distance=1cm]
		\normalsize
		\node[and gate US,point up,logic gate inputs=nnn, seq=4pt] (N)
		{\rotatebox{-90}{$N$}};
		\node[and gate US,point up,logic gate inputs=nnn,
			left of=N,node distance=2cm] (B)
		{\rotatebox{-90}{$B$}};
		\node[state,below of=B,node distance=1.5cm] (b2) {$b_2$};
		\node[state,left of=b2] (b1) {$b_1$};
		\node[state,right of=b2] (b3) {$b_3$};
		\draw (b1.north) -- ([yshift=1.5mm]b1.north) -| (B.input 1);
		\draw (b2.north) -- (B.input 2);
		\draw (b3.north) -- ([yshift=1.5mm]b3.north) -| (B.input 3);
		\node[and gate US,point up,logic gate inputs=nnn,
			above of=B, node distance=3.5cm] (A)
		{\rotatebox{-90}{$A$}};
		\node[state,below of=A,node distance=1.5cm] (a2) {$a_2$};
		\node[state,left of=a2] (a1) {$a_1$};
		\node[state,right of=a2] (a3) {$a_3$};
		\draw (a1.north) -- ([yshift=1.5mm]a1.north) -| (A.input 1);
		\draw (a2.north) -- (A.input 2);
		\draw (a3.north) -- ([yshift=1.5mm]a3.north) -| (A.input 3);
		\node[and gate US,point up,logic gate inputs=nnn,
			below of=B, node distance=3.5cm] (C)
		{\rotatebox{-90}{$C$}};
		\node[state,below of=C,node distance=1.5cm] (c2) {$c_2$};
		\node[state,left of=c2] (c1) {$c_1$};
		\node[state,right of=c2] (c3) {$c_3$};
		\draw (c1.north) -- ([yshift=1.5mm]c1.north) -| (C.input 1);
		\draw (c2.north) -- (C.input 2);
		\draw (c3.north) -- ([yshift=1.5mm]c3.north) -| (C.input 3);
		\draw (A.east) -- ([yshift=1.5mm]A.east) -| (N.input 1);
		\draw (B.east) -- (N.input 2);
		\draw (C.east) -- ([yshift=1.5mm]C.east) -| (N.input 3);

		\node[and gate US,point up,logic gate inputs=nnn,
			below of=B, node distance=7cm] (A')
		{\rotatebox{-90}{$A$}};
		\node[state,below of=A',node distance=1.5cm] (a2') {$a_2$};
		\node[state,left of=a2'] (a1') {$a_1$};
		\node[state,right of=a2'] (a3') {$a_3$};
		\draw (a1'.north) -- ([yshift=1.5mm]a1'.north) -| (A'.input 1);
		\draw (a2'.north) -- (A'.input 2);
		\draw (a3'.north) -- ([yshift=1.5mm]a3'.north) -| (A'.input 3);
		\node[gateNULL,draw,above of=A',node distance=1.5cm] (N1) {$N_1$};
		\draw (N1) -- (A');
		\node[state,above of=N1,node distance=1.5cm] (b2') {$b_2$};
		\node[and gate US,point up,logic gate inputs=nnn,
			right of=b2',node distance=1.5cm] (B')
		{\rotatebox{-90}{$B$}};
		\node[state,left of=b2'] (b1') {$b_1$};
		\node[state,right of=b2'] (b3') {$b_3$};
		\draw (b1'.north) -- ([yshift=1.5mm]b1'.north) -| (B'.input 1);
		\draw (b2'.north) -- (B'.input 2);
		\draw (b3'.north) -- ([yshift=1.5mm]b3'.north) -| (B'.input 3);
		\draw (N1) -- (b1');
		\draw (N1) -- (b2');
		\draw (N1) -- (b3');
		\node[gateNULL,draw,above of=B',node distance=1.5cm] (N2) {$N_2$};
		\draw (N2) -- (B');
		\node[state,above of=N2,node distance=1.5cm] (c2') {$c_2$};
		\node[and gate US,point up,logic gate inputs=nnn,
			right of=c2',node distance=1.5cm] (C')
		{\rotatebox{-90}{$C$}};
		\node[state,left of=c2'] (c1') {$c_1$};
		\node[state,right of=c2'] (c3') {$c_3$};
		\draw (c1'.north) -- ([yshift=1.5mm]c1'.north) -| (C'.input 1);
		\draw (c2'.north) -- (C'.input 2);
		\draw (c3'.north) -- ([yshift=1.5mm]c3'.north) -| (C'.input 3);
		\draw (N2) -- (c1');
		\draw (N2) -- (c2');
		\draw (N2) -- (c3');
		\node[gateNULL,above of= C',node distance=1.5cm] (N3) {$N_3$};
		\draw (C') -- (N3);
		\node [above of= N,xshift=-1cm,yshift=4.5cm,text width=5.8cm,align=justify,scale=1/0.57,draw=none]
		{
			{
					\begin{example}
						Let us consider the \gateSAND{} node $N$ depicted on the left-hand
						side of \Cref{fig:pre:sand:ex},
						where all other nodes have already been processed by the time
						normalisation.
						The transformation of \Cref{prop:pre:sand} leads to the \DAG/ on the right-hand side
						of \Cref{fig:pre:sand:ex} where subtrees $A$, $B$ and $C$ occur in sequence, as imposed
						by the \gateNULL{} nodes $N_1$ and $N_2$ between them, and then action $N_3$ occurs.
					\end{example}
				}};
	\end{tikzpicture}
}
\endgroup
\caption{Normalising \gateSAND node $N$\label{fig:pre:sand:ex}}
\end{figure}
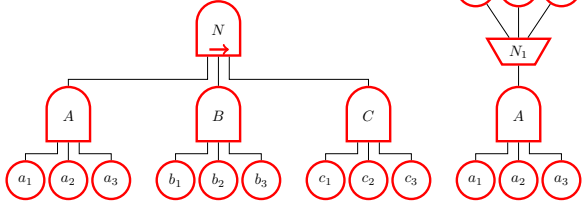
\endgroup


\subsection{Handling defences}
\label{sec:pre:def}
The scheduling we are seeking to obtain will guarantee the necessary attacks are performed.
Hence, when dealing with defence nodes, we can assume that all attacks are successful.
However, they may not be mandatory, in which case they should be avoided so as to
obtain a better scheduling of agents.

Taking into account each possible choice of defences will lead to as many \DAG/s representing
the attacks to be performed. This allows for answering the question: ``What is the
minimal schedule of attacker agents if these defences are operating?''



\emph{Composite defences.} Defences resulting from an \gateAND{}, \gateSAND
or \gateOR
between several defences are operating according to the success of their subtrees:
for \gateAND{} and \gateSAND{}, all subtrees should be operating, while only one is
necessary for \gateOR{}. This can easily be computed by a boolean bottom-up labelling
of nodes. Note that different choices of elementary defences can lead to disabling
the same higher-level composite defence, thus limiting the number of \DAG/s that will
need to be considered for the scheduling.

\emph{No Defence nodes (\gateNODEF).} 
A \gateNODEF succeeds if its attack succeeds or its defence fails. 
Hence, if the defence is not operating, the attack is not necessary.
Thus, the \gateNODEF{} node can be replaced by a \gateNULL{} node without children,
and the children subtrees 
deleted.
On the contrary, if the defence is operating, the attack must take place. The defence
subtree is deleted, while the attack one is kept, and the \gateNODEF{} node can be
replaced by a \gateNULL{} node, as 
pictured in \Cref{fig:pre:nodef}.

\emph{Counter Defence (\gateCAND{}) and Failed Reactive Defence (\gateSCAND
	{}) nodes.} A \gateCAND{} succeeds if its attack is successful and its defence is
not. A \gateSCAND{} behaves similarly but in a sequential fashion, \ie{} the defence
takes place after the attack. In both cases, if the defence is not operating, its
subtree is deleted, while the attack one is kept, and the \gateCAND{} (or \gateSCAND{})
node can be replaced by a \gateNULL{} node, as was the case in \Cref{fig:pre:nodef:dok}.
Otherwise, the \gateCAND{} (or \gateSCAND{}) node is deleted, as well as its subtrees.
Moreover, it transmits its failure recursively to its parents, until a choice of another
branch is possible. Thus, all ancestors are deleted bottom up until an \gateOR is
reached.


Thus, we have a set of \DAG/s with attack nodes only.

\begin{figure}[!!htb]
\centering
\begingroup
\def\nodesubtree{\node[isosceles triangle, isosceles triangle apex angle=70, shape border rotate=90, minimum size=6.5mm]}

\begin{subfigure}[b]{0.3\linewidth}
	\centering
	\scalebox{0.75}{
		\begin{tikzpicture}
			[every node/.style={ultra thick,draw=red,minimum size=6mm},
				node distance=1cm]
			\normalsize
			\node[or gate US,point up,logic gate inputs=ni] (A)
			{\rotatebox{-90}{$A$}};
			\nodesubtree             at (-1,-1.28) (a) {$a$};
			\nodesubtree[draw=Green] at ( 1,-1.24) (d) {$d$};
			\draw (a.north) -- ([yshift=1.4mm]a.north) -| (A.input 1);
			\draw (d.north) -- ([yshift=0.6mm]d.north) -| (A.input 2);
		\end{tikzpicture}
	}
	\subcaption{\gateNODEF{} node\label{fig:pre:nodef:gate}}
\end{subfigure}
\begin{subfigure}[b]{0.3\linewidth}
	\centering
	\scalebox{0.75}{
		\begin{tikzpicture}
			[every node/.style={ultra thick,draw=red,minimum size=6mm},
				node distance=1cm]
			\normalsize
			\node[draw=none] at (0,0) () {};
			\node[gateNULL] at (0,1.28) (A') {$A'$};
		\end{tikzpicture}
	}
	\subcaption{Case $d$ fails\label{fig:pre:nodef:dnok}}
\end{subfigure}
\begin{subfigure}[b]{0.3\linewidth}
	\centering
	\scalebox{0.75}{
		\begin{tikzpicture}
			[every node/.style={ultra thick,draw=red,minimum size=6mm},
				node distance=1cm]
			\normalsize
			\node[gateNULL] (A') {$A'$};
			\nodesubtree             at (0,-1.28) (a) {$a$};
			\draw (a.north) -- (A');
		\end{tikzpicture}
	}
	\subcaption{Case $d$ operates\label{fig:pre:nodef:dok}}
\end{subfigure}
\endgroup
\caption{Handling \gateNODEF $A$\label{fig:pre:nodef}}
\end{figure}
\subsection{Handling OR branches}
\label{sec:pre:or}
\gateOR nodes give the choice between several series of actions,
only one of which will be chosen in an optimal assignment of events.
However, one cannot simply keep the shortest branch of an \gateOR node and prune all others.
Doing so minimises attack time, but not necessarily the number of agents.
In particular, a slightly longer, but narrower branch may require fewer agents without increasing attack time,
provided there is a longer sequence elsewhere in the \DAG/.
Consequently, only branches that are guaranteed not to lead to an optimal assignment can be pruned,
which is the case when a branch is the longest one in the entire graph.
All other cases need to be investigated, leading to multiple variants depending on
the \gateOR branch executed, similar to the approach for defence nodes.

\subsection{Preprocessing the treasure hunters \ADT/}
\label{sec:pre:ex}

\Cref{fig:pre:treasure:1,fig:pre:treasure:2} detail the preprocessing of the treasure
hunters example step by step. The time unit is one minute. Long sequences of \gateSEQ
{} are shortened with dotted lines.
Note that when handling the defence, at step 3, we should obtain two \DAG/s corresponding
to the case where the defence fails (see \Cref{fig:pre:treasure:def}), or where the
defence is successful. This latter case leads to an empty \DAG/ where no attack can
succeed. Therefore, we can immediately conclude that if the police is successful,
there is no scheduling of agents.

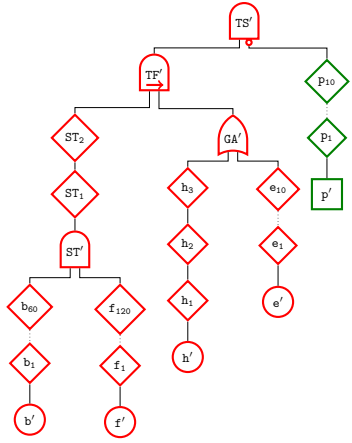
\begin{figure}[!!htb]
	\centering
	\scalebox{0.5}{
		\begin{tikzpicture}[every node/.style={ultra thick,draw=red,minimum size=6mm},
				node distance=1.5cm]
			\node[and gate US,point up,logic gate inputs=ni] (ts)
			{\rotatebox{-90}{\gate{TS'}}};
			\node[gateSEQ,draw=Green,minimum size=8mm,
				below = 5mm of ts.west, xshift=22mm]
			(p10) {\leaf{p_{10}}};
			\node[gateSEQ,draw=Green,minimum size=8mm,below of= p10]
			(p1) {\leaf{p_{1}}};
			\node[rectangle,draw=Green,minimum size=8mm,below of= p1]
			(p) {\leaf{p'}};
			\draw (p10.north) -- ([yshift=0.28cm]p10.north) -| (ts.input 2);
			\draw[dotted] (p10) --(p1);
			\draw (p) -- (p1);
			\node[and gate US,point up,logic gate inputs=nn, seq=4pt,
				below = 9mm of ts.west, yshift=28mm] (tf)
			{\rotatebox{-90}{\gate{TF'}}};
			\draw (tf.east) -- ([yshift=0.15cm]tf.east) -| (ts.input 1);
			\node[gateSEQ,minimum size=8mm,
				left of = p1, node distance=6.7cm]
			(st2) {\leaf{ST_{2}}};
			\node[gateSEQ,minimum size=8mm,below of= st2]
			(st1) {\leaf{ST_{1}}};
			\node[and gate US,point up,logic gate inputs=nn,
				left of= st1]
			(st) {\rotatebox{-90}{\gate{ST'}}};
			\draw (st2.north) -- ([yshift=0.15cm]st2.north) -| (tf.input 1);
			\draw (st2) -- (st1) -- (st.east);
			\node[gateSEQ, below = 4mm of st.west, xshift=-1.2cm]
			(b60) {\leaf{b_{60}}};
			\node[gateSEQ,minimum size=8mm,below of= b60]
			(b1) {\leaf{b_{1}}};
			\node[state,minimum size=8mm,below of= b1]
			(b) {\leaf{b'}};
			\draw (b60.north) -- ([yshift=0.15cm]b60.north) -| (st.input 1);
			\draw[dotted] (b60) --(b1);
			\draw (b) -- (b1);
			\node[gateSEQ, below = 4mm of st.west, xshift=1.2cm]
			(f120) {\leaf{f_{120}}};
			\node[gateSEQ,minimum size=8mm,below of= f120]
			(f1) {\leaf{f_{1}}};
			\node[state,minimum size=8mm,below of= f1]
			(f) {\leaf{f'}};
			\draw (f120.north) -- ([yshift=0.15cm]f120.north) -| (st.input 2);
			\draw[dotted] (f120) --(f1);
			\draw (f) -- (f1);
			\node[or gate US,point up,logic gate inputs=nn,
				above of = p1, node distance=2.5cm]
			(ga) {\rotatebox{-90}{\gate{GA'}}};
			\draw (ga.east) -- ([yshift=0.15cm]ga.east) -| (tf.input 2);
			\node[gateSEQ, below = 4mm of ga.west, xshift=-1.2cm]
			(h3) {\leaf{h_{3}}};
			\node[gateSEQ,minimum size=8mm,below of= h3]
			(h2) {\leaf{h_{2}}};
			\node[gateSEQ,minimum size=8mm,below of= h2]
			(h1) {\leaf{h_{1}}};
			\node[state,minimum size=8mm,below of= h1]
			(h) {\leaf{h'}};
			\draw (h3.north) -- ([yshift=0.15cm]h3.north) -| (ga.input 1);
			\draw (h3) -- (h2) -- (h1) -- (h);
			\node[gateSEQ, below = 4mm of ga.west, xshift=1.2cm]
			(e10) {\leaf{e_{10}}};
			\node[gateSEQ,minimum size=8mm,below of= e10]
			(e1) {\leaf{e_{1}}};
			\node[state,minimum size=8mm,below of= e1]
			(e) {\leaf{e'}};
			\draw (e10.north) -- ([yshift=0.15cm]e10.north) -| (ga.input 2);
			\draw[dotted] (e10) --(e1);
			\draw (e) -- (e1);
		\end{tikzpicture}
	}
	\caption{Treasure hunters \ADT/: time normalisation\label{fig:pre:treasure:1}}
\end{figure}

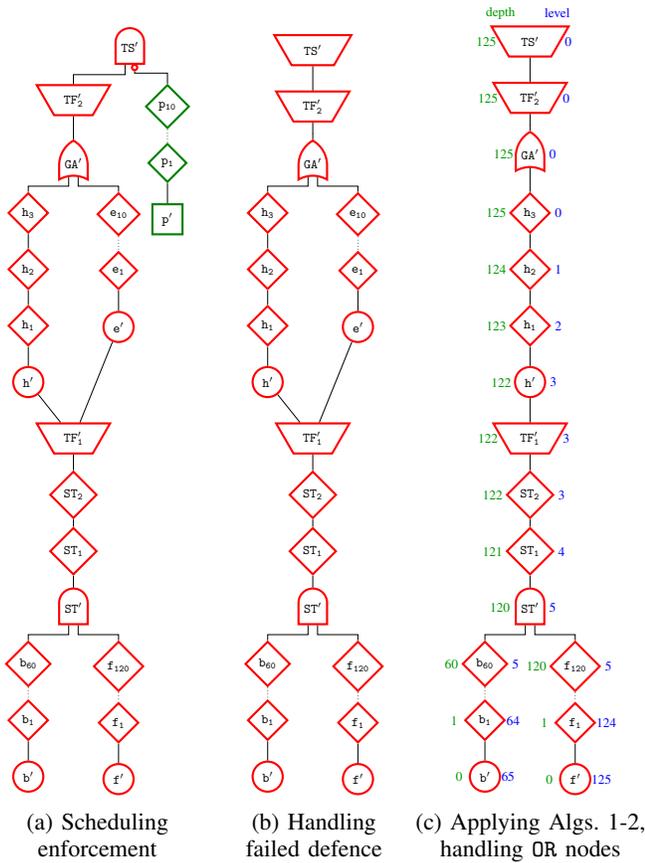
\begin{figure}[!!htb]
	\captionsetup[subfigure]{justification=centering}
	\centering
	\begin{subfigure}[b]{0.35\linewidth}
		\centering
		\scalebox{0.5}{
			\begin{tikzpicture}[every node/.style={ultra thick,draw=red,minimum size=6mm},
					node distance=1.5cm]
				\node[and gate US,point up,logic gate inputs=ni] (ts)
				{\rotatebox{-90}{\gate{TS'}}};
				\node[gateSEQ,draw=Green,minimum size=8mm,
					below = 5mm of ts.west, xshift=10mm]
				(p10) {\leaf{p_{10}}};
				\node[gateSEQ,draw=Green,minimum size=8mm,below of= p10]
				(p1) {\leaf{p_{1}}};
				\node[rectangle,draw=Green,minimum size=8mm,below of= p1]
				(p) {\leaf{p'}};
				\draw (p10.north) -- ([yshift=0.28cm]p10.north) -| (ts.input 2);
				\draw[dotted] (p10) --(p1);
				\draw (p) -- (p1);
				\node[gateNULL,minimum size=8mm,yshift=1.7mm,
					left of= p10, node distance=2.5cm] (tf2)
				{\gate{TF'_2}};
				\draw (tf2.north) -- ([yshift=0.25cm]tf2.north) -| (ts.input 1);
				\node[or gate US,point up,logic gate inputs=nn,
					above of = p1, node distance=2.5cm]
				(ga) {\rotatebox{-90}{\gate{GA'}}};
				\draw (ga.east) -- ([yshift=0.15cm]ga.east) -| (tf2);
				\node[gateSEQ, below = 4mm of ga.west, xshift=-1.2cm]
				(h3) {\leaf{h_{3}}};
				\node[gateSEQ,minimum size=8mm,below of= h3]
				(h2) {\leaf{h_{2}}};
				\node[gateSEQ,minimum size=8mm,below of= h2]
				(h1) {\leaf{h_{1}}};
				\node[state,minimum size=8mm,below of= h1]
				(h) {\leaf{h'}};
				\draw (h3.north) -- ([yshift=0.15cm]h3.north) -| (ga.input 1);
				\draw (h3) -- (h2) -- (h1) -- (h);
				\node[gateSEQ, below = 4mm of ga.west, xshift=1.2cm]
				(e10) {\leaf{e_{10}}};
				\node[gateSEQ,minimum size=8mm,below of= e10]
				(e1) {\leaf{e_{1}}};
				\node[state,minimum size=8mm,below of= e1]
				(e) {\leaf{e'}};
				\draw (e10.north) -- ([yshift=0.15cm]e10.north) -| (ga.input 2);
				\draw[dotted] (e10) --(e1);
				\draw (e) -- (e1);
				\node[gateNULL,minimum size=8mm,xshift=1.2cm,
					below of= h]
				(tf1) {\gate{TF'_1}};
				\draw (tf1) -- (h);
				\draw (tf1) -- (e);
				\node[gateSEQ,minimum size=8mm,below of = tf1]
				(st2) {\leaf{ST_{2}}};
				\node[gateSEQ,minimum size=8mm,below of= st2]
				(st1) {\leaf{ST_{1}}};
				\node[and gate US,point up,logic gate inputs=nn,
					left of= st1]
				(st) {\rotatebox{-90}{\gate{ST'}}};
				\draw (st2) -- (tf1);
				\draw (st2) -- (st1) -- (st.east);
				\node[gateSEQ, below = 4mm of st.west, xshift=-1.2cm]
				(b60) {\leaf{b_{60}}};
				\node[gateSEQ,minimum size=8mm,below of= b60]
				(b1) {\leaf{b_{1}}};
				\node[state,minimum size=8mm,below of= b1]
				(b) {\leaf{b'}};
				\draw (b60.north) -- ([yshift=0.15cm]b60.north) -| (st.input 1);
				\draw[dotted] (b60) --(b1);
				\draw (b) -- (b1);
				\node[gateSEQ, below = 4mm of st.west, xshift=1.2cm]
				(f120) {\leaf{f_{120}}};
				\node[gateSEQ,minimum size=8mm,below of= f120]
				(f1) {\leaf{f_{1}}};
				\node[state,minimum size=8mm,below of= f1]
				(f) {\leaf{f'}};
				\draw (f120.north) -- ([yshift=0.15cm]f120.north) -| (st.input 2);
				\draw[dotted] (f120) --(f1);
				\draw (f) -- (f1);
			\end{tikzpicture}
		}
		\subcaption{Scheduling\\enforcement}
		\label{fig:pre:treasure:sched}
	\end{subfigure}
	\hfill
	\begin{subfigure}[b]{0.25\linewidth}
		\centering
		\scalebox{0.5}{
			\begin{tikzpicture}[every node/.style={ultra thick,draw=red,minimum size=6mm},
					node distance=1.5cm]
				\node[gateNULL,minimum size=8mm]
				(ts) {\gate{TS'}};
				\node[gateNULL,minimum size=8mm,below of= ts]
				(tf2) {\gate{TF'_2}};
				\draw (tf2) -- (ts);
				\node[or gate US,point up,logic gate inputs=nn,
					left of = tf2]
				(ga) {\rotatebox{-90}{\gate{GA'}}};
				\draw (ga.east) -- ([yshift=0.15cm]ga.east) -| (tf2);
				\node[gateSEQ, below = 4mm of ga.west, xshift=-1.2cm]
				(h3) {\leaf{h_{3}}};
				\node[gateSEQ,minimum size=8mm,below of= h3]
				(h2) {\leaf{h_{2}}};
				\node[gateSEQ,minimum size=8mm,below of= h2]
				(h1) {\leaf{h_{1}}};
				\node[state,minimum size=8mm,below of= h1]
				(h) {\leaf{h'}};
				\draw (h3.north) -- ([yshift=0.15cm]h3.north) -| (ga.input 1);
				\draw (h3) -- (h2) -- (h1) -- (h);
				\node[gateSEQ, below = 4mm of ga.west, xshift=1.2cm]
				(e10) {\leaf{e_{10}}};
				\node[gateSEQ,minimum size=8mm,below of= e10]
				(e1) {\leaf{e_{1}}};
				\node[state,minimum size=8mm,below of= e1]
				(e) {\leaf{e'}};
				\draw (e10.north) -- ([yshift=0.15cm]e10.north) -| (ga.input 2);
				\draw[dotted] (e10) --(e1);
				\draw (e) -- (e1);
				\node[gateNULL,minimum size=8mm,xshift=1.2cm,
					below of= h]
				(tf1) {\gate{TF'_1}};
				\draw (tf1) -- (h);
				\draw (tf1) -- (e);
				\node[gateSEQ,minimum size=8mm,below of = tf1]
				(st2) {\leaf{ST_{2}}};
				\node[gateSEQ,minimum size=8mm,below of= st2]
				(st1) {\leaf{ST_{1}}};
				\node[and gate US,point up,logic gate inputs=nn,
					left of= st1]
				(st) {\rotatebox{-90}{\gate{ST'}}};
				\draw (st2) -- (tf1);
				\draw (st2) -- (st1) -- (st.east);
				\node[gateSEQ, below = 4mm of st.west, xshift=-1.2cm]
				(b60) {\leaf{b_{60}}};
				\node[gateSEQ,minimum size=8mm,below of= b60]
				(b1) {\leaf{b_{1}}};
				\node[state,minimum size=8mm,below of= b1]
				(b) {\leaf{b'}};
				\draw (b60.north) -- ([yshift=0.15cm]b60.north) -| (st.input 1);
				\draw[dotted] (b60) --(b1);
				\draw (b) -- (b1);
				\node[gateSEQ, below = 4mm of st.west, xshift=1.2cm]
				(f120) {\leaf{f_{120}}};
				\node[gateSEQ,minimum size=8mm,below of= f120]
				(f1) {\leaf{f_{1}}};
				\node[state,minimum size=8mm,below of= f1]
				(f) {\leaf{f'}};
				\draw (f120.north) -- ([yshift=0.15cm]f120.north) -| (st.input 2);
				\draw[dotted] (f120) --(f1);
				\draw (f) -- (f1);
			\end{tikzpicture}
		}
		\subcaption{Handling\\failed defence}
		\label{fig:pre:treasure:def}
	\end{subfigure}
	\hfill
		\begin{subfigure}[b]{0.35\linewidth}
		\centering
		\scalebox{0.5}{
			\begin{tikzpicture}[every node/.style={ultra thick,
							draw=red,minimum size=6mm},
					node distance=1.5cm]
				\node[gateNULL,minimum size=8mm,
					depth label=125,
					label={[draw=none,green!60!black]120:depth},
					level label=0,
					label={[draw=none,blue]60:level}]
				(ts) {\gate{TS'}};
				\node[gateNULL,minimum size=8mm,below of= ts,
					depth label=125,level label=0]
				(tf2) {\gate{TF'_2}};
				\draw (tf2) -- (ts);
				\node[or gate US,point up,logic gate inputs=nn,
					left of = tf2]
				(ga) {\rotatebox{-90}{\gate{GA'}}};
				\node[draw=none,green!60!black,xshift=-0.7cm]
				at (ga.center) (){125};
				\node[draw=none,blue,xshift=0.6cm]
				at (ga.center) (){0};
				\draw (ga.east) -- ([yshift=0.15cm]ga.east) -| (tf2);
				\node[gateSEQ, below = of ga.center,yshift=0.5cm,
					depth label=125,level label=0]
				(h3) {\leaf{h_{3}}};
				\node[gateSEQ,minimum size=8mm,below of= h3,
					depth label=124,level label=1]
				(h2) {\leaf{h_{2}}};
				\node[gateSEQ,minimum size=8mm,below of= h2,
					depth label=123,level label=2]
				(h1) {\leaf{h_{1}}};
				\node[state,minimum size=8mm,below of= h1,
					depth label=122,level label=3]
				(h) {\leaf{h'}};
				\draw (h3.north) -- (ga);
				\draw (h3) -- (h2) -- (h1) -- (h);
				\node[gateNULL,minimum size=8mm,below of= h,
					depth label=122,level label=3]
				(tf1) {\gate{TF'_1}};
				\draw (tf1) -- (h);
				\node[gateSEQ,minimum size=8mm,below of = tf1,
					depth label=122,level label=3]
				(st2) {\leaf{ST_{2}}};
				\node[gateSEQ,minimum size=8mm,below of= st2,
					depth label=121,level label=4]
				(st1) {\leaf{ST_{1}}};
				\node[and gate US,point up,logic gate inputs=nn,
					left of= st1]
				(st) {\rotatebox{-90}{\gate{ST'}}};
				\node[draw=none,green!60!black,xshift=-0.8cm]
				at (st.center) (){120};
				\node[draw=none,blue,xshift=0.6cm]
				at (st.center) (){5};
				\draw (st2) -- (tf1);
				\draw (st2) -- (st1) -- (st.east);
				\node[gateSEQ, below = 4mm of st.west,xshift=-1.2cm,
					depth label=60,level label=5]
				(b60) {\leaf{b_{60}}};
				\node[gateSEQ,minimum size=8mm,below of= b60,
					depth label=1,level label=64]
				(b1) {\leaf{b_{1}}};
				\node[state,minimum size=8mm,below of= b1,
					depth label=0,level label=65]
				(b) {\leaf{b'}};
				\draw (b60.north) -- ([yshift=0.15cm]b60.north) -|
				(st.input 1);
				\draw[dotted] (b60) --(b1);
				\draw (b) -- (b1);
				\node[gateSEQ, below = 4mm of st.west,xshift=1.2cm,
					depth label=120,level label=5]
				(f120) {\leaf{f_{120}}};
				\node[gateSEQ,minimum size=8mm,below of= f120,
					depth label=1,level label=124]
				(f1) {\leaf{f_{1}}};
				\node[state,minimum size=8mm,below of= f1,
					depth label=0,level label=125]
				(f) {\leaf{f'}};
				\draw (f120.north) -- ([yshift=0.15cm]f120.north) -|
				(st.input 2);
				\draw[dotted] (f120) --(f1);
				\draw (f) -- (f1);
			\end{tikzpicture}
		}
		\subcaption{Applying Algs.~\ref{algo:depth}-\ref{algo:level}, handling \gateOR nodes}%
		\label{fig:sched:pruneNlabel}
	\end{subfigure}

	\caption{Treasure hunters \ADT/: final preprocessing steps (left, middle) and initial part of the main algorithm (right)\label{fig:pre:treasure:2}}
\end{figure}

%% file: algo.tex






At this stage, we have \DAG/s where nodes are either (i) a leaf, or of type \gateAND
{}, \gateOR
{}, or \gateNULL, all with duration $0$ or (ii) of type \gateSEQ{} with duration
$\tunit{}$. Their branches mimic the possible runs in the system.

The algorithm's input is a set of \DAG/s preprocessed as described in \Cref{sec:preprocess},
corresponding to possible configurations of defence nodes' outcomes and choices of \gateOR branches in the original \ADT/.
%
For each of these \DAG/s, $n$ denotes the number of \gateSEQ{} nodes (all other ones have 0-duration).
Furthermore, nodes (denoted by \N) have some attributes: their $\type$; four integers
$\depth$, $\level$, $\agent$ and $\slot$, initially with value 0.
The values of $\depth$ and $\level$ denote, respectively,
the height of a node's tallest subtree and the distance from the root
(both without counting the zero duration nodes),
while $\agent$ and $\slot$ store a node's assignment in the schedule.

We first compute the nodes' depth in \Cref{sec:algo:depth}, then
compute the level of nodes in \Cref{sec:algo:level},
and finally compute an optimal scheduling in \Cref{sec:algo:algo}.

\subsection{Depth of nodes}
\label{sec:algo:depth}

Starting from the root, the procedure $\textsc{DepthNode}$ (\Cref{algo:depth})
explores the \DAG/ in a DFS (\emph{depth first search}) manner.
During backtracking, \ie starting from the leaves, $\depth$ is computed for the different types of nodes as follows:

{\bf\gateLEAF node:} After the time normalisation, a leaf node takes $0$ time.
	      It may still be an actual leaf,
	      and
	      its total duration is also 0 since
	      it has no children (not satisfying condition at l.~\ref{line:depth:leaf}). Or,
	      it may have a child
	      due to
	      scheduling enforcement,
	      and then its time is the same as the one of its only child (l.~\ref{line:depth:null}).

{\bf\gateSEQ node:} Its duration is one $\tunit{}$, which must be added to the duration
	      of its only child to obtain the minimum time of execution from the start
	      (l.~\ref{line:depth:seq}--\ref{line:depth:seq:op}).

{\bf\gateAND node:} All children 
		  must be completed before it
	      occurs. Therefore, its minimal time is the maximum one of all its children
	      (l.~\ref{line:depth:and}--\ref{line:depth:and:op}).

{\bf\gateOR node:} 
	      One child must complete for the \gateOR{} node to happen.
	      Its time is thus the minimal one of all its children (l.~\ref{line:depth:or}--
	      \ref{line:depth:or:op}).

{\bf\gateNULL node:} Note that, by construction, a \gateNULL{} node may have several
	      parents but a single child. Its duration being null, its time is the same as
	      the one of its only child (l.~\ref{line:depth:null}).

Note that the condition at l.~\ref{line:depth:continue} avoids a second exploration
of a node which has already been handled in another branch.


\begin{algorithm}[htb]
	\caption{\textsc{DepthNode}($\node$)}
	\label{algo:depth}
	\For{$\N \in \children(\node)$}{
		\lIf{$\N.depth = 0$\label{line:depth:continue}}{
			$\textsc{DepthNode}(\N)$}
	}
	\If{$\children(\node) \neq \emptyset$\label{line:depth:leaf}}{
		\uIf{$\node.\type=\gateSEQ$\label{line:depth:seq}}{
			$\node.\depth \gets \N.\depth+1$, \textbf{s.t.} $\{\N\} = \children(\node)$\label{line:depth:seq:op}}
		\uElseIf{$\node.\type=\gateAND$\label{line:depth:and}}{
			$\node.\depth \gets \max(\{\N.\depth \ |\ \N \in \children(\node)\})$\label{line:depth:and:op}}
		\uElseIf{$\node.\type=\gateOR$\label{line:depth:or}}{
			$\node.\depth \gets \min(\{\N.\depth \ |\ \N \in \children(\node)\})$\label{line:depth:or:op}}
		\lElse{
			$\node.\depth \gets \N.\depth$, \textbf{s.t.} $\{\N\} = \children(\node)$\label{line:depth:null}}
	}
\end{algorithm}

%
%
%
%
%
%

\subsection{Level of nodes}
\label{sec:algo:level}

Levels are assigned recursively, starting with the root, using a DFS.
The procedure $\textsc{LevelNode}$ (\Cref{algo:level}) computes nodes' levels.
It first assigns the node's level (l.~\ref{line:level:init}) according to the call
argument.
Note that in case of multiple parents (or ancestors with multiple parents), the longest path to the
root is kept.

\begin{algorithm}[htb]
	\caption{\textsc{LevelNode}($\node, l$)}
	\label{algo:level}
	$\node.\level \gets \max(l,\node.\level)$\label{line:level:init}

	\For{$\N \in \children(\node)$}{
		\lIf{$\node.\type=\gateSEQ$}{
			$\textsc{LevelNode}(\N, l+1)$}
		\lElse{
			$\textsc{LevelNode}(\N, l)$}
	}
\end{algorithm}

\subsection{Number of agents: upper and lower bounds}
\label{sec:algo:bounds}

The upper bound on the number of agents is obtained from the maximal width of the preprocessed \DAG/,
\ie the maximal number of \gateSEQ nodes assigned the same value of \textit{level}
(that must be executed in parallel to ensure minimum time).


The minimal attack time is obtained from the number of levels $l$ in the preprocessed \DAG/.
Note that the longest path from the root to a leaf has exactly $l$ nodes of non-zero duration.
Clearly, none of these nodes can be executed in parallel,
therefore the number of time slots cannot be smaller than $l$.
%
Thus, if an optimal schedule of $l\times\tunit{}$ is realisable,
the $n$ nodes must fit in a schedule containing $l$ time slots.
Hence, the lower bound on the number of agents is  $\lceil\frac{n}{l}\rceil$.
There is, however, no guarantee that it can be achieved, and introducing additional agents may be necessary
depending on the \DAG/ structure, \eg if there are many parallel leaves.

\subsection{Minimal schedule}
\label{sec:algo:algo}

%

The algorithm for obtaining a schedule with the minimal attack time and also minimising
the number of agents is given in \Cref{algo:minsched}.
Input \DAG/s are processed sequentially, a schedule returned for each one.
Not restricting the output to the overall minimum allows to avoid ``no attack'' scenarios
where the time is 0 (\eg following a defence failure on a root \gateNODEF node).
Furthermore, with information on the repartition of agents for a successful minimal time attack in all cases of defences,
the defender is able to decide which defences to enable according to these results
(and maybe the costs of defences).

\begin{algorithm}[t]
	\caption{\textsc{MinSchedule}($\DAGset$)}
	\label{algo:minsched}
	$\varOutput = \emptyset$

	\While{$\DAGset \neq \emptyset$}{
		Pick $\inputDAG \in \DAGset$

		\lIf(\Comment*[f]{Skip empty \DAG/s}){$\inputDAG.n=0$}{
		\textbf{continue}}
		$\textsc{DepthNode}(root(\inputDAG))$ \Comment*{Compute depth of nodes}
		$\inputDAG \gets \inputDAG\setminus\{\N \ |\ \neg
			\N.\keep\}$\label{line:minsched:delete}

		$\textsc{LevelNode}(root(\inputDAG),0)$ \Comment*{Compute level of nodes}

		$\varSlots \gets root(\inputDAG).\depth$\label{line:minsched:numslots}

	  $\varBound \gets \lceil\frac{\inputDAG.n}{\varSlots}\rceil - 1$

		$\varMaxAgents \gets \max_j(|\{\N: \N.\type=\gateSEQ \land \N.\level=j\}|)$ \Comment*{Max. level width (concur. \gateSEQ nodes)}

		$\varUpperBound \gets \varMaxAgents$

		$\varCurrentOutput = \emptyset$

		\While{$(\varUpperBound-\varBound>1)$\label{line:minsched:newloop}}{

			$\varNumAgents \gets \varBound + \lfloor\frac{\varUpperBound-\varBound}{2}\rfloor$\label{line:minsched:adjustagents}

			$(\varCandidate,\varNodesLeft) \gets \textsc{Schedule}(\inputDAG,\varSlots,\varNumAgents)$\label{line:minsched:candidate}

		  \If(\Comment*[f]{Candidate schedule OK}){$\varNodesLeft=0$\label{line:minsched:discard}}{

				$\varUpperBound \gets \varNumAgents$

				$\varCurrentOutput \gets \varCandidate$}

			\lElse(\Comment*[f]{Cand. schedule not OK}){$\varBound=\varNumAgents$}

		} 

		\If{$\varUpperBound = \varMaxAgents$}{
			$ (\varCurrentOutput, \_) \gets \textsc{Schedule}(\inputDAG,\varSlots,\varMaxAgents)$
		}\label{line:minsched:minagents}

		$\textsc{ZeroAssign}(\inputDAG)$\label{line:minsched:zero}

		$\varOutput \gets \varOutput \cup \varCurrentOutput$

		$\DAGset \gets \DAGset \setminus \inputDAG$
	}
	\Return $\varOutput$
\end{algorithm}


The actual computation of the schedule is handled by the function \textsc{Schedule} (\Cref{algo:schedCandidate}).
Starting from the root and going top-down, all \gateSEQ nodes at the current level
are added to set $\workingSet$. The other nodes at that level have a null duration
and can be scheduled afterwards with either a parent or child.
An additional check in l.~\ref{line:sched:discard} ensures that non-optimal variants
(whose longest branch exceeds a previously encountered minimum) are discarded without needlesly computing the schedule.
Nodes in $\workingSet{}$ are assigned an agent and time slot, prioritising those with
higher $\depth$ (\ie taller subtrees), as long as an agent is available.
Assigned nodes are removed from $\workingSet$, and any that remains (\eg when the bound was exceeded)
is carried over to the next level iteration.
Note that at this point it is possible for a parent and child node to be in $\workingSet$ concurrently,
but since higher $\depth$ takes precedence, they will never be scheduled in the wrong order.
In such cases, an extra check in the while loop avoids scheduling both nodes to be executed in parallel.

\begin{algorithm}[t]
	\caption{\textsc{Schedule}($\inputDAG,\varSlots,\varNumAgents$)}\label{algo:schedCandidate}

		$l \gets 0, \varSlot \gets \varSlots,\workingSet \gets \emptyset, \varNodesLeft \gets \inputDAG.n$	\label{line:sched:numslots}

		\While{$\varNodesLeft>0$ \textbf{and} $\varSlot > 0$\label{line:minsched:mainloop}}{
			$\varAgent \gets 1$



			$\workingSet \gets \workingSet \cup \{\N \ |\ \N.\type=\gateSEQ \land \N.\level=l\}$

			\If{$\exists_{\N \in \workingSet}$, \textbf{s.t.} $\N.\depth < \varSlots
			- \varSlot$\label{line:sched:discard}}{
                \Return $\emptyset, \varNodesLeft$
			}

			\While{$\varAgent \leq \varNumAgents$ \textbf{and} $\workingSet \neq \emptyset$ \textbf{and} \qquad\qquad\qquad\qquad
					(Pick $\N\in\workingSet$, \textbf{s.t.}
					$\forall_{\N' \in \workingSet} \N.\depth \geq \N'.\depth\, \land$ \qquad
					$\forall_{\N': \N'.\varSlot = \varSlot} \N' \notin ancestors(\N)) \neq \emptyset$\label{line:minsched:prioritise}}{

				$\N.\agent \gets \varAgent$

				$\N.\slot \gets \varSlot$

				$\varAgent \gets \varAgent+1,\varNodesLeft	\gets \varNodesLeft-1\label{line:minsched:parexec}$

				$\workingSet \gets \workingSet \setminus \{\N\}$
			}

			$\textsc{ReshuffleSlot}(\slot, \varAgent-1)$
			\label{line:minsched:shuffle}

			$l \gets l+1, \varSlot \gets \varSlot-1$
			}

	$output \gets \bigcup_{N\in\inputDAG} \{ (\N.\agent,\N.\slot) \}$

	\Return $output, \varNodesLeft$
\end{algorithm}


\Cref{algo:schedCandidate} calls function $\textsc{ReshuffleSlot}$ after the complete
assignment of a time slot at l.~\ref{line:minsched:shuffle}
to ensure consistent assignment of sub-actions of the same \ADT/ node.
Note that depending on $depth$, a sub-action may be moved to the next slot,
creating an interrupted schedule where an agent stops an action for one or more time units to handle another.
Alternatively, agents may collaborate, each handling a node's action for a part of its total duration.
Such assignments could be deemed unsuitable for specific scenarios, \eg defusing a bomb,
in which case manual reshuffling or adding extra agent(s) is left to the user's discretion.

At this point, either the upper or the lower bound on the number of agents is adjusted,
depending on whether the resulting schedule is valid (that is, there are no nodes left to assign at the end).
Scheduling is then repeated for these updated values until the minimal number of agents is found (\ie the two bounds are equal).

After the complete computation for a given \DAG/, l.~\ref{line:minsched:zero}
calls function $\textsc{ZeroAssign}$ in order to obtain assignments
for all remaining nodes,
\ie those of zero duration. Functions $\textsc{ReshuffleSlot}$ and $\textsc{ZeroAssign}$
are detailed in \Cref{sec:algo:reshuffle,sec:algo:zero}, respectively.


Although this algorithm assumes the minimal time is of interest,
it can be easily modified to increase the number of time slots,
thus synthesising the minimal number of agents required for a successful attack of any given duration.

\subsection{Uniform assignment for SEQ nodes}
\label{sec:algo:reshuffle}

A separate subprocedure, given in \Cref{algo:reshuffle}, swaps assigned agents between nodes at the same level
so that the same agent handles all \gateSEQ nodes in sequences obtained during the time normalisation step
(\ie corresponding to a single node in the original \ADT/).

\begin{algorithm}
	\caption{\textsc{ReshuffleSlot}($slot, num\_agents$)}
	\label{algo:reshuffle}
	\For{$agent \in \{1..num\_agents\}$}{
		$\varCurrentNode \gets \N$, \textbf{s.t.} $\N.\varAgent
		= \varAgent\land \N.\slot=\slot$

		$\varParentAgent \gets \parent(\varCurrentNode).\varAgent$

		\If{$\varParentAgent \neq \varAgent\land\varParentAgent\neq 0$\label{line:reshuffle:seqonly}}{
			\If{$\exists \N' \neq \varCurrentNode$, \textbf{s.t.}
				$\quad\N'.\agent = \varParentAgent\land \N'.\slot=\slot$}{
				$\N'.\agent \gets \varAgent$ \Comment*{Swap with $\N'$ if it exists}

				$\N'.\slot \gets \slot$
			}
			$\varCurrentNode.\agent \gets \varParentAgent$

			$\varCurrentNode.\slot \gets \slot$
		}
	}
\end{algorithm}

\begin{proposition}
	\label{prop:algo:reshuffle}
	Reshuffling the assignment by swapping the agents assigned to a pair of nodes in the same slot
	does not affect the correctness of the scheduling.
\end{proposition}
\begin{proof}
	First, note that the procedure does not affect nodes whose parents
	have not yet been assigned an agent (l. \ref{line:reshuffle:seqonly}).
	Hence, reshuffling only applies to \gateSEQ nodes (since the assignment
	of 0 duration nodes occurs later in the main algorithm \textsc{MinSchedule}).
	Furthermore, 
	changes are restricted to pairs of nodes in the same time slot,
	so swapping assigned agents between them cannot break the 
	execution order
	and does not affect the schedule correctness.
\end{proof}

\subsection{Assigning nodes without duration}
\label{sec:algo:zero}


After all non-zero duration nodes have been assigned and possibly reshuffled at each
level,
\Cref{algo:zero} handles the remaining nodes.

\begin{algorithm}
	\caption{\textsc{ZeroAssign}($\inputDAG$)}
	\label{algo:zero}
		$\workingSet{}\gets\{\N\ |\ \N.\agent=0\}$ \Comment*
		{Nodes not assigned yet}
		\For{$\node{}\in\workingSet{}$\label{line:zero:seq:start}}{
			\If{$\N\in\parent(\node{})\land \N.\type=\gateSEQ{}$}{
				$\node.\agent\gets \N.\agent$

				$\node.\slot\gets \N.\slot$

				$\workingSet{}\gets\workingSet{}\setminus\{\node\}$
			}
		\label{line:zero:seq:end}}
		\While{$\workingSet{}\neq\emptyset$}{
			\For{$\node{}\in\workingSet{}$ \textbf{s.t.} $node.type \in \{\gateNULL,
			\gateOR, \gateLEAF\}$\label{line:zero:nullorleaf:start}}{
				\If{$\N.\agent\neq 0$ \textbf{s.t.} $\N\in\children(\node{})$}{
					$\node.\agent\gets \N.\agent$

					$\node.\slot\gets \N.\slot$

					$\workingSet{}\gets\workingSet{}\setminus\{\node\}$
				}
				\If{$(\children(\node)=\emptyset$\\
						$\lor (\N.\depth=0$ \textbf{s.t.}
						$\N\in\children(\node)))$}{
					$\varParentNode\gets\N\in\parent(\node)$ \textbf{s.t.} \mbox{\quad}$\forall_
					{\N'\in \parent(\node{})} \N.\slot\leq\N'.\slot$

					\If{$\varParentNode.agent\neq 0$}{
						$\node.\agent\gets \varParentNode.\agent$

						$\node.\slot\gets \varParentNode.\slot$

						$\workingSet{}\gets\workingSet{}\setminus\{\node\}$
					}
				}
			\label{line:zero:nullorleaf:end}}
			\For{$\node{}\in\workingSet{}$ \textbf{s.t.} $\node{}.\type=\gateAND$\label{line:zero:and:start}}{
				\If{$\node.\depth=0\land\parent(\node).agent\neq
						0$}{
					$\node.\agent\gets \parent(\node).\agent$

					$\node.\slot\gets \parent(\node).\slot$

					$\workingSet{}\gets\workingSet{}\setminus\{\node\}$
				}
				\If{$\node.\depth\neq 0$\\
					$\land\forall_{\N\in\children(\node)}
				(\N.\agent\neq 0\lor\N.\depth=0)$}{
					$\varChildNode{}\gets\N\in\children(\node)$ \textbf{s.t.}
					\mbox{\quad}$\forall_{\N'\in \children
						(\node{})} \N.\slot\geq\N'.\slot$

					$\node.\agent\gets \varChildNode.\agent$

					$\node.\slot\gets \varChildNode.\slot$

					$\workingSet{}\gets\workingSet{}\setminus\{\node\}$
				}
					\label{line:zero:and:end}}
		}
\end{algorithm}

Our choice here stems from the \ADT/ gate the node originates from.
We first assign zero-duration nodes to the same agent and the same time
slot as their parent if the parent is a \gateSEQ node (l.~\ref{line:zero:seq:start}--\ref{line:zero:seq:end}).

Of the remaining ones, nodes of type \gateNULL, \gateOR and \gateLEAF get the same
assignment as their only child if any, or as their parent if they have no child
(l.~\ref{line:zero:nullorleaf:start}--\ref{line:zero:nullorleaf:end}). The latter case
may happen for \gateNULL{} when handling defences as in \eg{}
\Cref{fig:pre:nodef:dnok}, and for \gateLEAF{} nodes with originally a null duration.
\gateAND nodes are assigned the same agent and time slot as the child that occurs
last (l.~\ref{line:zero:and:start}--\ref{line:zero:and:end}).

Note that in all cases the agents (and time slots) assigned to zero duration nodes are the same as those
of their immediate parents or children. Hence, no further reshuffling is necessary.

%

\begin{proposition}
	\label{prop:algo:zero}
	Adding nodes of zero duration to the assignment in \Cref{algo:zero} does not affect
	the correctness of the scheduling.
\end{proposition}
\begin{proof}
	Since all remaining nodes have zero duration, no extra agents or time slots are necessary.
	In all cases, the zero duration node is assigned with either its immediate parent or child,
	preserving the correct execution order.
	Consider possible cases at each step of the algorithm:
	\begin{itemize}
		\item l.~\ref{line:zero:seq:start}--\ref{line:zero:seq:end}:
		      Nodes with a \gateSEQ parent are the
		      final nodes of sequences obtained during time normalisation.
		      Clearly, they can be assigned the same agent and time slot as their immediate parent without affecting the schedule.
		\item l.~\ref{line:zero:nullorleaf:start}--\ref{line:zero:nullorleaf:end}:
			      \gateOR nodes: in each \DAG/ variant (see \Cref{sec:pre:or}),
			            they
			            are guaranteed to have a single child node
			            and can be scheduled together with this child provided the corresponding
			            sub-\DAG/ has some duration.

			      \gateNULL and \gateLEAF{} nodes have a single child if any and are handled
			            analogously
			            to \gateOR, being assigned the same agent as their child. Note that \gateLEAF
			            {} nodes can have gotten this child during \eg{} the scheduling enforcement
			            step (see \Cref{sec:pre:sched}).

			      \gateOR, \gateNULL and \gateLEAF{} nodes with no child
			            or a child sub-\DAG/ with no duration
			            are assigned as their parent. If a \gateNULL{} node has
			            several parents due to sequence enforcement, it
			            gets the same assignment as its parent that occurs first.
		\item l.~\ref{line:zero:and:start}--\ref{line:zero:and:end}:
			      In case all children are never able to get an assignment, \ie{} they
			            are subtrees of null duration and can be identified with a depth 0, the \gateAND
			            {} node gets the same assignment as its parent.

			      Otherwise, \gateAND nodes are also scheduled together with
			            one
			            of their children.
			            Note that the \gateAND condition is satisfied only if all its longest children
			            have completed, therefore the one that occurs last, \ie{} has the biggest time
			            slot, is chosen (l.~\ref{line:zero:and:start}--\ref{line:zero:and:end}).
			            Furthermore, note that since children subtrees with a null
			            duration are discarded,
			            such children of the \gateAND node have
			            already been assigned an agent at that point.
	\end{itemize}
	The pathological case of a full \ADT/ with no duration is not handled since the
	algorithm is not called for such \DAG/s.
\end{proof}

\subsection{Complexity and correctness}
\label{sec:algo:proof}


We now consider the algorithm's complexity and prove that it achieves its intended goal.

\begin{proposition}
	\Cref{algo:minsched} is in $\mathcal{O}(kn^2\log n)$, where $k$ is
	the number of input \DAG/ variants, and $n$ their average number of nodes.
	\begin{proof}
		Initially, \textsc{DepthNode}, and \textsc{LevelNode}
		each visit all \DAG/ nodes, hence $2n$ operations. 
		In \textsc{Schedule}, the outer while loop (l. \ref{line:minsched:mainloop}) iterates over nodes of non-zero duration;
		its inner loop and \textsc{ReshuffleSlot} both operate within a single time slot.
		Overapproximating these numbers by $n$ puts the function at no more than $n^2$ operations.
		The schedule computation is repeated at most $\log n$ times in a divide-and-conquer strategy (l. \ref{line:minsched:newloop}).

		Finally, \textsc{ZeroAssign} visits all zero duration nodes (again overapproximated by $n$),
		performing at most $2n$ iterations for each, for a total of $2n^2$.
		Thus, the complexity of processing a single \DAG/ is $\mathcal{O}(2n + n^2\log n + 2n^2) = \mathcal{O}(n^2\log n)$,
		and $\mathcal{O}(kn^2\log n)$ for the whole input set.
		
		Note that as per \Cref{sec:preprocess}, the preprocessing step introduces a number of additional nodes in resulting \DAG/s.
		However, since that factor depends on the structure and attributes of the original \ADT/ rather than its size,
		it is treated as a constant in the consideration of complexity.
	\end{proof}
\end{proposition}

Thus, while the complexity of the scheduling algorithm itself is quadratic,
it is executed for $k$ input \DAG/ variants, where $k$ is exponential
in the number of $\gateOR$ and defence nodes in the \ADT/.

\begin{proposition}
	The assignments returned by \Cref{algo:minsched} are correct and use the minimal number of agents
	for each variant $\inputDAG \in \DAGset$ to achieve the attack in	minimal time.
	\begin{proof}
		Let $\numLevels$ denote the number of levels in an input variant $\inputDAG \in \DAGset$,
		and $L_i$ the set of nodes at the $i$-th level.
		We need to show that the resulting assignment is 1)~\emph{correct}, and 2)~\emph{optimal} in both schedule length and number of agents.

		1) \textsc{Schedule} assigns time slot~$1$ to leaves at the bottom level,
		subsequent slots to their ancestors, and finally the last one $\numLevels$ to the root node.
		Thus, the execution order of nodes in $\inputDAG$ is correct.
		Furthermore, it is guaranteed that there are enough agents to handle all nodes
		by increasing $\varNumAgents$ accordingly after an invalid assignment with unassigned nodes is discarded (l.~\ref{line:minsched:adjustagents}),
		and that any nodes executed in parallel (\ie at the same level) are assigned to different agents (l.~\ref{line:minsched:parexec}).
		Note also that the while loop at l.~\ref{line:minsched:newloop} of \textsc{MinSchedule} is guaranteed to terminate
		as the value of $\varNumAgents$ is refined from its theoretical bounds in a divide-and-conquer strategy.

		2) Since the number of time slots is fixed at $L$ (\Cref{algo:minsched},~l.~\ref{line:minsched:numslots}),
		\ie the minimal value that follows directly from the structure of $\inputDAG$ as its longest branch (note that $\numLevels = root(\inputDAG).depth$),
		it follows that the total attack time $\numLevels \times \tunit$ is always minimal.

		To show that the number of agents is also minimal, consider the assignment of nodes at each level $L_i$ of $\inputDAG$.
		The case for the top level $L_0$ is trivial: it only contains the root node,
		which cannot be executed in parallel with any other and thus can be assigned to any agent.
		By induction on subsequent levels $L_i$, we can show agents are also optimally assigned at each one.
		Suppose that the assignment of agents and time slots for all nodes down to and including level $L_i$ is optimal.
		At $L_{i+1}$, there are two possibilities to consider.
		If $|L_{i+1}|\leq \varNumAgents$, some agents are idle in this time slot.
		However, this assignment cannot be improved upon: note that any lower values of $\varNumAgents$
		would have already been checked during earlier cycles of the while loop (l.~\ref{line:minsched:newloop}),
		and found to produce an invalid schedule where some nodes are left without any agent asssigned (l.~\ref{line:minsched:discard}).

		Conversely, if $|L_{i+1}| > \varNumAgents$, some nodes will be carried over to $L_{i+2}$.
		Similarly, it follows from the divide-and-conquer scheme in which the final value of $\varNumAgents$ is obtained (l.~\ref{line:minsched:newloop})
		that decreasing the number of agents further is impossible without adding an extra slot instead.

		Therefore, the assignment up to level $L_{i+1}$ cannot be improved and is optimal for a schedule containing $\numLevels$ time slots.
		Note that subsequently executed subprocedures \textsc{ReshuffleSlot} and \textsc{ZeroAssign} do not affect this in any way,
		since neither adds extra agents or time slots.

		Thus, for any $\inputDAG \in \DAGset$, schedule length is fixed at its theoretical minimum,
		and the optimality of agent assignment for this minimal length follows from the fact
		time slots are filled exhaustively wherever possible,
		but using the lowest number of agents that does not leave unassigned nodes (\ie an invalid schedule).
		Since all input \DAG/ variants are equivalent to the original \ADT/ \wrt{} scheduling
		(by \Cref{prop:pre:zero,prop:pre:time,prop:pre:sand}),
		it also holds that the assignment is optimal for the original \ADT/.

	\end{proof}
\end{proposition}

\subsection{Scheduling for the treasure hunters \ADT/}
\label{sec:algo:ex}

We now apply these algorithms to the treasure hunters example.
\Cref{fig:sched:pruneNlabel} shows the output of the three initial subprocedures.
The depth of nodes assigned by \Cref{algo:depth} is displayed in \textcolor{green!60!black}
{green}.
The branch corresponding to attack \leaf{e} has been pruned 
as per \Cref{sec:pre:or}.
Levels assigned by \Cref{algo:level} are displayed in \textcolor{blue}{blue}.
Finally, the agents assignment computed by \Cref{algo:minsched} is shown in \Cref{fig:sched:assign}.

\begin{figure}[!!htb]
	\centering
		\rowcolors{2}{lightgray!30}{white}

		\begin{tabular}{c|l|l}
			\diagbox[]{slot}{agent} & 1                                                & 2                       \\
			\hline
			125                     & \leaf{h_3}, \leaf{GA'}, \leaf{TF'_2}, \leaf{TS'} &                         \\
			124                     & \leaf{h_2}                                       &                         \\
			123                     & \leaf{h_1}, \leaf{h'}                            &                         \\
			122                     & \leaf{ST_2}, \leaf{TF'_1}                        &                         \\
			121                     & \leaf{ST_1}, \leaf{ST'}                          &                         \\
			120                     & \leaf{f_{120}}                                   & \leaf{b_{60}}           \\
			$\cdots$                & $\cdots$                                         & $\cdots$                \\
			61                      & \leaf{f_{61}}                                    & \leaf{b_{1}}, \leaf{b'} \\
			60                      & \leaf{f_{60}}                                    &                         \\
			$\cdots$                & $\cdots$                                         &                         \\
			1                       & \leaf{f_{1}}, \leaf{f'}                          &                         \\
			\hline
		\end{tabular}

		\caption{Treasure hunters: Assignment of \Cref{algo:minsched}}
		\label{fig:sched:assign}
\end{figure}

%% file: experiments.tex


We have implemented the algorithms presented in this paper in our open source tool \tool~\cite{adt2amas}, written in \texttt{C++17}.
It allows for specifying input \ADT/s either via simple-syntax text files or using
an intuitive graphical user interface,
and handles both their translation to extended AMAS and subsequent computation of an optimal schedule with minimal number of agents.
Furthermore, intermediary steps of the algorithm can be exported as Tikz figures, allowing to easily visualise and understand them.
For more detailed information on the architecture of \tool, we refer the reader to~\cite{ADT2AMASDemoPaper}.
In the following, we present its application to the use cases from \cite{ICFEM2020}, plus
examples that feature some specific behaviour. The user can find 
all the figures and tables of the examples in the supplementary material of
this paper, which is available at \url{https://up13.fr/?nkPtK4eY}.

\paragraph*{forestall} This case study models forestalling a software instance. Depending
on the active defences, 4 cases are possible. However, the \DAG/ for no active
defence or only \leaf{id} active is the same. All three remaining
\DAG/s  have an optimal schedule with only 1 agent, in 43 days for the no defence
(or \leaf{id} only) case, 54 if only \leaf{scr} is active, and 55 if both defences
occur. Although only a single agent is needed to achieve the attack in minimal time,
the schedule exhibits which specific attacks must be performed to do so.

\paragraph*{iot-dev} This example models an attack on an IoT device via a network.
There are 4 cases, according to the active
defences, but only the one with no
defence leads to a \DAG/. Indeed, \leaf{tla} causes the failure of \gate{GVC} which
in turn makes \gate{APN} and then \gate{APNS} fail, independent of \leaf{inc}.
Thus the attack necessarily fails. This is also the case if defence \leaf{inc} is
active. The only way for an attack to succeed is that all defences fail, leading to
an optimal schedule in 694 minutes with 2 agents. Hence an attacker will use 2 agents
to perform the fastest attack. On the other hand, the defender knows that a single
one of the two defences is sufficient to block any attack.

\paragraph*{gain-admin} This third case study features an attacker trying to gain
administration privileges on a computer system. There are 16 possible defences
combinations,
which are covered
by only 3 cases: \leaf{scr} is not active; \leaf{scr} is active but not \gate{DTH};
both of them are active. In all three cases, the shortest attack requires only a single
agent, and can be scheduled
in 2942, 4320 and 5762 minutes, respectively.


\paragraph*{Exhibiting particular scheduling features} Experiments were conducted
on the example used in~\cite{ICFEM2020} to evaluate the impact of the number of agents on the
attack time and two small examples designed
to exhibit particular characteristics of the schedule.
Our algorithm confirms an optimal schedule in 5 minutes with 6
agents for the example of~\cite{ICFEM2020}.
Then, \emph{interrupted}
shows that the scheduling algorithm can produce an the interleaved execution of two
attacks (\leaf{b} and \leaf{e}), assigned to the same agent.
Finally, the \emph{last} example provides a succession of nodes with 0 duration (
\leaf{a'},
\leaf{e'}, \leaf{f'}, \leaf{h'}
and \leaf{i'}), and shows they are handled as expected.

\paragraph*{Scaling example}
In the \emph{scaling} example, the first agent processes the longest path while the
second agent handles all other actions.
It is extended to analyse the scaling capabilities of the scheduling
algorithm. For this purpose, we wrote an automatic generator of \ADT/s and a
notebook that processes the output of our tool in order to create
\autoref{fig:scaling_experiments}. The parameters of the
generated \ADT/s are the \emph{depth},
the \emph{width} corresponding to the number of deepmost leaves,
the number of \emph{children} for each \gate{AND},
and the total number of \emph{nodes}.
All nodes have time 1 except the first leaf that has time $\mathit{width}-1$.
The results show that the number of agents is not proportional to the width of the
tree (red bars - top of \autoref{fig:scaling_experiments}), and the optimal scheduling varies
according to the time of nodes (blue bars - bottom of \autoref{fig:scaling_experiments}).

\begin{figure}[t]
  \centering
  \includegraphics[width=\linewidth]{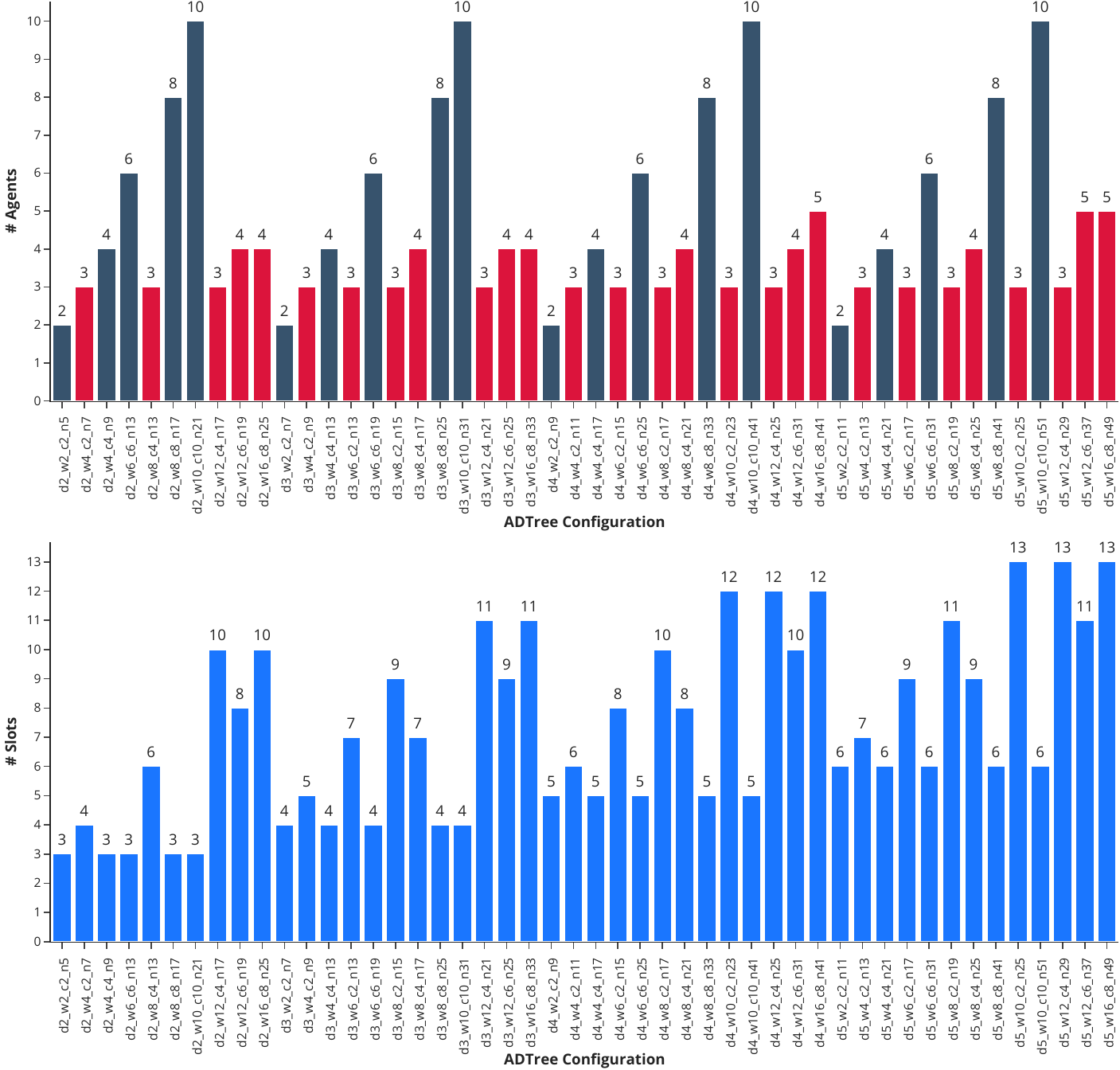}
  \caption{Scalability of different agent configurations\label{fig:scaling_experiments}}
\end{figure}

%% file: conclusion.tex

This paper has presented an agents scheduling algorithm that allows for evaluating
attack/defence models. It synthesises a minimal number of agents and their schedule,
providing insight to both parties as to the number of agents and actions necessary for a successful attack,
and the defences required to counter it.
While the scheduling algorithm itself is optimal,
further improvements can be made in the number of \DAG/ variants it is executed on.
One possible approach involves encoding configurations induced by \gateOR and defence nodes
to SAT or SMT and using a solver to find the optimal one, to be passed to the algorithm.

A natural extension is to consider characteristics other than time (\eg cost)
and, more importantly, the additional constraints on nodes,
thereby obtaining a complete framework for not only analysis but also synthesis of agent configurations
and schedules to achieve a given goal in a multi-agent system.
Targeting more elaborate goals, expressed in the TATL logic \cite{KnapikAPJP19},
will allow for analysing more general multi-agent systems and their properties.

%% file: examples/forestall.tex

\begin{figure}[ht]
	\vspace{-3ex}
	\centering
	\begin{tikzpicture}
		\node at (0,0) {\scalebox{.5}{\input{examples/ADTree_forestall.tikz}}};
		\node at (4,0) {\scalebox{.555}{\input{examples/tabAttributes_forestall}}};
	\end{tikzpicture}
	\caption{Forestall a software release (\csfs)}
	\label{fig:forestall}
\end{figure}

\begin{figure}[!htb]
  \centering
  \begin{subfigure}[b]{0.3\linewidth}
  \centering
    \scalebox{0.5}{
      \begin{tikzpicture}[node distance=1.8cm]
        \tikzstyle{SEQ}=[diamond]
        \tikzstyle{NULL}=[trapezium, trapezium left angle=120, trapezium right angle=120, minimum size=8mm]
        \tikzstyle{AND}=[and gate US, logic gate inputs=nn, rotate=90 ]
        \tikzstyle{OR}=[or gate US, logic gate inputs=nn, rotate=90 ]
        \tikzset{every node/.style={ultra thick, draw=red, minimum size=6mm}}
        \normalsize
  \node[draw=red, SEQ, xshift=0.000000cm ] (FS_10) {\ensuremath{\mathtt{FS_{10}}}};
\node[draw=none, blue, xshift=2mm, yshift=7mm] at (FS_10.east) {\small{\ensuremath{\mathtt{level}}}};\node[draw=none, green!60!black, xshift=-2mm, yshift=7mm] at (FS_10.west) {\small{\ensuremath{\mathtt{depth}}}};\node[draw=none, blue, xshift=2mm, yshift=0mm] at (FS_10.east) {\small{\ensuremath{\mathtt{0}}}};\node[draw=none, green!60!black, xshift=-2mm, yshift=0mm] at (FS_10.west) {\small{\ensuremath{\mathtt{43}}}};
  \node[draw=red, SEQ, xshift=-0.000000cm , below of=FS_10] (FS_1) {\ensuremath{\mathtt{FS_
  {1}}}};
\node[draw=none, blue, xshift=2mm, yshift=0mm] at (FS_1.east) {\small{\ensuremath{
\mathtt{9}}}};\node[draw=none, green!60!black, xshift=-2mm, yshift=0mm] at (FS_1.west) {\small{\ensuremath{\mathtt{34}}}};  \draw[dotted, red, thick] (FS_10) edge (FS_1);
  \node[draw=red, NULL, xshift=-0.000000cm , below of=FS_1] (FS'_3) {\ensuremath{\mathtt{FS'_{3}}}};
\node[draw=none, blue, xshift=2mm, yshift=0mm] at (FS'_3.east) {\small{\ensuremath{\mathtt{10}}}};\node[draw=none, green!60!black, xshift=-2mm, yshift=0mm] at (FS'_3.west) {\small{\ensuremath{\mathtt{33}}}};  \draw[solid] (FS_1) edge (FS'_3);
  \node[draw=red, SEQ, xshift=-0.000000cm , below of=FS'_3] (dtm_5) {\ensuremath{\mathtt{dtm_{5}}}};
\node[draw=none, blue, xshift=2mm, yshift=0mm] at (dtm_5.east) {\small{\ensuremath{\mathtt{10}}}};\node[draw=none, green!60!black, xshift=-2mm, yshift=0mm] at (dtm_5.west) {\small{\ensuremath{\mathtt{33}}}};  \draw[solid] (FS'_3) edge (dtm_5);
  \node[draw=red, SEQ, xshift=-0.000000cm , below of=dtm_5] (dtm_1) {\ensuremath{
  \mathtt{dtm_{1}}}};
\node[draw=none, blue, xshift=2mm, yshift=0mm] at (dtm_1.east) {\small{\ensuremath{
\mathtt{14}}}};\node[draw=none, green!60!black, xshift=-2mm, yshift=0mm] at (dtm_1.west) {\small{\ensuremath{\mathtt{29}}}};  \draw[dotted, red, thick] (dtm_5) edge (dtm_1);
  \node[draw=red, state, xshift=-0.000000cm , below of=dtm_1] (dtm') {\ensuremath{\mathtt{dtm'}}};
\node[draw=none, blue, xshift=2mm, yshift=0mm] at (dtm'.east) {\small{\ensuremath{\mathtt{15}}}};\node[draw=none, green!60!black, xshift=-2mm, yshift=0mm] at (dtm'.west) {\small{\ensuremath{\mathtt{28}}}};  \draw[solid] (dtm_1) edge (dtm');
  \node[draw=red, NULL, xshift=-0.000000cm , below of=dtm'] (FS'_2) {\ensuremath{\mathtt{FS'_{2}}}};
\node[draw=none, blue, xshift=2mm, yshift=0mm] at (FS'_2.east) {\small{\ensuremath{\mathtt{15}}}};\node[draw=none, green!60!black, xshift=-2mm, yshift=0mm] at (FS'_2.west) {\small{\ensuremath{\mathtt{28}}}};  \draw[solid] (dtm') edge (FS'_2);
  \node[draw=red, SEQ, xshift=-0.000000cm , below of=FS'_2] (icp_15) {\ensuremath{\mathtt{icp_{15}}}};
\node[draw=none, blue, xshift=2mm, yshift=0mm] at (icp_15.east) {\small{\ensuremath{\mathtt{15}}}};\node[draw=none, green!60!black, xshift=-2mm, yshift=0mm] at (icp_15.west) {\small{\ensuremath{\mathtt{28}}}};  \draw[solid] (FS'_2) edge (icp_15);
  \node[draw=red, SEQ, xshift=-0.000000cm , below of=icp_15] (icp_1) {\ensuremath{
  \mathtt{icp_{1}}}};
\node[draw=none, blue, xshift=2mm, yshift=0mm] at (icp_1.east) {\small{\ensuremath{
\mathtt{29}}}};\node[draw=none, green!60!black, xshift=-2mm, yshift=0mm] at (icp_1.west) {\small{\ensuremath{\mathtt{14}}}};  \draw[dotted, red, thick] (icp_15) edge (icp_1);
  \node[draw=red, state, xshift=-0.000000cm , below of=icp_1] (icp') {\ensuremath{\mathtt{icp'}}};
\node[draw=none, blue, xshift=2mm, yshift=0mm] at (icp'.east) {\small{\ensuremath{\mathtt{30}}}};\node[draw=none, green!60!black, xshift=-2mm, yshift=0mm] at (icp'.west) {\small{\ensuremath{\mathtt{13}}}};  \draw[solid] (icp_1) edge (icp');
  \node[draw=red, NULL, xshift=-0.000000cm , below of=icp'] (FS'_1) {\ensuremath{\mathtt{FS'_{1}}}};
\node[draw=none, blue, xshift=2mm, yshift=0mm] at (FS'_1.east) {\small{\ensuremath{\mathtt{30}}}};\node[draw=none, green!60!black, xshift=-2mm, yshift=0mm] at (FS'_1.west) {\small{\ensuremath{\mathtt{13}}}};  \draw[solid] (icp') edge (FS'_1);
  \node[draw=red, OR, xshift=-0.000000cm , yshift=4mm, below = 1.4cm of FS'_1.south] (SC') {\rotatebox {-90}{\ensuremath{\mathtt{SC'}}}};
\node[draw=none, blue, xshift=2mm, yshift=0mm] at (SC'.south) {\small{\ensuremath{\mathtt{30}}}};\node[draw=none, green!60!black, xshift=-2mm, yshift=0mm] at (SC'.north) {\small{\ensuremath{\mathtt{13}}}};  \draw[solid] (FS'_1) edge (SC'.east);
  \node[draw=red, NULL, xshift=-0.000000cm , below of=SC'] (PRS') {\ensuremath{\mathtt{PRS'}}};
\node[draw=none, blue, xshift=2mm, yshift=0mm] at (PRS'.east) {\small{\ensuremath{\mathtt{30}}}};\node[draw=none, green!60!black, xshift=-2mm, yshift=0mm] at (PRS'.west) {\small{\ensuremath{\mathtt{13}}}};  \draw[solid] (SC'.west) edge (PRS');
  \node[draw=red, NULL, xshift=-0.000000cm , below of=PRS'] (PR'_3) {\ensuremath{\mathtt{PR'_{3}}}};
\node[draw=none, blue, xshift=2mm, yshift=0mm] at (PR'_3.east) {\small{\ensuremath{\mathtt{30}}}};\node[draw=none, green!60!black, xshift=-2mm, yshift=0mm] at (PR'_3.west) {\small{\ensuremath{\mathtt{13}}}};  \draw[solid] (PRS') edge (PR'_3);
  \node[draw=red, state, xshift=-0.000000cm , below of=PR'_3] (rfc') {\ensuremath{\mathtt{rfc'}}};
\node[draw=none, blue, xshift=2mm, yshift=0mm] at (rfc'.east) {\small{\ensuremath{\mathtt{30}}}};\node[draw=none, green!60!black, xshift=-2mm, yshift=0mm] at (rfc'.west) {\small{\ensuremath{\mathtt{13}}}};  \draw[solid] (PR'_3) edge (rfc');
  \node[draw=red, NULL, xshift=-0.000000cm , below of=rfc'] (PR'_2) {\ensuremath{\mathtt{PR'_{2}}}};
\node[draw=none, blue, xshift=2mm, yshift=0mm] at (PR'_2.east) {\small{\ensuremath{\mathtt{30}}}};\node[draw=none, green!60!black, xshift=-2mm, yshift=0mm] at (PR'_2.west) {\small{\ensuremath{\mathtt{13}}}};  \draw[solid] (rfc') edge (PR'_2);
  \node[draw=red, SEQ, xshift=-0.000000cm , below of=PR'_2] (reb_3) {\ensuremath{\mathtt{reb_{3}}}};
\node[draw=none, blue, xshift=2mm, yshift=0mm] at (reb_3.east) {\small{\ensuremath{\mathtt{30}}}};\node[draw=none, green!60!black, xshift=-2mm, yshift=0mm] at (reb_3.west) {\small{\ensuremath{\mathtt{13}}}};  \draw[solid] (PR'_2) edge (reb_3);
  \node[draw=red, SEQ, xshift=-0.000000cm , below of=reb_3] (reb_2) {\ensuremath{\mathtt{reb_{2}}}};
\node[draw=none, blue, xshift=2mm, yshift=0mm] at (reb_2.east) {\small{\ensuremath{\mathtt{31}}}};\node[draw=none, green!60!black, xshift=-2mm, yshift=0mm] at (reb_2.west) {\small{\ensuremath{\mathtt{12}}}};  \draw[solid] (reb_3) edge (reb_2);
  \node[draw=red, SEQ, xshift=-0.000000cm , below of=reb_2] (reb_1) {\ensuremath{\mathtt{reb_{1}}}};
\node[draw=none, blue, xshift=2mm, yshift=0mm] at (reb_1.east) {\small{\ensuremath{\mathtt{32}}}};\node[draw=none, green!60!black, xshift=-2mm, yshift=0mm] at (reb_1.west) {\small{\ensuremath{\mathtt{11}}}};  \draw[solid] (reb_2) edge (reb_1);
  \node[draw=red, state, xshift=-0.000000cm , below of=reb_1] (reb') {\ensuremath{\mathtt{reb'}}};
\node[draw=none, blue, xshift=2mm, yshift=0mm] at (reb'.east) {\small{\ensuremath{\mathtt{33}}}};\node[draw=none, green!60!black, xshift=-2mm, yshift=0mm] at (reb'.west) {\small{\ensuremath{\mathtt{10}}}};  \draw[solid] (reb_1) edge (reb');
  \node[draw=red, NULL, xshift=-0.000000cm , below of=reb'] (PR'_1) {\ensuremath{\mathtt{PR'_{1}}}};
\node[draw=none, blue, xshift=2mm, yshift=0mm] at (PR'_1.east) {\small{\ensuremath{\mathtt{33}}}};\node[draw=none, green!60!black, xshift=-2mm, yshift=0mm] at (PR'_1.west) {\small{\ensuremath{\mathtt{10}}}};  \draw[solid] (reb') edge (PR'_1);
  \node[draw=red, SEQ, xshift=-0.000000cm , below of=PR'_1] (hr_10) {\ensuremath{\mathtt{hr_{10}}}};
\node[draw=none, blue, xshift=2mm, yshift=0mm] at (hr_10.east) {\small{\ensuremath{\mathtt{33}}}};\node[draw=none, green!60!black, xshift=-2mm, yshift=0mm] at (hr_10.west) {\small{\ensuremath{\mathtt{10}}}};  \draw[solid] (PR'_1) edge (hr_10);
  \node[draw=red, SEQ, xshift=-0.000000cm , below of=hr_10] (hr_1) {\ensuremath{\mathtt{hr_
  {1}}}};
\node[draw=none, blue, xshift=2mm, yshift=0mm] at (hr_1.east) {\small{\ensuremath{
\mathtt{42}}}};\node[draw=none, green!60!black, xshift=-2mm, yshift=0mm] at (hr_1.west) {\small{\ensuremath{\mathtt{1}}}};  \draw[dotted, red, thick] (hr_10) edge (hr_1);
  \node[draw=red, state, xshift=-0.000000cm , below of=hr_1] (hr') {\ensuremath{\mathtt{hr'}}};
\node[draw=none, blue, xshift=2mm, yshift=0mm] at (hr'.east) {\small{\ensuremath{\mathtt{43}}}};\node[draw=none, green!60!black, xshift=-2mm, yshift=0mm] at (hr'.west) {\small{\ensuremath{\mathtt{0}}}};  \draw[solid] (hr_1) edge (hr');
      \end{tikzpicture}
    }

    \subcaption{Cases no defence or \leaf{id} active}
  \end{subfigure}
  \begin{subfigure}[b]{0.3\linewidth}
  \centering
    \scalebox{0.5}{
      \begin{tikzpicture}[node distance=1.8cm]
        \tikzstyle{SEQ}=[diamond]
        \tikzstyle{NULL}=[trapezium, trapezium left angle=120, trapezium right angle=120, minimum size=8mm]
        \tikzstyle{AND}=[and gate US, logic gate inputs=nn, rotate=90 ]
        \tikzstyle{OR}=[or gate US, logic gate inputs=nn, rotate=90 ]
        \tikzset{every node/.style={ultra thick, draw=red, minimum size=6mm}}
        \normalsize
  \node[draw=red, SEQ, xshift=0.000000cm ] (FS_10) {\ensuremath{\mathtt{FS_{10}}}};
\node[draw=none, blue, xshift=2mm, yshift=7mm] at (FS_10.east) {\small{\ensuremath{\mathtt{level}}}};\node[draw=none, green!60!black, xshift=-2mm, yshift=7mm] at (FS_10.west) {\small{\ensuremath{\mathtt{depth}}}};\node[draw=none, blue, xshift=2mm, yshift=0mm] at (FS_10.east) {\small{\ensuremath{\mathtt{0}}}};\node[draw=none, green!60!black, xshift=-2mm, yshift=0mm] at (FS_10.west) {\small{\ensuremath{\mathtt{54}}}};
  \node[draw=red, SEQ, xshift=-0.000000cm , below of=FS_10] (FS_1) {\ensuremath{\mathtt{FS_
  {1}}}};
\node[draw=none, blue, xshift=2mm, yshift=0mm] at (FS_1.east) {\small{\ensuremath{
\mathtt{9}}}};\node[draw=none, green!60!black, xshift=-2mm, yshift=0mm] at (FS_1.west) {\small{\ensuremath{\mathtt{45}}}};  \draw[dotted, red, thick] (FS_10) edge (FS_1);
  \node[draw=red, NULL, xshift=-0.000000cm , below of=FS_1] (FS'_3) {\ensuremath{\mathtt{FS'_{3}}}};
\node[draw=none, blue, xshift=2mm, yshift=0mm] at (FS'_3.east) {\small{\ensuremath{\mathtt{10}}}};\node[draw=none, green!60!black, xshift=-2mm, yshift=0mm] at (FS'_3.west) {\small{\ensuremath{\mathtt{44}}}};  \draw[solid] (FS_1) edge (FS'_3);
  \node[draw=red, SEQ, xshift=-0.000000cm , below of=FS'_3] (dtm_5) {\ensuremath{\mathtt{dtm_{5}}}};
\node[draw=none, blue, xshift=2mm, yshift=0mm] at (dtm_5.east) {\small{\ensuremath{\mathtt{10}}}};\node[draw=none, green!60!black, xshift=-2mm, yshift=0mm] at (dtm_5.west) {\small{\ensuremath{\mathtt{44}}}};  \draw[solid] (FS'_3) edge (dtm_5);
  \node[draw=red, SEQ, xshift=-0.000000cm , below of=dtm_5] (dtm_1) {\ensuremath{
  \mathtt{dtm_{1}}}};
\node[draw=none, blue, xshift=2mm, yshift=0mm] at (dtm_1.east) {\small{\ensuremath{
\mathtt{14}}}};\node[draw=none, green!60!black, xshift=-2mm, yshift=0mm] at (dtm_1.west)
{\small{\ensuremath{\mathtt{40}}}};  \draw[dotted, red, thick] (dtm_5) edge (dtm_1);
  \node[draw=red, state, xshift=-0.000000cm , below of=dtm_1] (dtm') {\ensuremath{\mathtt{dtm'}}};
\node[draw=none, blue, xshift=2mm, yshift=0mm] at (dtm'.east) {\small{\ensuremath{\mathtt{15}}}};\node[draw=none, green!60!black, xshift=-2mm, yshift=0mm] at (dtm'.west) {\small{\ensuremath{\mathtt{39}}}};  \draw[solid] (dtm_1) edge (dtm');
  \node[draw=red, NULL, xshift=-0.000000cm , below of=dtm'] (FS'_2) {\ensuremath{\mathtt{FS'_{2}}}};
\node[draw=none, blue, xshift=2mm, yshift=0mm] at (FS'_2.east) {\small{\ensuremath{\mathtt{15}}}};\node[draw=none, green!60!black, xshift=-2mm, yshift=0mm] at (FS'_2.west) {\small{\ensuremath{\mathtt{39}}}};  \draw[solid] (dtm') edge (FS'_2);
  \node[draw=red, SEQ, xshift=-0.000000cm , below of=FS'_2] (icp_15) {\ensuremath{\mathtt{icp_{15}}}};
\node[draw=none, blue, xshift=2mm, yshift=0mm] at (icp_15.east) {\small{\ensuremath{\mathtt{15}}}};\node[draw=none, green!60!black, xshift=-2mm, yshift=0mm] at (icp_15.west) {\small{\ensuremath{\mathtt{39}}}};  \draw[solid] (FS'_2) edge (icp_15);
  \node[draw=red, SEQ, xshift=-0.000000cm , below of=icp_15] (icp_1) {\ensuremath{
  \mathtt{icp_{1}}}};
\node[draw=none, blue, xshift=2mm, yshift=0mm] at (icp_1.east) {\small{\ensuremath{
\mathtt{29}}}};\node[draw=none, green!60!black, xshift=-2mm, yshift=0mm] at (icp_1.west) {\small{\ensuremath{\mathtt{25}}}};  \draw[dotted, red, thick] (icp_15) edge (icp_1);
  \node[draw=red, state, xshift=-0.000000cm , below of=icp_1] (icp') {\ensuremath{\mathtt{icp'}}};
\node[draw=none, blue, xshift=2mm, yshift=0mm] at (icp'.east) {\small{\ensuremath{\mathtt{30}}}};\node[draw=none, green!60!black, xshift=-2mm, yshift=0mm] at (icp'.west) {\small{\ensuremath{\mathtt{24}}}};  \draw[solid] (icp_1) edge (icp');
  \node[draw=red, NULL, xshift=-0.000000cm , below of=icp'] (FS'_1) {\ensuremath{\mathtt{FS'_{1}}}};
\node[draw=none, blue, xshift=2mm, yshift=0mm] at (FS'_1.east) {\small{\ensuremath{\mathtt{30}}}};\node[draw=none, green!60!black, xshift=-2mm, yshift=0mm] at (FS'_1.west) {\small{\ensuremath{\mathtt{24}}}};  \draw[solid] (icp') edge (FS'_1);
  \node[draw=red, OR, xshift=-0.000000cm , yshift=4mm, below = 1.4cm of FS'_1.south] (SC') {\rotatebox {-90}{\ensuremath{\mathtt{SC'}}}};
\node[draw=none, blue, xshift=2mm, yshift=0mm] at (SC'.south) {\small{\ensuremath{\mathtt{30}}}};\node[draw=none, green!60!black, xshift=-2mm, yshift=0mm] at (SC'.north) {\small{\ensuremath{\mathtt{24}}}};  \draw[solid] (FS'_1) edge (SC'.east);
  \node[draw=red, NULL, xshift=-0.000000cm , below of=SC'] (NAS') {\ensuremath{\mathtt{NAS'}}};
\node[draw=none, blue, xshift=2mm, yshift=0mm] at (NAS'.east) {\small{\ensuremath{\mathtt{30}}}};\node[draw=none, green!60!black, xshift=-2mm, yshift=0mm] at (NAS'.west) {\small{\ensuremath{\mathtt{24}}}};  \draw[solid] (SC'.west) edge (NAS');
  \node[draw=red, SEQ, xshift=-0.000000cm , below of=NAS'] (NA_1) {\ensuremath{\mathtt{NA_{1}}}};
\node[draw=none, blue, xshift=2mm, yshift=0mm] at (NA_1.east) {\small{\ensuremath{\mathtt{30}}}};\node[draw=none, green!60!black, xshift=-2mm, yshift=0mm] at (NA_1.west) {\small{\ensuremath{\mathtt{24}}}};  \draw[solid] (NAS') edge (NA_1);
  \node[draw=red, NULL, xshift=-0.000000cm , below of=NA_1] (NA'_3) {\ensuremath{\mathtt{NA'_{3}}}};
\node[draw=none, blue, xshift=2mm, yshift=0mm] at (NA'_3.east) {\small{\ensuremath{\mathtt{31}}}};\node[draw=none, green!60!black, xshift=-2mm, yshift=0mm] at (NA'_3.west) {\small{\ensuremath{\mathtt{23}}}};  \draw[solid] (NA_1) edge (NA'_3);
  \node[draw=red, SEQ, xshift=-0.000000cm , below of=NA'_3] (heb_3) {\ensuremath{\mathtt{heb_{3}}}};
\node[draw=none, blue, xshift=2mm, yshift=0mm] at (heb_3.east) {\small{\ensuremath{\mathtt{31}}}};\node[draw=none, green!60!black, xshift=-2mm, yshift=0mm] at (heb_3.west) {\small{\ensuremath{\mathtt{23}}}};  \draw[solid] (NA'_3) edge (heb_3);
  \node[draw=red, SEQ, xshift=-0.000000cm , below of=heb_3] (heb_2) {\ensuremath{\mathtt{heb_{2}}}};
\node[draw=none, blue, xshift=2mm, yshift=0mm] at (heb_2.east) {\small{\ensuremath{\mathtt{32}}}};\node[draw=none, green!60!black, xshift=-2mm, yshift=0mm] at (heb_2.west) {\small{\ensuremath{\mathtt{22}}}};  \draw[solid] (heb_3) edge (heb_2);
  \node[draw=red, SEQ, xshift=-0.000000cm , below of=heb_2] (heb_1) {\ensuremath{\mathtt{heb_{1}}}};
\node[draw=none, blue, xshift=2mm, yshift=0mm] at (heb_1.east) {\small{\ensuremath{\mathtt{33}}}};\node[draw=none, green!60!black, xshift=-2mm, yshift=0mm] at (heb_1.west) {\small{\ensuremath{\mathtt{21}}}};  \draw[solid] (heb_2) edge (heb_1);
  \node[draw=red, state, xshift=-0.000000cm , below of=heb_1] (heb') {\ensuremath{\mathtt{heb'}}};
\node[draw=none, blue, xshift=2mm, yshift=0mm] at (heb'.east) {\small{\ensuremath{\mathtt{34}}}};\node[draw=none, green!60!black, xshift=-2mm, yshift=0mm] at (heb'.west) {\small{\ensuremath{\mathtt{20}}}};  \draw[solid] (heb_1) edge (heb');
  \node[draw=red, NULL, xshift=-0.000000cm , below of=heb'] (NA'_2) {\ensuremath{\mathtt{NA'_{2}}}};
\node[draw=none, blue, xshift=2mm, yshift=0mm] at (NA'_2.east) {\small{\ensuremath{\mathtt{34}}}};\node[draw=none, green!60!black, xshift=-2mm, yshift=0mm] at (NA'_2.west) {\small{\ensuremath{\mathtt{20}}}};  \draw[solid] (heb') edge (NA'_2);
  \node[draw=red, state, xshift=-0.000000cm , below of=NA'_2] (sb') {\ensuremath{\mathtt{sb'}}};
\node[draw=none, blue, xshift=2mm, yshift=0mm] at (sb'.east) {\small{\ensuremath{\mathtt{34}}}};\node[draw=none, green!60!black, xshift=-2mm, yshift=0mm] at (sb'.west) {\small{\ensuremath{\mathtt{20}}}};  \draw[solid] (NA'_2) edge (sb');
  \node[draw=red, NULL, xshift=-0.000000cm , below of=sb'] (NA'_1) {\ensuremath{\mathtt{NA'_{1}}}};
\node[draw=none, blue, xshift=2mm, yshift=0mm] at (NA'_1.east) {\small{\ensuremath{\mathtt{34}}}};\node[draw=none, green!60!black, xshift=-2mm, yshift=0mm] at (NA'_1.west) {\small{\ensuremath{\mathtt{20}}}};  \draw[solid] (sb') edge (NA'_1);
  \node[draw=red, SEQ, xshift=-0.000000cm , below of=NA'_1] (hh_20) {\ensuremath{\mathtt{hh_{20}}}};
\node[draw=none, blue, xshift=2mm, yshift=0mm] at (hh_20.east) {\small{\ensuremath{\mathtt{34}}}};\node[draw=none, green!60!black, xshift=-2mm, yshift=0mm] at (hh_20.west) {\small{\ensuremath{\mathtt{20}}}};  \draw[solid] (NA'_1) edge (hh_20);
  \node[draw=red, SEQ, xshift=-0.000000cm , below of=hh_20] (hh_1) {\ensuremath{\mathtt{hh_{1}}}};
\node[draw=none, blue, xshift=2mm, yshift=0mm] at (hh_1.east) {\small{\ensuremath{
\mathtt{53}}}};\node[draw=none, green!60!black, xshift=-2mm, yshift=0mm] at (hh_1.west) {\small{\ensuremath{\mathtt{1}}}};  \draw[dotted, red, thick] (hh_20) edge (hh_1);
  \node[draw=red, state, xshift=-0.000000cm , below of=hh_1] (hh') {\ensuremath{\mathtt{hh'}}};
\node[draw=none, blue, xshift=2mm, yshift=0mm] at (hh'.east) {\small{\ensuremath{\mathtt{54}}}};\node[draw=none, green!60!black, xshift=-2mm, yshift=0mm] at (hh'.west) {\small{\ensuremath{\mathtt{0}}}};  \draw[solid] (hh_1) edge (hh');
      \end{tikzpicture}
    }
    \subcaption{Case \leaf{scr} active}
  \end{subfigure}
  \begin{subfigure}[b]{0.3\linewidth}
  \centering
    \scalebox{0.5}{
      \begin{tikzpicture}[node distance=1.8cm]
        \tikzstyle{SEQ}=[diamond]
        \tikzstyle{NULL}=[trapezium, trapezium left angle=120, trapezium right angle=120, minimum size=8mm]
        \tikzstyle{AND}=[and gate US, logic gate inputs=nn, rotate=90 ]
        \tikzstyle{OR}=[or gate US, logic gate inputs=nn, rotate=90 ]
        \tikzset{every node/.style={ultra thick, draw=red, minimum size=6mm}}
        \normalsize
  \node[draw=red, SEQ, xshift=0.000000cm ] (FS_10) {\ensuremath{\mathtt{FS_{10}}}};
\node[draw=none, blue, xshift=2mm, yshift=7mm] at (FS_10.east) {\small{\ensuremath{\mathtt{level}}}};\node[draw=none, green!60!black, xshift=-2mm, yshift=7mm] at (FS_10.west) {\small{\ensuremath{\mathtt{depth}}}};\node[draw=none, blue, xshift=2mm, yshift=0mm] at (FS_10.east) {\small{\ensuremath{\mathtt{0}}}};\node[draw=none, green!60!black, xshift=-2mm, yshift=0mm] at (FS_10.west) {\small{\ensuremath{\mathtt{55}}}};
  \node[draw=red, SEQ, xshift=-0.000000cm , below of=FS_10] (FS_1) {\ensuremath{\mathtt{FS_
  {1}}}};
\node[draw=none, blue, xshift=2mm, yshift=0mm] at (FS_1.east) {\small{\ensuremath{\mathtt{9}}}};\node[draw=none, green!60!black, xshift=-2mm, yshift=0mm] at (FS_1.west) {\small{\ensuremath{\mathtt{46}}}};\draw[dotted, red, thick] (FS_10) edge (FS_1);
  \node[draw=red, NULL, xshift=-0.000000cm , below of=FS_1] (FS'_3) {\ensuremath{\mathtt{FS'_{3}}}};
\node[draw=none, blue, xshift=2mm, yshift=0mm] at (FS'_3.east) {\small{\ensuremath{\mathtt{10}}}};\node[draw=none, green!60!black, xshift=-2mm, yshift=0mm] at (FS'_3.west) {\small{\ensuremath{\mathtt{45}}}};  \draw[solid] (FS_1) edge (FS'_3);
  \node[draw=red, SEQ, xshift=-0.000000cm , below of=FS'_3] (dtm_5) {\ensuremath{\mathtt{dtm_{5}}}};
\node[draw=none, blue, xshift=2mm, yshift=0mm] at (dtm_5.east) {\small{\ensuremath{\mathtt{10}}}};\node[draw=none, green!60!black, xshift=-2mm, yshift=0mm] at (dtm_5.west) {\small{\ensuremath{\mathtt{45}}}};  \draw[solid] (FS'_3) edge (dtm_5);
  \node[draw=red, SEQ, xshift=-0.000000cm , below of=dtm_5] (dtm_1) {\ensuremath{
  \mathtt{dtm_{1}}}};
\node[draw=none, blue, xshift=2mm, yshift=0mm] at (dtm_1.east) {\small{\ensuremath{
\mathtt{14}}}};\node[draw=none, green!60!black, xshift=-2mm, yshift=0mm] at (dtm_1.west) {\small{\ensuremath{\mathtt{41}}}};  \draw[dotted, red, thick] (dtm_5) edge (dtm_1);
  \node[draw=red, state, xshift=-0.000000cm , below of=dtm_1] (dtm') {\ensuremath{\mathtt{dtm'}}};
\node[draw=none, blue, xshift=2mm, yshift=0mm] at (dtm'.east) {\small{\ensuremath{\mathtt{15}}}};\node[draw=none, green!60!black, xshift=-2mm, yshift=0mm] at (dtm'.west) {\small{\ensuremath{\mathtt{40}}}};  \draw[solid] (dtm_1) edge (dtm');
  \node[draw=red, NULL, xshift=-0.000000cm , below of=dtm'] (FS'_2) {\ensuremath{\mathtt{FS'_{2}}}};
\node[draw=none, blue, xshift=2mm, yshift=0mm] at (FS'_2.east) {\small{\ensuremath{\mathtt{15}}}};\node[draw=none, green!60!black, xshift=-2mm, yshift=0mm] at (FS'_2.west) {\small{\ensuremath{\mathtt{40}}}};  \draw[solid] (dtm') edge (FS'_2);
  \node[draw=red, SEQ, xshift=-0.000000cm , below of=FS'_2] (icp_15) {\ensuremath{\mathtt{icp_{15}}}};
\node[draw=none, blue, xshift=2mm, yshift=0mm] at (icp_15.east) {\small{\ensuremath{\mathtt{15}}}};\node[draw=none, green!60!black, xshift=-2mm, yshift=0mm] at (icp_15.west) {\small{\ensuremath{\mathtt{40}}}};  \draw[solid] (FS'_2) edge (icp_15);
  \node[draw=red, SEQ, xshift=-0.000000cm , below of=icp_15] (icp_1) {\ensuremath{
  \mathtt{icp_{1}}}};
\node[draw=none, blue, xshift=2mm, yshift=0mm] at (icp_1.east) {\small{\ensuremath{
\mathtt{29}}}};\node[draw=none, green!60!black, xshift=-2mm, yshift=0mm] at (icp_1.west) {\small{\ensuremath{\mathtt{26}}}};  \draw[dotted, red, thick] (icp_15) edge (icp_1);
  \node[draw=red, state, xshift=-0.000000cm , below of=icp_1] (icp') {\ensuremath{\mathtt{icp'}}};
\node[draw=none, blue, xshift=2mm, yshift=0mm] at (icp'.east) {\small{\ensuremath{\mathtt{30}}}};\node[draw=none, green!60!black, xshift=-2mm, yshift=0mm] at (icp'.west) {\small{\ensuremath{\mathtt{25}}}};  \draw[solid] (icp_1) edge (icp');
  \node[draw=red, NULL, xshift=-0.000000cm , below of=icp'] (FS'_1) {\ensuremath{\mathtt{FS'_{1}}}};
\node[draw=none, blue, xshift=2mm, yshift=0mm] at (FS'_1.east) {\small{\ensuremath{\mathtt{30}}}};\node[draw=none, green!60!black, xshift=-2mm, yshift=0mm] at (FS'_1.west) {\small{\ensuremath{\mathtt{25}}}};  \draw[solid] (icp') edge (FS'_1);
  \node[draw=red, OR, xshift=-0.000000cm , yshift=4mm, below = 1.4cm of FS'_1.south] (SC') {\rotatebox {-90}{\ensuremath{\mathtt{SC'}}}};
\node[draw=none, blue, xshift=2mm, yshift=0mm] at (SC'.south) {\small{\ensuremath{\mathtt{30}}}};\node[draw=none, green!60!black, xshift=-2mm, yshift=0mm] at (SC'.north) {\small{\ensuremath{\mathtt{25}}}};  \draw[solid] (FS'_1) edge (SC'.east);
  \node[draw=red, SEQ, xshift=-0.000000cm , below of=SC'] (BRB_3) {\ensuremath{\mathtt{BRB_{3}}}};
\node[draw=none, blue, xshift=2mm, yshift=0mm] at (BRB_3.east) {\small{\ensuremath{\mathtt{30}}}};\node[draw=none, green!60!black, xshift=-2mm, yshift=0mm] at (BRB_3.west) {\small{\ensuremath{\mathtt{25}}}};  \draw[solid] (SC'.west) edge (BRB_3);
  \node[draw=red, SEQ, xshift=-0.000000cm , below of=BRB_3] (BRB_2) {\ensuremath{\mathtt{BRB_{2}}}};
\node[draw=none, blue, xshift=2mm, yshift=0mm] at (BRB_2.east) {\small{\ensuremath{\mathtt{31}}}};\node[draw=none, green!60!black, xshift=-2mm, yshift=0mm] at (BRB_2.west) {\small{\ensuremath{\mathtt{24}}}};  \draw[solid] (BRB_3) edge (BRB_2);
  \node[draw=red, SEQ, xshift=-0.000000cm , below of=BRB_2] (BRB_1) {\ensuremath{\mathtt{BRB_{1}}}};
\node[draw=none, blue, xshift=2mm, yshift=0mm] at (BRB_1.east) {\small{\ensuremath{\mathtt{32}}}};\node[draw=none, green!60!black, xshift=-2mm, yshift=0mm] at (BRB_1.west) {\small{\ensuremath{\mathtt{23}}}};  \draw[solid] (BRB_2) edge (BRB_1);
  \node[draw=red, NULL, xshift=-0.000000cm , below of=BRB_1] (BRB'_2) {\ensuremath{\mathtt{BRB'_{2}}}};
\node[draw=none, blue, xshift=2mm, yshift=0mm] at (BRB'_2.east) {\small{\ensuremath{\mathtt{33}}}};\node[draw=none, green!60!black, xshift=-2mm, yshift=0mm] at (BRB'_2.west) {\small{\ensuremath{\mathtt{22}}}};  \draw[solid] (BRB_1) edge (BRB'_2);
  \node[draw=red, SEQ, xshift=-0.000000cm , below of=BRB'_2] (psc_7) {\ensuremath{\mathtt{psc_{7}}}};
\node[draw=none, blue, xshift=2mm, yshift=0mm] at (psc_7.east) {\small{\ensuremath{\mathtt{33}}}};\node[draw=none, green!60!black, xshift=-2mm, yshift=0mm] at (psc_7.west) {\small{\ensuremath{\mathtt{22}}}};  \draw[solid] (BRB'_2) edge (psc_7);
  \node[draw=red, SEQ, xshift=-0.000000cm , below of=psc_7] (psc_1) {\ensuremath{
  \mathtt{psc_{1}}}};
\node[draw=none, blue, xshift=2mm, yshift=0mm] at (psc_1.east) {\small{\ensuremath{
\mathtt{39}}}};\node[draw=none, green!60!black, xshift=-2mm, yshift=0mm] at (psc_1.west) {\small{\ensuremath{\mathtt{16}}}};  \draw[dotted, red, thick] (psc_7) edge (psc_1);
  \node[draw=red, state, xshift=-0.000000cm , below of=psc_1] (psc') {\ensuremath{\mathtt{psc'}}};
\node[draw=none, blue, xshift=2mm, yshift=0mm] at (psc'.east) {\small{\ensuremath{\mathtt{40}}}};\node[draw=none, green!60!black, xshift=-2mm, yshift=0mm] at (psc'.west) {\small{\ensuremath{\mathtt{15}}}};  \draw[solid] (psc_1) edge (psc');
  \node[draw=red, NULL, xshift=-0.000000cm , below of=psc'] (BRB'_1) {\ensuremath{\mathtt{BRB'_{1}}}};
\node[draw=none, blue, xshift=2mm, yshift=0mm] at (BRB'_1.east) {\small{\ensuremath{\mathtt{40}}}};\node[draw=none, green!60!black, xshift=-2mm, yshift=0mm] at (BRB'_1.west) {\small{\ensuremath{\mathtt{15}}}};  \draw[solid] (psc') edge (BRB'_1);
  \node[draw=red, SEQ, xshift=-0.000000cm , below of=BRB'_1] (bp_15) {\ensuremath{\mathtt{bp_{15}}}};
\node[draw=none, blue, xshift=2mm, yshift=0mm] at (bp_15.east) {\small{\ensuremath{\mathtt{40}}}};\node[draw=none, green!60!black, xshift=-2mm, yshift=0mm] at (bp_15.west) {\small{\ensuremath{\mathtt{15}}}};  \draw[solid] (BRB'_1) edge (bp_15);
  \node[draw=red, SEQ, xshift=-0.000000cm , below of=bp_15] (bp_1) {\ensuremath{\mathtt{bp_
  {1}}}};
\node[draw=none, blue, xshift=2mm, yshift=0mm] at (bp_1.east) {\small{\ensuremath{
\mathtt{54}}}};\node[draw=none, green!60!black, xshift=-2mm, yshift=0mm] at (bp_1.west) {\small{\ensuremath{\mathtt{1}}}};  \draw[dotted, red, thick] (bp_15) edge (bp_1);
  \node[draw=red, state, xshift=-0.000000cm , below of=bp_1] (bp') {\ensuremath{\mathtt{bp'}}};
\node[draw=none, blue, xshift=2mm, yshift=0mm] at (bp'.east) {\small{\ensuremath{\mathtt{55}}}};\node[draw=none, green!60!black, xshift=-2mm, yshift=0mm] at (bp'.west) {\small{\ensuremath{\mathtt{0}}}};  \draw[solid] (bp_1) edge (bp');
      \end{tikzpicture}
    }
    \subcaption{Case \leaf{id} and \leaf{scr} active}
  \end{subfigure}

  \caption{Preprocessing the \texttt{forestall} \ADT/\label{fig:pre:forestall}.}
\end{figure}

\begin{table*}[!!htb]
    \caption{Assignment for \texttt{forestall}}
\rowcolors{2}{lightgray!30}{white}
\begin{subtable}[t]{0.4\linewidth}
\subcaption{\DAG/ (a)}
\centering
\begin{tabular}{c|l|}
\diagbox[]{slot}{agent}&1\\
\hline
1&{\ensuremath{\mathtt{hr',hr_{1}}}}\\
2&{\ensuremath{\mathtt{hr_{2}}}}\\
$\cdots$&$\cdots$\\
9&{\ensuremath{\mathtt{hr_{9}}}}\\
10&{\ensuremath{\mathtt{PR'_{1},hr_{10}}}}\\
11&{\ensuremath{\mathtt{reb',reb_{1}}}}\\
12&{\ensuremath{\mathtt{reb_{2}}}}\\
13&{\ensuremath{\mathtt{FS'_{1},PR'_{2},PR'_{3},PRS',SC',reb_{3},rfc'}}}\\
14&{\ensuremath{\mathtt{icp',icp_{1}}}}\\
15&{\ensuremath{\mathtt{icp_{2}}}}\\
$\cdots$&$\cdots$\\
27&{\ensuremath{\mathtt{icp_{14}}}}\\
28&{\ensuremath{\mathtt{FS'_{2},icp_{15}}}}\\
29&{\ensuremath{\mathtt{dtm',dtm_{1}}}}\\
30&{\ensuremath{\mathtt{dtm_{2}}}}\\
$\cdots$&$\cdots$\\
33&{\ensuremath{\mathtt{dtm_{5}}}}\\
34&{\ensuremath{\mathtt{FS'_{3},FS_{1}}}}\\
35&{\ensuremath{\mathtt{FS_{2}}}}\\
$\cdots$&$\cdots$\\
43&{\ensuremath{\mathtt{FS_{10}}}}\\
\hline
\end{tabular}
\end{subtable}
\begin{subtable}[t]{0.3\linewidth}
\subcaption{\DAG/ (b)}
\centering
\begin{tabular}{c|l|}
\diagbox[]{slot}{agent}&1\\
\hline
1&{\ensuremath{\mathtt{hh',hh_{1}}}}\\
2&{\ensuremath{\mathtt{hh_{2}}}}\\
$\cdots$&$\cdots$\\
19&{\ensuremath{\mathtt{hh_{19}}}}\\
20&{\ensuremath{\mathtt{NA'_{1},NA'_{2},hh_{20},sb'}}}\\
21&{\ensuremath{\mathtt{heb',heb_{1}}}}\\
22&{\ensuremath{\mathtt{heb_{2}}}}\\
23&{\ensuremath{\mathtt{heb_{3}}}}\\
24&{\ensuremath{\mathtt{FS'_{1},NA'_{3},NAS',NA_{1},SC'}}}\\
25&{\ensuremath{\mathtt{icp',icp_{1}}}}\\
26&{\ensuremath{\mathtt{icp_{2}}}}\\
$\cdots$&$\cdots$\\
38&{\ensuremath{\mathtt{icp_{14}}}}\\
39&{\ensuremath{\mathtt{FS'_{2},icp_{15}}}}\\
40&{\ensuremath{\mathtt{dtm',dtm_{1}}}}\\
41&{\ensuremath{\mathtt{dtm_{2}}}}\\
$\cdots$&$\cdots$\\
44&{\ensuremath{\mathtt{dtm_{5}}}}\\
45&{\ensuremath{\mathtt{FS'_{3,FS_1}}}}\\
46&{\ensuremath{\mathtt{FS_{2}}}}\\
$\cdots$&$\cdots$\\
54&{\ensuremath{\mathtt{FS_{10}}}}\\
\hline
\end{tabular}
\end{subtable}
\begin{subtable}[t]{0.25\linewidth}
\subcaption{\DAG/ (c)}
\centering
\begin{tabular}{c|l|}
\diagbox[]{slot}{agent}&1\\
\hline
1&{\ensuremath{\mathtt{bp',bp_{1}}}}\\
2&{\ensuremath{\mathtt{bp_{2}}}}\\
$\cdots$&$\cdots$\\
14&{\ensuremath{\mathtt{bp_{14}}}}\\
15&{\ensuremath{\mathtt{BRB'_{1},bp_{15}}}}\\
16&{\ensuremath{\mathtt{psc',psc_{1}}}}\\
17&{\ensuremath{\mathtt{psc_{2}}}}\\
$\cdots$&$\cdots$\\
22&{\ensuremath{\mathtt{psc_{7}}}}\\
23&{\ensuremath{\mathtt{BRB'_{2},BRB_{1}}}}\\
24&{\ensuremath{\mathtt{BRB_{2}}}}\\
25&{\ensuremath{\mathtt{BRB_{3},FS'_{1},SC'}}}\\
26&{\ensuremath{\mathtt{icp',icp_{1}}}}\\
27&{\ensuremath{\mathtt{icp_{2}}}}\\
$\cdots$&$\cdots$\\
39&{\ensuremath{\mathtt{icp_{14}}}}\\
40&{\ensuremath{\mathtt{FS'_{2},icp_{15}}}}\\
41&{\ensuremath{\mathtt{dtm',dtm_{1}}}}\\
42&{\ensuremath{\mathtt{dtm_{2}}}}\\
$\cdots$&$\cdots$\\
45&{\ensuremath{\mathtt{dtm_{5}}}}\\
46&{\ensuremath{\mathtt{FS'_{3,FS_1}}}}\\
47&{\ensuremath{\mathtt{FS_{2}}}}\\
$\cdots$&$\cdots$\\
55&{\ensuremath{\mathtt{FS_{10}}}}\\
\hline
\end{tabular}
\end{subtable}
\end{table*}

%% file: examples/ADTree_forestall.tikz

\begin{tikzpicture}[every node/.style={ultra thick,draw=red,minimum size=6mm},
	                node distance=1.5cm]

	\node[and gate US,point up,logic gate inputs=nnn, seq=8pt] (FS)
		{\rotatebox{-90}{\gate{FS}}};

	\node[state, below = 3mm of FS.west] (icp) {\leaf{icp}};
	\draw (icp.north) -- (FS.input 2);

	\node[state, right of = icp] (dtm) {\leaf{dtm}};
	\draw (dtm.north) -- ([yshift=0.15cm]dtm.north) -| (FS.input 3);

	\node[or gate US,point up,logic gate inputs=nnn,above of = icp,
		xshift=-4mm] (SC)
		{\rotatebox{-90}{\gate{SC}}};
	\draw (SC.east) -- ([yshift=0.15cm]SC.east) -| (FS.input 1);

	\node[and gate US,point up,logic gate inputs=ni, seq=10pt,
		left of = SC] (NAS) {\rotatebox{-90}{\gate{NAS}}};
	\draw (NAS.east) -- (SC.input 2);

	\node[and gate US,point up,logic gate inputs=ni,
		below of = NAS, node distance=3cm]
		(PRS) {\rotatebox{-90}{\gate{PRS}}};
	\draw (PRS.east) -- ([yshift=0.15cm]PRS.east) -| (SC.input 3);

	\node[and gate US,point up,logic gate inputs=nn, seq=5pt,
		above of = NAS, node distance=3cm]
		(BRB) {\rotatebox{-90}{\gate{BRB}}};
	\draw (BRB.east) -- ([yshift=0.15cm]BRB.east) -| (SC.input 1);

	\node[state, below = 4mm of BRB.west, xshift=-0.9cm]
		(bp) {\leaf{bp}};
	\draw (bp.north) -- ([yshift=0.15cm]bp.north) -| (BRB.input 1);

	\node[state, below = 4mm of BRB.west, xshift=0.9cm]
		(psc) {\leaf{psc}};
	\draw (psc.north) -- ([yshift=0.15cm]psc.north) -| (BRB.input 2);

	\node[and gate US,point up,logic gate inputs=nnn, seq=8pt,
		left = 8mm of NAS.input 1,xshift=-4mm]
		(NA) {\rotatebox{-90}{\gate{NA}}};
	\draw (NA.east) -- ([yshift=0.15cm]NA.east) -| (NAS.input 1);

	\node[state, below = 3mm of NA.input 3] (heb) {\leaf{heb}};
	\draw (heb.north) -- (NA.input 3);

	\node[state, left of = heb] (sb) {\leaf{sb}};
	\draw (sb.north) -- ([yshift=0.1cm]sb.north) -| (NA.input 2);

	\node[state, left of = sb] (hh) {\leaf{hh}};
	\draw (hh.north) -- ([yshift=0.18cm]hh.north) -| (NA.input 1);

	\node[rectangle,draw=Green,minimum size=8mm,
		below = 4mm of NAS.west, xshift=10mm]
		(id) {\leaf{id}};
	\draw (id.north) -- ([yshift=0.15cm]id.north) -| (NAS.input 2);

	\node[and gate US,point up,logic gate inputs=nnn, seq=8pt,
		left = 8mm of PRS.input 1,xshift=-4mm]
		(PR) {\rotatebox{-90}{\gate{PR}}};
	\draw (PR.east) -- ([yshift=0.15cm]PR.east) -| (PRS.input 1);

	\node[state, below = 3mm of PR.west] (reb) {\leaf{reb}};
	\draw (reb.north) -- (PR.input 2);

	\node[state, right of = reb] (rfc) {\leaf{rfc}};
	\draw (rfc.north) -- ([yshift=0.15cm]rfc.north) -| (PR.input 3);

	\node[state, left of = reb] (hr) {\leaf{hr}};
	\draw (hr.north) -- ([yshift=0.15cm]hr.north) -| (PR.input 1);

	\node[rectangle,draw=Green,minimum size=8mm,
		below = 4mm of PRS.west, xshift=10mm]
		(scr) {\leaf{scr}};
	\draw (scr.north) -- ([yshift=0.15cm]scr.north) -| (PRS.input 2);

\end{tikzpicture}

%% file: examples/iot-dev.tex

\begin{figure}[ht]
	\centering
	\begin{tikzpicture}
		\node at (0,0) {\scalebox{.5}{\input{examples/ADTree_iotdev.tikz}}};
		\node at (6,0) {\scalebox{.555}{\input{examples/tabAttributes_iotdev}}};
	\end{tikzpicture}
	\caption{Compromise IoT device (\csiot)}
	\label{fig:iot}
\end{figure}

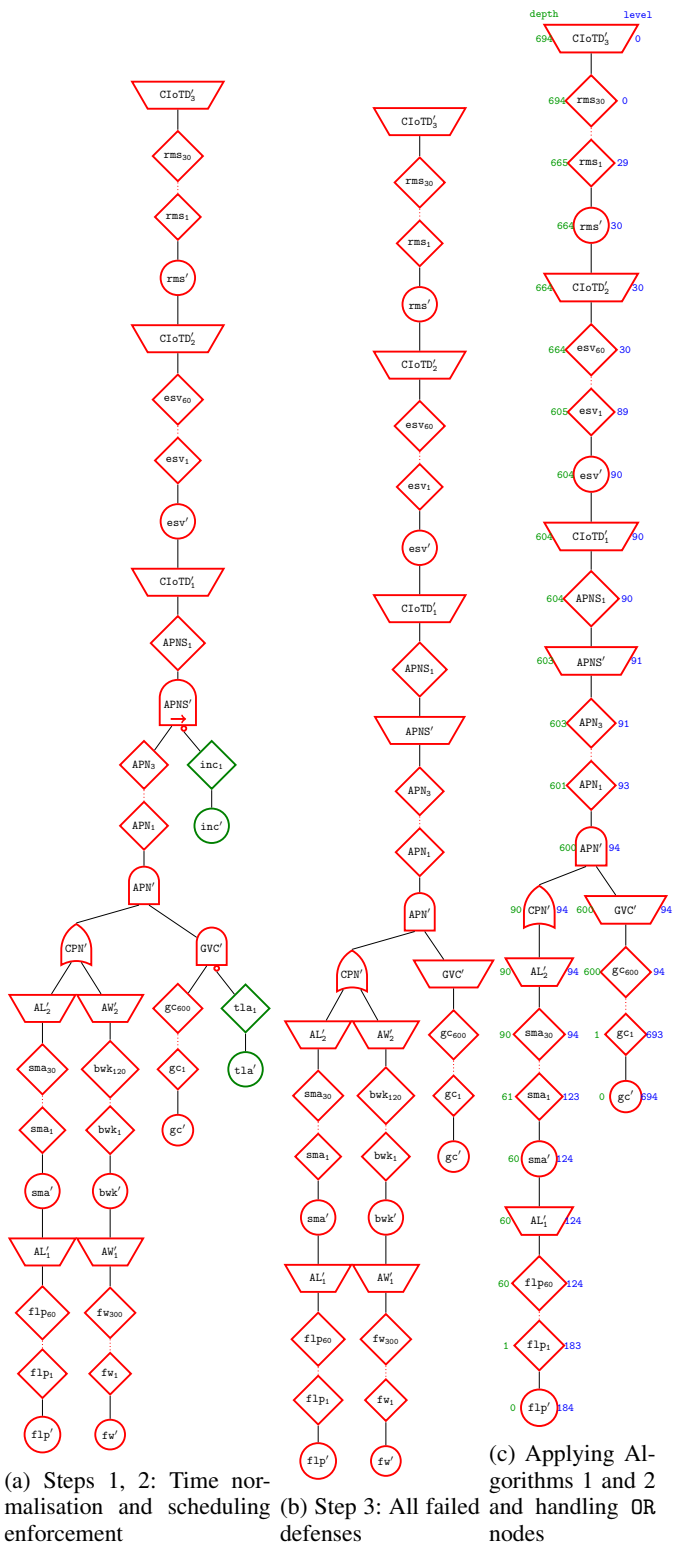
\begin{figure}[!htb]
  \centering
  \begin{subfigure}[b]{0.4\linewidth}
  \centering
    \scalebox{0.45}{
      \begin{tikzpicture}[node distance=1.8cm]
        \tikzstyle{SEQ}=[diamond]
        \tikzstyle{NULL}=[trapezium, trapezium left angle=120, trapezium right angle=120, minimum size=8mm]
        \tikzstyle{AND}=[and gate US, logic gate inputs=nn, rotate=90 ]
        \tikzstyle{OR}=[or gate US, logic gate inputs=nn, rotate=90 ]
        \tikzset{every node/.style={ultra thick, draw=red, minimum size=6mm}}
        \normalsize
        \node[draw=red, NULL] (CIoTD'_3) {\nodeLabel{CIoTD'_{3}}};
        \node[draw=red, SEQ, xshift=-0.000000cm, below of =CIoTD'_3] (rms_30) {\nodeLabel
        {rms_{30}}};
        \draw[solid] (CIoTD'_3) edge (rms_30);
        \node[draw=red, SEQ, xshift=-0.000000cm, below of =rms_30] (rms_1) {\nodeLabel{rms_{1}}};
        \draw[dotted, red, thick] (rms_30) edge (rms_1);
        \node[draw=red, state, xshift=-0.000000cm, below of =rms_1] (rms') {\nodeLabel{rms'}};
        \draw[solid] (rms_1) edge (rms');
        \node[draw=red, NULL, xshift=-0.000000cm, below of =rms'] (CIoTD'_2) {\nodeLabel{CIoTD'_{2}}};
        \draw[solid] (rms') edge (CIoTD'_2);
        \node[draw=red, SEQ, xshift=-0.000000cm, below of =CIoTD'_2] (esv_60) {\nodeLabel{esv_{60}}};
        \draw[solid] (CIoTD'_2) edge (esv_60);
        \node[draw=red, SEQ, xshift=-0.000000cm, below of =esv_60] (esv_1) {\nodeLabel{esv_{1}}};
        \draw[dotted, red, thick] (esv_60) edge (esv_1);
        \node[draw=red, state, xshift=-0.000000cm, below of =esv_1] (esv') {\nodeLabel{esv'}};
        \draw[solid] (esv_1) edge (esv');
        \node[draw=red, NULL, xshift=-0.000000cm, below of =esv'] (CIoTD'_1) {\nodeLabel{CIoTD'_{1}}};
        \draw[solid] (esv') edge (CIoTD'_1);
        \node[draw=red, SEQ, xshift=-0.000000cm, below of =CIoTD'_1] (APNS_1) {\nodeLabel{APNS_{1}}};
        \draw[solid] (CIoTD'_1) edge (APNS_1);
        \node[draw=none,below of=APNS_1] (fakeAPNS') {};
        \node at (fakeAPNS'.center) [draw=red, AND, logic gate inputs=ni, seq=10pt]
        (APNS') {\rotatebox {-90}{\nodeLabel{APNS'}}};
        \draw[solid] (APNS_1) edge (APNS');
        \node[draw=red, SEQ, xshift=-1.000000cm, below of =APNS'] (APN_3) {\nodeLabel{APN_{3}}};
        \draw[solid] (APNS'.input 1) edge (APN_3);
        \node[draw=Green, SEQ, xshift=1.000000cm, below of =APNS'] (inc_1) {\nodeLabel{inc_{1}}};
        \draw[solid] (APNS'.input 2) edge (inc_1);
        \node[draw=Green, state, xshift=-0.000000cm, below of =inc_1] (inc') {\nodeLabel{inc'}};
        \draw[solid] (inc_1) edge (inc');
        \node[draw=red, SEQ, xshift=-0.000000cm, below of =APN_3] (APN_1) {\nodeLabel{APN_{1}}};
        \draw[dotted, red, thick] (APN_3) edge (APN_1);
        \node[draw=none,below of=APN_1] (fakeAPN') {};
        \node at (fakeAPN'.center) [draw=red, AND, logic gate inputs=nn]
        (APN') {\rotatebox {-90}{\nodeLabel{APN'}}};
        \draw[solid] (APN_1) edge (APN');
        \node[draw=none, xshift=-2.000000cm, below of =APN'] (fakeCPN') {};
        \node at (fakeCPN'.center) [draw=red, OR, logic gate inputs=nn]
        (CPN') {\rotatebox {-90}{\nodeLabel{CPN'}}};
        \draw[solid] (APN'.input 1) edge (CPN'.east);
        \node[draw=none, xshift=2.000000cm, below of =APN'] (fakeGVC') {};
        \node at (fakeGVC'.center) [draw=red, AND, logic gate inputs=ni]
        (GVC') {\rotatebox {-90}{\nodeLabel{GVC'}}};
        \draw[solid] (APN'.input 2) edge (GVC');
        \node[draw=red, SEQ, xshift=-1.000000cm, below of =GVC'] (gc_600) {\nodeLabel{gc_{600}}};
        \draw[solid] (GVC'.input 1) edge (gc_600);
        \node[draw=Green, SEQ, xshift=1.000000cm, below of =GVC'] (tla_1) {\nodeLabel{tla_{1}}};
        \draw[solid] (GVC'.input 2) edge (tla_1);
        \node[draw=Green, state, xshift=-0.000000cm, below of =tla_1] (tla') {\nodeLabel{tla'}};
        \draw[solid] (tla_1) edge (tla');
        \node[draw=red, SEQ, xshift=-0.000000cm, below of =gc_600] (gc_1) {\nodeLabel
        {gc_{1}}};
        \draw[dotted, red, thick] (gc_600) edge (gc_1);
        \node[draw=red, state, xshift=-0.000000cm, below of =gc_1] (gc') {\nodeLabel{gc'}};
        \draw[solid] (gc_1) edge (gc');
        \node[draw=red, NULL, xshift=-1.000000cm, below of =CPN'] (AL'_2) {\nodeLabel
        {AL'_{2}}};
        \draw[solid] (CPN'.input 1) edge (AL'_2);
        \node[draw=red, NULL, xshift=1.000000cm, below of =CPN'] (AW'_2) {\nodeLabel
        {AW'_{2}}};
        \draw[solid] (CPN'.input 2) edge (AW'_2);
        \node[draw=red, SEQ, xshift=-0.000000cm, below of =AW'_2] (bwk_120) {\nodeLabel
        {bwk_{120}}};
        \draw[solid] (AW'_2) edge (bwk_120);
        \node[draw=red, SEQ, xshift=-0.000000cm, below of =bwk_120] (bwk_1) {\nodeLabel{bwk_{1}}};
        \draw[dotted, red, thick] (bwk_120) edge (bwk_1);
        \node[draw=red, state, xshift=-0.000000cm, below of =bwk_1] (bwk') {\nodeLabel{bwk'}};
        \draw[solid] (bwk_1) edge (bwk');
        \node[draw=red, NULL, xshift=-0.000000cm, below of =bwk'] (AW'_1) {\nodeLabel{AW'_{1}}};
        \draw[solid] (bwk') edge (AW'_1);
        \node[draw=red, SEQ, xshift=-0.000000cm, below of =AW'_1] (fw_300) {\nodeLabel
        {fw_{300}}};
        \draw[solid] (AW'_1) edge (fw_300);
        \node[draw=red, SEQ, xshift=-0.000000cm, below of =fw_300] (fw_1) {\nodeLabel{fw_{1}}};
        \draw[dotted, red, thick] (fw_300) edge (fw_1);
        \node[draw=red, state, xshift=-0.000000cm, below of =fw_1] (fw') {\nodeLabel{fw'}};
        \draw[solid] (fw_1) edge (fw');
        \node[draw=red, SEQ, xshift=-0.000000cm, below of =AL'_2] (sma_30) {\nodeLabel
        {sma_{30}}};
        \draw[solid] (AL'_2) edge (sma_30);
        \node[draw=red, SEQ, xshift=-0.000000cm, below of =sma_30] (sma_1) {\nodeLabel{sma_{1}}};
        \draw[dotted, red, thick] (sma_30) edge (sma_1);
        \node[draw=red, state, xshift=-0.000000cm, below of =sma_1] (sma') {\nodeLabel{sma'}};
        \draw[solid] (sma_1) edge (sma');
        \node[draw=red, NULL, xshift=-0.000000cm, below of =sma'] (AL'_1) {\nodeLabel{AL'_{1}}};
        \draw[solid] (sma') edge (AL'_1);
        \node[draw=red, SEQ, xshift=-0.000000cm, below of =AL'_1] (flp_60) {\nodeLabel
        {flp_{60}}};
        \draw[solid] (AL'_1) edge (flp_60);
        \node[draw=red, SEQ, xshift=-0.000000cm, below of =flp_60] (flp_1) {\nodeLabel{flp_{1}}};
        \draw[dotted, red, thick] (flp_60) edge (flp_1);
        \node[draw=red, state, xshift=-0.000000cm, below of =flp_1] (flp') {\nodeLabel{flp'}};
        \draw[solid] (flp_1) edge (flp');
      \end{tikzpicture}
    }

    \subcaption{Steps 1, 2: Time normalisation and scheduling enforcement}
  \end{subfigure}
  \begin{subfigure}[b]{0.3\linewidth}
  \centering
    \scalebox{0.45}{
      \begin{tikzpicture}[node distance=1.8cm]
        \tikzstyle{SEQ}=[diamond]
        \tikzstyle{NULL}=[trapezium, trapezium left angle=120, trapezium right angle=120, minimum size=8mm]
        \tikzstyle{AND}=[and gate US, logic gate inputs=nn, rotate=90 ]
        \tikzstyle{OR}=[or gate US, logic gate inputs=nn, rotate=90 ]
        \tikzset{every node/.style={ultra thick, draw=red, minimum size=6mm}}
        \normalsize
        \node[draw=red, NULL] (CIoTD'_3) {\nodeLabel{CIoTD'_{3}}};
        \node[draw=red, SEQ, xshift=-0.000000cm, below of =CIoTD'_3] (rms_30) {\nodeLabel
        {rms_{30}}};
        \draw[solid] (CIoTD'_3) edge (rms_30);
        \node[draw=red, SEQ, xshift=-0.000000cm, below of =rms_30] (rms_1) {\nodeLabel{rms_{1}}};
        \draw[dotted, red, thick] (rms_30) edge (rms_1);
        \node[draw=red, state, xshift=-0.000000cm, below of =rms_1] (rms') {\nodeLabel{rms'}};
        \draw[solid] (rms_1) edge (rms');
        \node[draw=red, NULL, xshift=-0.000000cm, below of =rms'] (CIoTD'_2) {\nodeLabel{CIoTD'_{2}}};
        \draw[solid] (rms') edge (CIoTD'_2);
        \node[draw=red, SEQ, xshift=-0.000000cm, below of =CIoTD'_2] (esv_60) {\nodeLabel{esv_{60}}};
        \draw[solid] (CIoTD'_2) edge (esv_60);
        \node[draw=red, SEQ, xshift=-0.000000cm, below of =esv_60] (esv_1) {\nodeLabel{esv_{1}}};
        \draw[dotted, red, thick] (esv_60) edge (esv_1);
        \node[draw=red, state, xshift=-0.000000cm, below of =esv_1] (esv') {\nodeLabel{esv'}};
        \draw[solid] (esv_1) edge (esv');
        \node[draw=red, NULL, xshift=-0.000000cm, below of =esv'] (CIoTD'_1) {\nodeLabel{CIoTD'_{1}}};
        \draw[solid] (esv') edge (CIoTD'_1);
        \node[draw=red, SEQ, xshift=-0.000000cm, below of =CIoTD'_1] (APNS_1) {\nodeLabel{APNS_{1}}};
        \draw[solid] (CIoTD'_1) edge (APNS_1);
        \node[draw=red, NULL, xshift=-0.000000cm, below of =APNS_1] (APNS') {\nodeLabel{APNS'}};
        \draw[solid] (APNS_1) edge (APNS');
        \node[draw=red, SEQ, xshift=-0.000000cm, below of =APNS'] (APN_3) {\nodeLabel{APN_{3}}};
        \draw[solid] (APNS') edge (APN_3);
        \node[draw=red, SEQ, xshift=-0.000000cm, below of =APN_3] (APN_1) {\nodeLabel{APN_{1}}};
        \draw[dotted, red, thick] (APN_3) edge (APN_1);
        \node[draw=none,below of=APN_1] (fakeAPN') {};
        \node at (fakeAPN'.center) [draw=red, AND, logic gate inputs=nn]
        (APN') {\rotatebox {-90}{\nodeLabel{APN'}}};
        \draw[solid] (APN_1) edge (APN');
        \node[draw=none, xshift=-2.000000cm, below of =APN'] (fakeCPN') {};
        \node at (fakeCPN'.center) [draw=red, OR]
        (CPN') {\rotatebox {-90}{\nodeLabel{CPN'}}};
        \draw[solid] (APN'.input 1) edge (CPN'.east);
        \node[draw=red, NULL, xshift=1.000000cm, below of =APN'] (GVC') {\nodeLabel
        {GVC'}};
        \draw[solid] (APN'.input 2) edge (GVC');
        \node[draw=red, SEQ, xshift=-0.000000cm, below of =GVC'] (gc_600) {\nodeLabel{gc_{600}}};
        \draw[solid] (GVC') edge (gc_600);
        \node[draw=red, SEQ, xshift=-0.000000cm, below of =gc_600] (gc_1) {\nodeLabel{gc_{1}}};
        \draw[dotted, red, thick] (gc_600) edge (gc_1);
        \node[draw=red, state, xshift=-0.000000cm, below of =gc_1] (gc') {\nodeLabel{gc'}};
        \draw[solid] (gc_1) edge (gc');
        \node[draw=red, NULL, xshift=-1.000000cm, below of =CPN'] (AL'_2) {\nodeLabel
        {AL'_{2}}};
        \draw[solid] (CPN') edge (AL'_2);
        \node[draw=red, NULL, xshift=1.000000cm, below of =CPN'] (AW'_2) {\nodeLabel
        {AW'_{2}}};
        \draw[solid] (CPN') edge (AW'_2);
        \node[draw=red, SEQ, xshift=-0.000000cm, below of =AW'_2] (bwk_120) {\nodeLabel
        {bwk_{120}}};
        \draw[solid] (AW'_2) edge (bwk_120);
        \node[draw=red, SEQ, xshift=-0.000000cm, below of =bwk_120] (bwk_1) {\nodeLabel{bwk_{1}}};
        \draw[dotted, red, thick] (bwk_120) edge (bwk_1);
        \node[draw=red, state, xshift=-0.000000cm, below of =bwk_1] (bwk') {\nodeLabel{bwk'}};
        \draw[solid] (bwk_1) edge (bwk');
        \node[draw=red, NULL, xshift=-0.000000cm, below of =bwk'] (AW'_1) {\nodeLabel
        {AW'_{1}}};
        \draw[solid] (bwk') edge (AW'_1);
        \node[draw=red, SEQ, xshift=-0.000000cm, below of =AW'_1] (fw_300) {\nodeLabel{fw_{300}}};
        \draw[solid] (AW'_1) edge (fw_300);
        \node[draw=red, SEQ, xshift=-0.000000cm, below of =fw_300] (fw_1) {\nodeLabel{fw_{1}}};
        \draw[dotted, red, thick] (fw_300) edge (fw_1);
        \node[draw=red, state, xshift=-0.000000cm, below of =fw_1] (fw') {\nodeLabel{fw'}};
        \draw[solid] (fw_1) edge (fw');
        \node[draw=red, SEQ, xshift=-0.000000cm, below of =AL'_2] (sma_30) {\nodeLabel
        {sma_{30}}};
        \draw[solid] (AL'_2) edge (sma_30);
        \node[draw=red, SEQ, xshift=-0.000000cm, below of =sma_30] (sma_1) {\nodeLabel{sma_{1}}};
        \draw[dotted, red, thick] (sma_30) edge (sma_1);
        \node[draw=red, state, xshift=-0.000000cm, below of =sma_1] (sma') {\nodeLabel{sma'}};
        \draw[solid] (sma_1) edge (sma');
        \node[draw=red, NULL, xshift=-0.000000cm, below of =sma'] (AL'_1) {\nodeLabel
        {AL'_{1}}};
        \draw[solid] (sma') edge (AL'_1);
        \node[draw=red, SEQ, xshift=-0.000000cm, below of =AL'_1] (flp_60) {\nodeLabel
        {flp_{60}}};
        \draw[solid] (AL'_1) edge (flp_60);
        \node[draw=red, SEQ, xshift=-0.000000cm, below of =flp_60] (flp_1) {\nodeLabel{flp_{1}}};
        \draw[dotted, red, thick] (flp_60) edge (flp_1);
        \node[draw=red, state, xshift=-0.000000cm, below of =flp_1] (flp') {\nodeLabel{flp'}};
        \draw[solid] (flp_1) edge (flp');
      \end{tikzpicture}
    }
    \subcaption{Step 3: All failed defenses}
  \end{subfigure}
  \begin{subfigure}[b]{0.25\linewidth}
  \centering
    \scalebox{0.46}{
      \begin{tikzpicture}[node distance=1.8cm]
        \tikzstyle{SEQ}=[diamond]
        \tikzstyle{NULL}=[trapezium, trapezium left angle=120, trapezium right angle=120, minimum size=8mm]
        \tikzstyle{AND}=[and gate US, logic gate inputs=nn, rotate=90 ]
        \tikzstyle{OR}=[or gate US, logic gate inputs=nn, rotate=90 ]
        \tikzset{every node/.style={ultra thick, draw=red, minimum size=6mm}}
        \normalsize
        \node[draw=red, NULL] (CIoTD'_3) {\nodeLabel{CIoTD'_3}};
\node[draw=none, blue, xshift=2mm, yshift=7mm] at (CIoTD'_3.east) {\small{\ensuremath{\mathtt{level}}}};\node[draw=none, green!60!black, xshift=-2mm, yshift=7mm] at (CIoTD'_3.west) {\small{\ensuremath{\mathtt{depth}}}};\node[draw=none, blue, xshift=2mm, yshift=0mm] at (CIoTD'_3.east) {\small{\ensuremath{\mathtt{0}}}};\node[draw=none, green!60!black, xshift=-2mm, yshift=0mm] at (CIoTD'_3.west) {\small{\ensuremath{\mathtt{694}}}};
        \node[draw=red, SEQ, xshift=-0.000000cm, below of =CIoTD'_3] (rms_30) {\nodeLabel
        {rms_{30}}};
\node[draw=none, blue, xshift=2mm, yshift=0mm] at (rms_30.east) {\small{\ensuremath{\mathtt{0}}}};\node[draw=none, green!60!black, xshift=-2mm, yshift=0mm] at (rms_30.west) {\small{\ensuremath{\mathtt{694}}}};
        \draw[solid] (CIoTD'_3) edge (rms_30);
        \node[draw=red, SEQ, xshift=-0.000000cm, below of =rms_30] (rms_1) {\nodeLabel{rms_{1}}};
\node[draw=none, blue, xshift=2mm, yshift=0mm] at (rms_1.east) {\small{\ensuremath{\mathtt{29}}}};\node[draw=none, green!60!black, xshift=-2mm, yshift=0mm] at (rms_1.west) {\small{\ensuremath{\mathtt{665}}}};
        \draw[dotted, red, thick] (rms_30) edge (rms_1);
        \node[draw=red, state, xshift=-0.000000cm, below of =rms_1] (rms') {\nodeLabel{rms'}};
\node[draw=none, blue, xshift=2mm, yshift=0mm] at (rms'.east) {\small{\ensuremath{\mathtt{30}}}};\node[draw=none, green!60!black, xshift=-2mm, yshift=0mm] at (rms'.west) {\small{\ensuremath{\mathtt{664}}}};
        \draw[solid] (rms_1) edge (rms');
        \node[draw=red, NULL, xshift=-0.000000cm, below of =rms'] (CIoTD'_2) {\nodeLabel{CIoTD'_{2}}};
\node[draw=none, blue, xshift=2mm, yshift=0mm] at (CIoTD'_2.east) {\small{\ensuremath{\mathtt{30}}}};\node[draw=none, green!60!black, xshift=-2mm, yshift=0mm] at (CIoTD'_2.west) {\small{\ensuremath{\mathtt{664}}}};
        \draw[solid] (rms') edge (CIoTD'_2);
        \node[draw=red, SEQ, xshift=-0.000000cm, below of =CIoTD'_2] (esv_60) {\nodeLabel{esv_{60}}};
\node[draw=none, blue, xshift=2mm, yshift=0mm] at (esv_60.east) {\small{\ensuremath{\mathtt{30}}}};\node[draw=none, green!60!black, xshift=-2mm, yshift=0mm] at (esv_60.west) {\small{\ensuremath{\mathtt{664}}}};
        \draw[solid] (CIoTD'_2) edge (esv_60);
        \node[draw=red, SEQ, xshift=-0.000000cm, below of =esv_60] (esv_1) {\nodeLabel{esv_{1}}};
\node[draw=none, blue, xshift=2mm, yshift=0mm] at (esv_1.east) {\small{\ensuremath{\mathtt{89}}}};\node[draw=none, green!60!black, xshift=-2mm, yshift=0mm] at (esv_1.west) {\small{\ensuremath{\mathtt{605}}}};        \draw[dotted, red, thick] (esv_60) edge (esv_1);
        \node[draw=red, state, xshift=-0.000000cm, below of =esv_1] (esv') {\nodeLabel{esv'}};
\node[draw=none, blue, xshift=2mm, yshift=0mm] at (esv'.east) {\small{\ensuremath{\mathtt{90}}}};\node[draw=none, green!60!black, xshift=-2mm, yshift=0mm] at (esv'.west) {\small{\ensuremath{\mathtt{604}}}};
        \draw[solid] (esv_1) edge (esv');
        \node[draw=red, NULL, xshift=-0.000000cm, below of =esv'] (CIoTD'_1) {\nodeLabel{CIoTD'_{1}}};
\node[draw=none, blue, xshift=2mm, yshift=0mm] at (CIoTD'_1.east) {\small{\ensuremath{\mathtt{90}}}};\node[draw=none, green!60!black, xshift=-2mm, yshift=0mm] at (CIoTD'_1.west) {\small{\ensuremath{\mathtt{604}}}};
        \draw[solid] (esv') edge (CIoTD'_1);
        \node[draw=red, SEQ, xshift=-0.000000cm, below of =CIoTD'_1] (APNS_1) {\nodeLabel{APNS_{1}}};
\node[draw=none, blue, xshift=2mm, yshift=0mm] at (APNS_1.east) {\small{\ensuremath{\mathtt{90}}}};\node[draw=none, green!60!black, xshift=-2mm, yshift=0mm] at (APNS_1.west) {\small{\ensuremath{\mathtt{604}}}};
        \draw[solid] (CIoTD'_1) edge (APNS_1);
        \node[draw=red, NULL, xshift=-0.000000cm, below of =APNS_1] (APNS') {\nodeLabel{APNS'}};
\node[draw=none, blue, xshift=2mm, yshift=0mm] at (APNS'.east) {\small{\ensuremath{\mathtt{91}}}};\node[draw=none, green!60!black, xshift=-2mm, yshift=0mm] at (APNS'.west) {\small{\ensuremath{\mathtt{603}}}};
        \draw[solid] (APNS_1) edge (APNS');
        \node[draw=red, SEQ, xshift=-0.000000cm, below of =APNS'] (APN_3) {\nodeLabel{APN_{3}}};
\node[draw=none, blue, xshift=2mm, yshift=0mm] at (APN_3.east) {\small{\ensuremath{\mathtt{91}}}};\node[draw=none, green!60!black, xshift=-2mm, yshift=0mm] at (APN_3.west) {\small{\ensuremath{\mathtt{603}}}};
        \draw[solid] (APNS') edge (APN_3);
        \node[draw=red, SEQ, xshift=-0.000000cm, below of =APN_3] (APN_1) {\nodeLabel{APN_{1}}};
\node[draw=none, blue, xshift=2mm, yshift=0mm] at (APN_1.east) {\small{\ensuremath{\mathtt{93}}}};\node[draw=none, green!60!black, xshift=-2mm, yshift=0mm] at (APN_1.west) {\small{\ensuremath{\mathtt{601}}}};
        \draw[dotted, red, thick] (APN_3) edge (APN_1);
        \node[draw=none,below of=APN_1] (fakeAPN') {};
        \node at (fakeAPN'.center) [draw=red, AND, logic gate inputs=nn]
        (APN') {\rotatebox {-90}{\nodeLabel{APN'}}};
\node[draw=none, blue, xshift=2mm, yshift=0mm] at (APN'.south) {\small{\ensuremath{\mathtt{94}}}};\node[draw=none, green!60!black, xshift=-2mm, yshift=0mm] at (APN'.north) {\small{\ensuremath{\mathtt{600}}}};        \draw[solid] (APN_1) edge (APN');
        \node[draw=none, xshift=-1.500000cm, below of =APN'] (fakeCPN') {};
        \node at (fakeCPN'.center) [draw=red, OR]
        (CPN') {\rotatebox {-90}{\nodeLabel{CPN'}}};
\node[draw=none, blue, xshift=2mm, yshift=0mm] at (CPN'.south) {\small{\ensuremath{\mathtt{94}}}};\node[draw=none, green!60!black, xshift=-2mm, yshift=0mm] at (CPN'.north) {\small{\ensuremath{\mathtt{90}}}};
        \draw[solid] (APN'.input 1) edge (CPN'.east);
        \node[draw=red, NULL, xshift=1.000000cm, below of =APN'] (GVC') {\nodeLabel{GVC'}};
\node[draw=none, blue, xshift=2mm, yshift=0mm] at (GVC'.east) {\small{\ensuremath{\mathtt{94}}}};\node[draw=none, green!60!black, xshift=-2mm, yshift=0mm] at (GVC'.west) {\small{\ensuremath{\mathtt{600}}}};
        \draw[solid] (APN'.input 2) edge (GVC');
        \node[draw=red, SEQ, xshift=-0.000000cm, below of =GVC'] (gc_600) {\nodeLabel{gc_{600}}};
\node[draw=none, blue, xshift=2mm, yshift=0mm] at (gc_600.east) {\small{\ensuremath{\mathtt{94}}}};\node[draw=none, green!60!black, xshift=-2mm, yshift=0mm] at (gc_600.west) {\small{\ensuremath{\mathtt{600}}}};
        \draw[solid] (GVC') edge (gc_600);
        \node[draw=red, SEQ, xshift=-0.000000cm, below of =gc_600] (gc_1) {\nodeLabel{gc_{1}}};
\node[draw=none, blue, xshift=2mm, yshift=0mm] at (gc_1.east) {\small{\ensuremath{\mathtt{693}}}};\node[draw=none, green!60!black, xshift=-2mm, yshift=0mm] at (gc_1.west) {\small{\ensuremath{\mathtt{1}}}};
        \draw[dotted, red, thick] (gc_600) edge (gc_1);
        \node[draw=red, state, xshift=-0.000000cm, below of =gc_1] (gc') {\nodeLabel{gc'}};
\node[draw=none, blue, xshift=2mm, yshift=0mm] at (gc'.east) {\small{\ensuremath{\mathtt{694}}}};\node[draw=none, green!60!black, xshift=-2mm, yshift=0mm] at (gc'.west) {\small{\ensuremath{\mathtt{0}}}};
        \draw[solid] (gc_1) edge (gc');
        \node[draw=red, NULL, xshift=-0.000000cm, below of =CPN'] (AL'_2) {\nodeLabel
        {AL'_{2}}};
\node[draw=none, blue, xshift=2mm, yshift=0mm] at (AL'_2.east) {\small{\ensuremath{\mathtt{94}}}};\node[draw=none, green!60!black, xshift=-2mm, yshift=0mm] at (AL'_2.west) {\small{\ensuremath{\mathtt{90}}}};
        \draw[solid] (CPN') edge (AL'_2);
        \node[draw=red, SEQ, xshift=-0.000000cm, below of =AL'_2] (sma_30) {\nodeLabel
        {sma_{30}}};
\node[draw=none, blue, xshift=2mm, yshift=0mm] at (sma_30.east) {\small{\ensuremath{\mathtt{94}}}};\node[draw=none, green!60!black, xshift=-2mm, yshift=0mm] at (sma_30.west) {\small{\ensuremath{\mathtt{90}}}};
        \draw[solid] (AL'_2) edge (sma_30);
        \node[draw=red, SEQ, xshift=-0.000000cm, below of =sma_30] (sma_1) {\nodeLabel{sma_{1}}};
\node[draw=none, blue, xshift=2mm, yshift=0mm] at (sma_1.east) {\small{\ensuremath{\mathtt{123}}}};\node[draw=none, green!60!black, xshift=-2mm, yshift=0mm] at (sma_1.west) {\small{\ensuremath{\mathtt{61}}}};
        \draw[dotted, red, thick] (sma_30) edge (sma_1);
        \node[draw=red, state, xshift=-0.000000cm, below of =sma_1] (sma') {\nodeLabel{sma'}};
\node[draw=none, blue, xshift=2mm, yshift=0mm] at (sma'.east) {\small{\ensuremath{\mathtt{124}}}};\node[draw=none, green!60!black, xshift=-2mm, yshift=0mm] at (sma'.west) {\small{\ensuremath{\mathtt{60}}}};
        \draw[solid] (sma_1) edge (sma');
        \node[draw=red, NULL, xshift=-0.000000cm, below of =sma'] (AL'_1) {\nodeLabel{AL'_{1}}};
\node[draw=none, blue, xshift=2mm, yshift=0mm] at (AL'_1.east) {\small{\ensuremath{\mathtt{124}}}};\node[draw=none, green!60!black, xshift=-2mm, yshift=0mm] at (AL'_1.west) {\small{\ensuremath{\mathtt{60}}}};
        \draw[solid] (sma') edge (AL'_1);
        \node[draw=red, SEQ, xshift=-0.000000cm, below of =AL'_1] (flp_60) {\nodeLabel{flp_{60}}};
\node[draw=none, blue, xshift=2mm, yshift=0mm] at (flp_60.east) {\small{\ensuremath{\mathtt{124}}}};\node[draw=none, green!60!black, xshift=-2mm, yshift=0mm] at (flp_60.west) {\small{\ensuremath{\mathtt{60}}}};
        \draw[solid] (AL'_1) edge (flp_60);
        \node[draw=red, SEQ, xshift=-0.000000cm, below of =flp_60] (flp_1) {\nodeLabel{flp_{1}}};
\node[draw=none, blue, xshift=2mm, yshift=0mm] at (flp_1.east) {\small{\ensuremath{\mathtt{183}}}};\node[draw=none, green!60!black, xshift=-2mm, yshift=0mm] at (flp_1.west) {\small{\ensuremath{\mathtt{1}}}};
        \draw[dotted, red, thick] (flp_60) edge (flp_1);
        \node[draw=red, state, xshift=-0.000000cm, below of =flp_1] (flp') {\nodeLabel{flp'}};
\node[draw=none, blue, xshift=2mm, yshift=0mm] at (flp'.east) {\small{\ensuremath{\mathtt{184}}}};\node[draw=none, green!60!black, xshift=-2mm, yshift=0mm] at (flp'.west) {\small{\ensuremath{\mathtt{0}}}};        \draw[solid] (flp_1) edge (flp');
      \end{tikzpicture}
    }
    \subcaption{Applying \Cref{algo:depth,algo:level} and handling \gateOR nodes}
  \end{subfigure}

  \caption{Preprocessing the \texttt{iot-dev} \ADT/\label{fig:pre:iot_dev:1}. When
  either \texttt{tla} or \texttt{inc} are enabled, step 3 produces an empty \DAG/.}
\end{figure}

\begin{table}[!!htb]
  \centering
    \caption{Assignment for \texttt{iot-dev}}
\rowcolors{2}{lightgray!30}{white}
\begin{tabular}{c|l|l|}
\diagbox[]{slot}{agent}&1&2\\
\hline
1&{\ensuremath{\mathtt{gc',gc_{1}}}}&{\ensuremath{\mathtt{}}}\\
2&{\ensuremath{\mathtt{gc_{2}}}}&{\ensuremath{\mathtt{}}}\\
$\cdots$&$\cdots$&{\ensuremath{\mathtt{}}}\\
510&{\ensuremath{\mathtt{gc_{510}}}}&{\ensuremath{\mathtt{}}}\\
511&{\ensuremath{\mathtt{gc_{511}}}}&{\ensuremath{\mathtt{flp',flp_{1}}}}\\
512&{\ensuremath{\mathtt{gc_{512}}}}&{\ensuremath{\mathtt{flp_{2}}}}\\
$\cdots$&$\cdots$&$\cdots$\\
569&{\ensuremath{\mathtt{gc_{569}}}}&{\ensuremath{\mathtt{flp_{59}}}}\\
570&{\ensuremath{\mathtt{gc_{570}}}}&{\ensuremath{\mathtt{AL'_{1},flp_{60}}}}\\
571&{\ensuremath{\mathtt{gc_{571}}}}&{\ensuremath{\mathtt{sma',sma_{1}}}}\\
572&{\ensuremath{\mathtt{gc_{572}}}}&{\ensuremath{\mathtt{sma_{2}}}}\\
$\cdots$&$\cdots$&$\cdots$\\
599&{\ensuremath{\mathtt{gc_{599}}}}&{\ensuremath{\mathtt{sma_{29}}}}\\
600&{\ensuremath{\mathtt{GVC',gc_{600}}}}&{\ensuremath{\mathtt{AL'_{2},CPN',sma_{30}}}}\\
601&{\ensuremath{\mathtt{APN',APN_{1}}}}&{\ensuremath{\mathtt{}}}\\
602&{\ensuremath{\mathtt{APN_{2}}}}&{\ensuremath{\mathtt{}}}\\
603&{\ensuremath{\mathtt{APN_{3}}}}&{\ensuremath{\mathtt{}}}\\
604&{\ensuremath{\mathtt{APNS',APNS_{1},CIoTD'_{1}}}}&{\ensuremath{\mathtt{}}}\\
605&{\ensuremath{\mathtt{esv',esv_{1}}}}&{\ensuremath{\mathtt{}}}\\
606&{\ensuremath{\mathtt{esv_{2}}}}&{\ensuremath{\mathtt{}}}\\

$\cdots$&$\cdots$&{\ensuremath{\mathtt{}}}\\
663&{\ensuremath{\mathtt{esv_{59}}}}&{\ensuremath{\mathtt{}}}\\
664&{\ensuremath{\mathtt{CIoTD'_{2},esv_{60}}}}&{\ensuremath{\mathtt{}}}\\
665&{\ensuremath{\mathtt{rms',rms_{1}}}}&{\ensuremath{\mathtt{}}}\\
666&{\ensuremath{\mathtt{rms_{2}}}}&{\ensuremath{\mathtt{}}}\\
$\cdots$&$\cdots$&{\ensuremath{\mathtt{}}}\\
693&{\ensuremath{\mathtt{rms_{29}}}}&{\ensuremath{\mathtt{}}}\\
694&{\ensuremath{\mathtt{CIoTD'_{3},rms_{30}}}}&{\ensuremath{\mathtt{}}}\\
\hline
\end{tabular}
\end{table}

%% file: examples/ADTree_iotdev.tikz

\begin{tikzpicture}
	[every node/.style={ultra thick,draw=red,minimum size=6mm},
	node distance=1.5cm]

	\node[and gate US,point up,logic gate inputs=nnn, seq=8pt] (CIoTD)
		{\rotatebox{-90}{\gate{CIoTD}}};

	\node[state, below = 3mm of CIoTD.west] (esv) {\leaf{esv}};
	\draw (esv.north) -- (CIoTD.input 2);

	\node[state, right of = esv] (rms) {\leaf{rms}};
	\draw (rms.north) -- ([yshift=0.15cm]rms.north) -| (CIoTD.input 3);

	\node[and gate US,point up,logic gate inputs=ni, seq=12pt,
		above of = esv, xshift=-2mm] (GAPNS)
		{\rotatebox{-90}{\gate{APNS}}};
	\draw (GAPNS.east) -- ([yshift=0.15cm]GAPNS.east) -| (CIoTD.input 1);

	\node[and gate US,point up,logic gate inputs=nn,
		left = 8mm of GAPNS.input 1,xshift=-4mm]
		(GAPN) {\rotatebox{-90}{\gate{APN}}};
	\draw (GAPN.east) -- ([yshift=0.15cm]GAPN.east) -| (GAPNS.input 1);

	\node[or gate US,point up,logic gate inputs=nn,
		below = 1cm of GAPN.input 1,yshift=14mm] (CPN)
		{\rotatebox{-90}{\gate{CPN}}};
	\draw (CPN.east) -- ([yshift=0.15cm]CPN.east) -| (GAPN.input 1);

	\node[and gate US,point up,logic gate inputs=nn, seq=5pt,
		below = 9mm of CPN.input 1, yshift=12mm] (AL)
		{\rotatebox{-90}{\gate{AL}}};
	\draw (AL.east) -- ([yshift=0.15cm]AL.east) -| (CPN.input 1);

	\node[state, below = 4mm of AL.input 1, xshift=-8mm]
		(flp) {\leaf{flp}};
	\draw (flp.north) -- ([yshift=0.15cm]flp.north) -| (AL.input 1);

	\node[state, below = 4mm of AL.input 2] (sma) {\leaf{sma}};
	\draw (sma.north) -- ([yshift=0.15cm]sma.north) -| (AL.input 2);

	\node[and gate US,point up,logic gate inputs=nn, seq=5pt,
		below = 9mm of CPN.input 2, yshift=-4mm] (AW)
		{\rotatebox{-90}{\gate{AW}}};
	\draw (AW.east) -- ([yshift=0.15cm]AW.east) -| (CPN.input 2);

	\node[state, below = 4mm of AW.input 1] (fw) {\leaf{fw}};
	\draw (fw.north) -- ([yshift=0.15cm]fw.north) -| (AW.input 1);

	\node[state, below = 4mm of AW.input 2, xshift=8mm]
		(bwk) {\leaf{bwk}};
	\draw (bwk.north) -- ([yshift=0.15cm]bwk.north) -| (AW.input 2);

	\node[and gate US,point up,logic gate inputs=ni,
		below = 9mm of GAPN.input 2, yshift=-6mm]
		(GVC) {\rotatebox{-90}{\gate{GVC}}};
	\draw (GVC.east) -- ([yshift=0.15cm]GVC.east) -| (GAPN.input 2);

	\node[state, below = 4mm of GVC.input 1] (gc) {\leaf{gc}};
	\draw (gc.north) -- ([yshift=0.15cm]gc.north) -| (GVC.input 1);

	\node[rectangle,draw=Green,minimum size=8mm,
		below = 4mm of GVC.west, xshift=10mm]
		(tla) {\leaf{tla}};
	\draw (tla.north) -- ([yshift=0.15cm]tla.north) -| (GVC.input 2);

	\node[rectangle,draw=Green,minimum size=8mm,
		below = 4mm of GAPNS.west, xshift=10mm]
		(inc) {\leaf{inc}};
	\draw (inc.north) -- ([yshift=0.15cm]inc.north) -| (GAPNS.input 2);
\end{tikzpicture}

%% file: examples/gain-admin.tex

\begin{figure}[ht]
	\centering
	\begin{tikzpicture}
		\node at (0,0)   {\scalebox{.5}{\input{examples/ADTree_gainadmin.tikz}}};
		\node at (4.5,0) {\scalebox{.555}{\input{examples/tabAttributes_gainadmin}}};
	\end{tikzpicture}
	\caption{Obtain admin privileges (\csadmin)}
	\label{fig:admin}
\end{figure}

\begin{figure}[!htb]
  \centering
  \begin{subfigure}[b]{0.3\linewidth}
  \centering
    \scalebox{0.5}{
      \begin{tikzpicture}[node distance=1.8cm]
        \tikzstyle{SEQ}=[diamond]
        \tikzstyle{NULL}=[trapezium, trapezium left angle=120, trapezium right angle=120, minimum size=8mm]
        \tikzstyle{AND}=[and gate US, logic gate inputs=nn, rotate=90 ]
        \tikzstyle{OR}=[or gate US, logic gate inputs=nn, rotate=90 ]
        \tikzset{every node/.style={ultra thick, draw=red, minimum size=6mm}}
        \normalsize
  \node[draw=red, OR, xshift=0.000000cm ] (OAP') {\rotatebox {-90}{\ensuremath{\mathtt{OAP'}}}};
  \node[draw=none, blue, xshift=2mm, yshift=7mm] at (OAP'.south) {\small{\ensuremath{\mathtt{level}}}};
  \node[draw=none, green!60!black, xshift=-2mm, yshift=7mm] at (OAP'.north) {\small{\ensuremath{\mathtt{depth}}}};
  \node[draw=none, blue, xshift=2mm, yshift=0mm] at (OAP'.south) {\small{\ensuremath{\mathtt{0}}}};
  \node[draw=none, green!60!black, xshift=-4mm, yshift=0mm] at (OAP'.north) {\small{
  \ensuremath{\mathtt{2942}}}};
  \node[draw=red, SEQ, xshift=-0.000000cm , below of=OAP'] (ACLI_2) {\ensuremath{\mathtt{ACLI_{2}}}};
  \node[draw=none, blue, xshift=2mm, yshift=0mm] at (ACLI_2.east) {\small{\ensuremath{\mathtt{0}}}};
  \node[draw=none, green!60!black, xshift=-4mm, yshift=0mm] at (ACLI_2.west) {\small{
  \ensuremath{\mathtt{2942}}}};
  \draw[solid] (OAP'.west) edge (ACLI_2);
  \node[draw=red, SEQ, xshift=-0.000000cm , below of=ACLI_2] (ACLI_1) {\ensuremath{\mathtt{ACLI_{1}}}};
  \node[draw=none, blue, xshift=2mm, yshift=0mm] at (ACLI_1.east) {\small{\ensuremath{\mathtt{1}}}};
  \node[draw=none, green!60!black, xshift=-4mm, yshift=0mm] at (ACLI_1.west) {\small{
  \ensuremath{\mathtt{2941}}}};
  \draw[solid] (ACLI_2) edge (ACLI_1);
  \node[draw=red, OR, xshift=-0.000000cm , yshift=5mm, below = 1.4cm of ACLI_1.south]
  (ACLI') {\rotatebox {-90}{\ensuremath{\mathtt{ACLI'}}}};
  \node[draw=none, blue, xshift=2mm, yshift=0mm] at (ACLI'.south) {\small{\ensuremath{\mathtt{2}}}};
  \node[draw=none, green!60!black, xshift=-4mm, yshift=0mm] at (ACLI'.north) {\small{
  \ensuremath{\mathtt{2940}}}};
  \draw[solid] (ACLI_1) edge (ACLI'.east);
  \node[draw=red, SEQ, xshift=-0.000000cm , below of=ACLI'] (ECCS_60) {\ensuremath{\mathtt{ECCS_{60}}}};
  \node[draw=none, blue, xshift=2mm, yshift=0mm] at (ECCS_60.east) {\small{\ensuremath{\mathtt{2}}}};
  \node[draw=none, green!60!black, xshift=-4mm, yshift=0mm] at (ECCS_60.west) {\small{
  \ensuremath{\mathtt{2940}}}};
  \draw[solid] (ACLI'.west) edge (ECCS_60);
  \node[draw=red, SEQ, xshift=-0.000000cm , below of=ECCS_60] (ECCS_1) {\ensuremath{
  \mathtt{ECCS_{1}}}};
  \node[draw=none, blue, xshift=2mm, yshift=0mm] at (ECCS_1.east) {\small{\ensuremath{\mathtt{61}}}};
  \node[draw=none, green!60!black, xshift=-4mm, yshift=0mm] at (ECCS_1.west) {\small{
  \ensuremath{\mathtt{2881}}}};
  \draw[dotted, red, thick] (ECCS_60) edge (ECCS_1);
  \node[draw=red, NULL, xshift=-0.000000cm , below of=ECCS_1] (ECCS') {\ensuremath{\mathtt{ECCS'}}};
  \node[draw=none, blue, xshift=2mm, yshift=0mm] at (ECCS'.east) {\small{\ensuremath{\mathtt{62}}}};
  \node[draw=none, green!60!black, xshift=-4mm, yshift=0mm] at (ECCS'.west) {\small{
  \ensuremath{\mathtt{2880}}}};
  \draw[solid] (ECCS_1) edge (ECCS');
  \node[draw=red, OR, xshift=-0.000000cm , yshift=4mm, below = 1.4cm of ECCS'.south] (ECC') {\rotatebox {-90}{\ensuremath{\mathtt{ECC'}}}};
  \node[draw=none, blue, xshift=2mm, yshift=0mm] at (ECC'.south) {\small{\ensuremath{\mathtt{62}}}};
  \node[draw=none, green!60!black, xshift=-4mm, yshift=0mm] at (ECC'.north) {\small{
  \ensuremath{\mathtt{2880}}}};
  \draw[solid] (ECCS') edge (ECC'.east);
  \node[draw=red, SEQ, xshift=-0.000000cm , below of=ECC'] (bcc_2880) {\ensuremath{\mathtt{bcc_{2880}}}};
  \node[draw=none, blue, xshift=2mm, yshift=0mm] at (bcc_2880.east) {\small{\ensuremath{\mathtt{62}}}};
  \node[draw=none, green!60!black, xshift=-4mm, yshift=0mm] at (bcc_2880.west) {\small{
  \ensuremath{\mathtt{2880}}}};
  \draw[solid] (ECC'.west) edge (bcc_2880);
  \node[draw=red, SEQ, xshift=-0.000000cm , below of=bcc_2880] (bcc_1) {\ensuremath{
  \mathtt{bcc_{1}}}};
  \node[draw=none, blue, xshift=4mm, yshift=0mm] at (bcc_1.east) {\small{\ensuremath{\mathtt{2941}}}};
  \node[draw=none, green!60!black, xshift=-2mm, yshift=0mm] at (bcc_1.west) {\small{\ensuremath{\mathtt{1}}}};
  \draw[dotted, red, thick] (bcc_2880) edge (bcc_1);
  \node[draw=red, state, xshift=-0.000000cm , below of=bcc_1] (bcc') {\ensuremath{\mathtt{bcc'}}};
  \node[draw=none, blue, xshift=4mm, yshift=0mm] at (bcc'.east) {\small{\ensuremath{\mathtt{2942}}}};
  \node[draw=none, green!60!black, xshift=-2mm, yshift=0mm] at (bcc'.west) {\small{\ensuremath{\mathtt{0}}}};
  \draw[solid] (bcc_1) edge (bcc');
      \end{tikzpicture}
    }

    \subcaption{Case \leaf{scr} is not active}
  \end{subfigure}
  \begin{subfigure}[b]{0.3\linewidth}
  \centering
    \scalebox{0.5}{
      \begin{tikzpicture}[node distance=1.8cm]
        \tikzstyle{SEQ}=[diamond]
        \tikzstyle{NULL}=[trapezium, trapezium left angle=120, trapezium right angle=120, minimum size=8mm]
        \tikzstyle{AND}=[and gate US, logic gate inputs=nn, rotate=90 ]
        \tikzstyle{OR}=[or gate US, logic gate inputs=nn, rotate=90 ]
        \tikzset{every node/.style={ultra thick, draw=red, minimum size=6mm}}
        \normalsize
  \node[draw=red, OR, xshift=0.000000cm ] (OAP') {\rotatebox {-90}{\ensuremath{\mathtt{OAP'}}}};
  \node[draw=none, blue, xshift=2mm, yshift=7mm] at (OAP'.south) {\small{\ensuremath{\mathtt{level}}}};
  \node[draw=none, green!60!black, xshift=-2mm, yshift=7mm] at (OAP'.north) {\small{\ensuremath{\mathtt{depth}}}};
  \node[draw=none, blue, xshift=2mm, yshift=0mm] at (OAP'.south) {\small{\ensuremath{\mathtt{0}}}};
  \node[draw=none, green!60!black, xshift=-4mm, yshift=0mm] at (OAP'.north) {\small{\ensuremath{\mathtt{4320}}}};
  \node[draw=none, below of =OAP'] (fakeOAP') {};
  \node at (fakeOAP'.center) [draw=red, OR, logic gate inputs=nn]
  (GSAP') {\rotatebox {-90}{\ensuremath{\mathtt{GSAP'}}}};
  \node[draw=none, blue, xshift=2mm, yshift=0mm] at (GSAP'.south) {\small{\ensuremath{\mathtt{0}}}};
  \node[draw=none, green!60!black, xshift=-4mm, yshift=0mm] at (GSAP'.north) {\small{\ensuremath{\mathtt{4320}}}};
  \draw[solid] (OAP'.west) edge (GSAP'.east);
  \node[draw=red, NULL, xshift=-0.000000cm , below of=GSAP'] (TSA') {\ensuremath{\mathtt{TSA'}}};
  \node[draw=none, blue, xshift=2mm, yshift=0mm] at (TSA'.east) {\small{\ensuremath{\mathtt{0}}}};
  \node[draw=none, green!60!black, xshift=-4mm, yshift=0mm] at (TSA'.west) {\small{
  \ensuremath{\mathtt{4320}}}};
  \draw[solid] (GSAP'.west) edge (TSA');
  \node[draw=red, SEQ, xshift=-0.000000cm , below of=TSA'] (th_4320) {\ensuremath{\mathtt{th_{4320}}}};
  \node[draw=none, blue, xshift=2mm, yshift=0mm] at (th_4320.east) {\small{\ensuremath{\mathtt{0}}}};
  \node[draw=none, green!60!black, xshift=-4mm, yshift=0mm] at (th_4320.west) {\small{
  \ensuremath{\mathtt{4320}}}};
  \draw[solid] (TSA') edge (th_4320);
  \node[draw=red, SEQ, xshift=-0.000000cm , below of=th_4320] (th_1) {\ensuremath{
  \mathtt{th_{1}}}};
  \node[draw=none, blue, xshift=4mm, yshift=0mm] at (th_1.east) {\small{\ensuremath{\mathtt{4319}}}};
  \node[draw=none, green!60!black, xshift=-2mm, yshift=0mm] at (th_1.west) {\small{\ensuremath{\mathtt{1}}}};
  \draw[dotted, red, thick] (th_4320) edge (th_1);
  \node[draw=red, state, xshift=-0.000000cm , below of=th_1] (th') {\ensuremath{\mathtt{th'}}};
  \node[draw=none, blue, xshift=4mm, yshift=0mm] at (th'.east) {\small{\ensuremath{\mathtt{4320}}}};
  \node[draw=none, green!60!black, xshift=-2mm, yshift=0mm] at (th'.west) {\small{\ensuremath{\mathtt{0}}}};
  \draw[solid] (th_1) edge (th');
      \end{tikzpicture}
    }
    \subcaption{Case \leaf{scr} active, but not \gate{DTH}}
  \end{subfigure}
  \begin{subfigure}[b]{0.3\linewidth}
  \centering
    \scalebox{0.5}{
      \begin{tikzpicture}[node distance=1.8cm]
        \tikzstyle{SEQ}=[diamond]
        \tikzstyle{NULL}=[trapezium, trapezium left angle=120, trapezium right angle=120, minimum size=8mm]
        \tikzstyle{AND}=[and gate US, logic gate inputs=nn, rotate=90 ]
        \tikzstyle{OR}=[or gate US, logic gate inputs=nn, rotate=90 ]
        \tikzset{every node/.style={ultra thick, draw=red, minimum size=6mm}}
        \normalsize
  \node[draw=red, OR, xshift=0.000000cm ] (OAP') {\rotatebox {-90}{\ensuremath{\mathtt{OAP'}}}};
  \node[draw=none, blue, xshift=2mm, yshift=7mm] at (OAP'.south) {\small{\ensuremath{\mathtt{level}}}};
  \node[draw=none, green!60!black, xshift=-2mm, yshift=7mm] at (OAP'.north) {\small{\ensuremath{\mathtt{depth}}}};
  \node[draw=none, blue, xshift=2mm, yshift=0mm] at (OAP'.south) {\small{\ensuremath{\mathtt{0}}}};
  \node[draw=none, green!60!black, xshift=-4mm, yshift=0mm] at (OAP'.north) {\small{
  \ensuremath{\mathtt{5762}}}};
  \node[draw=red, SEQ, xshift=-0.000000cm , below of=OAP'] (ACLI_2) {\ensuremath{\mathtt{ACLI_{2}}}};
  \node[draw=none, blue, xshift=2mm, yshift=0mm] at (ACLI_2.east) {\small{\ensuremath{\mathtt{0}}}};
  \node[draw=none, green!60!black, xshift=-4mm, yshift=0mm] at (ACLI_2.west) {\small{
  \ensuremath{\mathtt{5762}}}};
  \draw[solid] (OAP'.west) edge (ACLI_2);
  \node[draw=red, SEQ, xshift=-0.000000cm , below of=ACLI_2] (ACLI_1) {\ensuremath{\mathtt{ACLI_{1}}}};
  \node[draw=none, blue, xshift=2mm, yshift=0mm] at (ACLI_1.east) {\small{\ensuremath{\mathtt{1}}}};
  \node[draw=none, green!60!black, xshift=-4mm, yshift=0mm] at (ACLI_1.west) {\small{
  \ensuremath{\mathtt{5761}}}};
  \draw[solid] (ACLI_2) edge (ACLI_1);
  \node[draw=red, OR, xshift=-0.000000cm , yshift=5mm, below = 1.4cm of ACLI_1.south]
  (ACLI') {\rotatebox {-90}{\ensuremath{\mathtt{ACLI'}}}};
  \node[draw=none, blue, xshift=2mm, yshift=0mm] at (ACLI'.south) {\small{\ensuremath{\mathtt{2}}}};
  \node[draw=none, green!60!black, xshift=-4mm, yshift=0mm] at (ACLI'.north) {\small{
  \ensuremath{\mathtt{5760}}}};
  \draw[solid] (ACLI_1) edge (ACLI'.east);
  \node[draw=red, SEQ, xshift=-0.000000cm , below of=ACLI'] (co_5760) {\ensuremath{\mathtt{co_{5760}}}};
  \node[draw=none, blue, xshift=2mm, yshift=0mm] at (co_5760.east) {\small{\ensuremath{\mathtt{2}}}};
  \node[draw=none, green!60!black, xshift=-4mm, yshift=0mm] at (co_5760.west) {\small{
  \ensuremath{\mathtt{5760}}}};
  \draw[solid] (ACLI'.west) edge (co_5760);
  \node[draw=red, SEQ, xshift=-0.000000cm , below of=co_5760] (co_1) {\ensuremath{
  \mathtt{co_{1}}}};
  \node[draw=none, blue, xshift=4mm, yshift=0mm] at (co_1.east) {\small{\ensuremath{\mathtt{5761}}}};
  \node[draw=none, green!60!black, xshift=-2mm, yshift=0mm] at (co_1.west) {\small{\ensuremath{\mathtt{1}}}};
  \draw[dotted, red, thick] (co_5760) edge (co_1);
  \node[draw=red, state, xshift=-0.000000cm , below of=co_1] (co') {\ensuremath{\mathtt{co'}}};
  \node[draw=none, blue, xshift=4mm, yshift=0mm] at (co'.east) {\small{\ensuremath{\mathtt{5762}}}};
  \node[draw=none, green!60!black, xshift=-2mm, yshift=0mm] at (co'.west) {\small{\ensuremath{\mathtt{0}}}};
  \draw[solid] (co_1) edge (co');
      \end{tikzpicture}
    }
    \subcaption{Case both \leaf{scr} and \gate{DTH} are active}
  \end{subfigure}

  \caption{Preprocessing the \texttt{gain-admin} \ADT/\label{fig:pre:gainadmin}.}
\end{figure}

\begin{table}[!!htb]
\centering
    \caption{Assignment for \texttt{gain-admin}}
\rowcolors{2}{lightgray!30}{white}
\begin{subtable}[t]{\linewidth}
\subcaption{\DAG/ (a)}
\centering
\begin{tabular}{c|l|}
\diagbox[]{slot}{agent}&1\\
\hline
1&{\ensuremath{\mathtt{bcc',bcc_{1}}}}\\
2&{\ensuremath{\mathtt{bcc_{2}}}}\\
$\cdots$ & $\cdots$ \\
2879&{\ensuremath{\mathtt{bcc_{2879}}}}\\
2880&{\ensuremath{\mathtt{ECC',bcc_{2880}}}}\\
2881&{\ensuremath{\mathtt{ECCS',ECCS_{1}}}}\\
2882&{\ensuremath{\mathtt{ECCS_{2}}}}\\
$\cdots$ & $\cdots$ \\
2940&{\ensuremath{\mathtt{ECCS_{60}}}}\\
2941&{\ensuremath{\mathtt{ACLI',ACLI_{1}}}}\\
2942&{\ensuremath{\mathtt{ACLI_{2,OAP'}}}}\\
\hline
\end{tabular}
\end{subtable}

\begin{subtable}[t]{\linewidth}
\subcaption{\DAG/ (b)}
\centering
\begin{tabular}{c|l|}
\diagbox[]{slot}{agent}&1\\
\hline
1&{\ensuremath{\mathtt{th',th_{1}}}}\\
2&{\ensuremath{\mathtt{th_{2}}}}\\
$\cdots$ & $\cdots$ \\
4319&{\ensuremath{\mathtt{th_{4319}}}}\\
4320&{\ensuremath{\mathtt{GSAP',OAP',TSA',th_{4320}}}}\\
\hline
\end{tabular}
\end{subtable}

\begin{subtable}[t]{\linewidth}
\subcaption{\DAG/ (c)}
\centering
\begin{tabular}{c|l|}
\diagbox[]{slot}{agent}&1\\
\hline
1&{\ensuremath{\mathtt{co',co_{1}}}}\\
2&{\ensuremath{\mathtt{co_{2}}}}\\
$\cdots$ & $\cdots$ \\
5760&{\ensuremath{\mathtt{co_{5760}}}}\\
5761&{\ensuremath{\mathtt{ACLI',ACLI_{1}}}}\\
5762&{\ensuremath{\mathtt{ACLI_{2,OAP'}}}}\\
\hline
\end{tabular}
\end{subtable}
\end{table}

%% file: examples/ADTree_gainadmin.tikz

\begin{tikzpicture}[
	every node/.style={ultra thick,draw=red,minimum size=6mm},
	node distance=1.5cm]

	\node[or gate US,point up,logic gate inputs=nn] (OAP)
		{\rotatebox{-90}{\gate{OAP}}};

	\node[or gate US,point up,logic gate inputs=nn,
		below = 13mm of OAP.input 1,yshift=14mm] (ACLI)
		{\rotatebox{-90}{\gate{ACLI}}};
	\draw (ACLI.east) -- ([yshift=0.15cm]ACLI.east) -| (OAP.input 1);

	\node[state, below = 9mm of ACLI.input 1, xshift=-8mm]
		(co) {\leaf{co}};
	\draw (co.north) -- ([yshift=0.15cm]co.north) -| (ACLI.input 1);

	\node[and gate US,point up,logic gate inputs=ni,
		below = 21.5mm of ACLI, yshift = 7.8mm] (ECCS)
		{\rotatebox{-90}{\gate{ECCS}}};
	\draw (ECCS.east) -- (ACLI.input 2);

	\node[or gate US,point up,logic gate inputs=nn,
		below = 12mm of ECCS.input 1,yshift=14mm] (ECC)
		{\rotatebox{-90}{\gate{ECC}}};
	\draw (ECC.east) -- ([yshift=0.15cm]ECC.east) -| (ECCS.input 1);

	\node[state, below = 5mm of ECC.input 1, xshift=-4.5mm]
		(bcc) {\leaf{bcc}};
	\draw (bcc.north) -- ([yshift=0.15cm]bcc.north) -| (ECC.input 1);

	\node[state, below = 5mm of ECC.input 2, xshift=4.5mm]
		(ccg) {\leaf{ccg}};
	\draw (ccg.north) -- ([yshift=0.15cm]ccg.north) -| (ECC.input 2);

	\node[rectangle,draw=Green,minimum size=8mm,
		below = 5.5mm of ECCS.west, xshift=1.65mm]
		(scr) {\leaf{scr}};
	\draw (scr.north) -- (ECCS.input 2);

	\node[or gate US,point up,logic gate inputs=nnnn,
		below = 14.5mm of OAP.input 2,yshift=-3mm] (GSAP)
		{\rotatebox{-90}{\gate{GSAP}}};
	\draw (GSAP.east) -- ([yshift=0.15cm]GSAP.east) -| (OAP.input 2);

	\node[and gate US,point up,logic gate inputs=ni, seq=12pt,
		below = 14mm of GSAP.west, yshift = 8.5mm] (GAPS)
		{\rotatebox{-90}{\gate{GAPS}}};
	\draw (GAPS.east) -- (GSAP.input 1);

	\node[and gate US,point up,logic gate inputs=nn, seq=6pt,
		below = 11mm of GAPS.input 1, yshift = 4.2mm] (GAP)
		{\rotatebox{-90}{\gate{GAP}}};
	\draw (GAP.east) -- (GAPS.input 1);

	\node[state, below = 5mm of GAP.input 1, xshift=-4.5mm]
		(opf) {\leaf{opf}};
	\draw (opf.north) -- ([yshift=0.15cm]opf.north) -| (GAP.input 1);

	\node[state, below = 5mm of GAP.input 2, xshift=4.5mm]
		(fgp) {\leaf{fgp}};
	\draw (fgp.north) -- ([yshift=0.15cm]fgp.north) -| (GAP.input 2);

	\node[rectangle,draw=Green,minimum size=8mm,
		below = 4mm of GAPS.input 2, xshift = 8mm] (tla)
		{\leaf{tla}};
	\draw (tla.north) -- ([yshift=0.15cm]tla.north) -| (GAPS.input 2);

	\node[and gate US,point up,logic gate inputs=ni,
		below = 14mm of GSAP.west, yshift = -20mm] (LSAS)
		{\rotatebox{-90}{\gate{LSAS}}};
	\draw (LSAS.east) -- ([yshift=0.12cm]LSAS.east) -| (GSAP.input 2);

	\node[and gate US,point up,logic gate inputs=nnn, seq=6pt,
		below = 12mm of LSAS.input 1, yshift = 5.3mm] (LSA)
		{\rotatebox{-90}{\gate{LSA}}};
	\draw (LSA.east) -- (LSAS.input 1);

	\node[state, below = 3mm of LSA.input 1, xshift=-7.5mm]
		(bsa) {\leaf{bsa}};
	\draw (bsa.north) -- ([yshift=0.15cm]bsa.north) -| (LSA.input 1);

	\node[state, below = 3mm of LSA.input 2]
		(vsa) {\leaf{vsa}};
	\draw (vsa.north) -- (LSA.input 2);

	\node[state, below = 3mm of LSA.input 3, xshift=7.5mm]
		(sat) {\leaf{sat}};
	\draw (sat.north) -- ([yshift=0.15cm]sat.north) -| (LSA.input 3);

	\node[rectangle,draw=Green,minimum size=8mm,
		below = 4mm of LSAS.input 2, xshift = 7.5mm] (nv)
		{\leaf{nv}};
	\draw (nv.north) -- ([yshift=0.15cm]nv.north) -| (LSAS.input 2);

	\node[and gate US,point up,logic gate inputs=ni,
		below = 13.6mm of GSAP.west, yshift = -48.5mm] (TSA)
		{\rotatebox{-90}{\gate{TSA}}};
	\draw (TSA.east) -- ([yshift=0.3cm]TSA.east) -| (GSAP.input 3);

	\node[state, below = 5.4mm of TSA.input 2, xshift = -10mm]
		(th) {\leaf{th}};
	\draw (th.north) -- ([yshift=0.15cm]th.north) -| (TSA.input 1);

	\node[or gate US,point up,logic gate inputs=nn,draw=Green,
		below = 12.4mm of TSA.input 2,yshift=4.25mm] (DTH)
		{\rotatebox{-90}{\gate{DTH}}};
	\draw (DTH.east) --(TSA.input 2);

	\node[rectangle,draw=Green,minimum size=8mm,
		below = 5mm of DTH.input 1, xshift = -6mm] (wd)
		{\leaf{wd}};
	\draw (wd.north) -- ([yshift=0.15cm]wd.north) -| (DTH.input 1);

	\node[rectangle,draw=Green,minimum size=8mm,
		below = 5mm of DTH.input 2, xshift = 6mm] (efw)
		{\leaf{efw}};
	\draw (efw.north) -- ([yshift=0.15cm]efw.north) -| (DTH.input 2);

	\node[state, below = 7mm of GSAP.input 4, xshift=60mm]
		(csa) {\leaf{csa}};
	\draw (csa.north) -- ([yshift=0.375cm]csa.north) -| (GSAP.input 4);

\end{tikzpicture}

%% file: examples/interrupted.tex

\begin{figure}[ht]
  \centering
  \begin{tikzpicture}
    \node at (0,0)   {\scalebox{.5}{\input{examples/ADTree_interrupted.tikz}}};
    \node at (4.5,0) {\scalebox{.555}{\input{examples/tabAttributes_interrupted}}};
  \end{tikzpicture}
  \caption{Interrupted schedule example (\texttt{interrupted})}
  \label{fig:interrupted}
\end{figure}

\begin{figure}[!htb]
  \centering
    \scalebox{0.55}{
\begin{tikzpicture}[node distance=1.8cm]
  \tikzstyle{SEQ}=[diamond]
  \tikzstyle{NULL}=[trapezium, trapezium left angle=120, trapezium right angle=120, minimum size=8mm]
  \tikzstyle{AND}=[and gate US, rotate=90 ]
  \tikzstyle{OR}=[or gate US, rotate=90 ]
  \tikzset{every node/.style={ultra thick, draw=red, minimum size=6mm}}
  \node[draw=red, AND, logic gate inputs=nn, xshift=0.000000cm ] (a') {\rotatebox {-90}{\ensuremath{\mathtt{a'}}}};
  \node[draw=none, blue, xshift=2mm, yshift=7mm] at (a'.south) {\small{\ensuremath{\mathtt{level}}}};
  \node[draw=none, green!60!black, xshift=-2mm, yshift=7mm] at (a'.north) {\small{\ensuremath{\mathtt{depth}}}};
  \node[draw=none, blue, xshift=2mm, yshift=0mm] at (a'.south) {\small{\ensuremath{\mathtt{0}}}};
  \node[draw=none, green!60!black, xshift=-2mm, yshift=0mm] at (a'.north) {\small{\ensuremath{\mathtt{5}}}};
  \node[draw=red, SEQ, xshift=-1.25000cm , below of=a'] (b_2) {\ensuremath{\mathtt{b_
  {2}}}};
  \node[draw=none, blue, xshift=2mm, yshift=0mm] at (b_2.east) {\small{\ensuremath{\mathtt{0}}}};
  \node[draw=none, green!60!black, xshift=-2mm, yshift=0mm] at (b_2.west) {\small{\ensuremath{\mathtt{2}}}};
  \draw[solid] (a'.input 1) edge (b_2);
  \node[draw=red, SEQ, xshift=1.250000cm , below of=a'] (c_1) {\ensuremath{\mathtt{c_
  {1}}}};
  \node[draw=none, blue, xshift=2mm, yshift=0mm] at (c_1.east) {\small{\ensuremath{\mathtt{0}}}};
  \node[draw=none, green!60!black, xshift=-2mm, yshift=0mm] at (c_1.west) {\small{\ensuremath{\mathtt{5}}}};
  \draw[solid] (a'.input 2) edge (c_1);
  \node[draw=red, AND, logic gate inputs=nn, xshift=-0.000000cm , yshift=4mm, below = 1.4cm of c_1.south] (c') {\rotatebox {-90}{\ensuremath{\mathtt{c'}}}};
  \node[draw=none, blue, xshift=2mm, yshift=0mm] at (c'.south) {\small{\ensuremath{\mathtt{1}}}};
  \node[draw=none, green!60!black, xshift=-2mm, yshift=0mm] at (c'.north) {\small{\ensuremath{\mathtt{4}}}};
  \draw[solid] (c_1) edge (c'.east);
  \node[draw=red, SEQ, xshift=-1.000000cm , below of=c'] (d_4) {\ensuremath{\mathtt{d_
  {4}}}};
  \node[draw=none, blue, xshift=2mm, yshift=0mm] at (d_4.east) {\small{\ensuremath{\mathtt{1}}}};
  \node[draw=none, green!60!black, xshift=-2mm, yshift=0mm] at (d_4.west) {\small{\ensuremath{\mathtt{4}}}};
  \draw[solid] (c'.input 1) edge (d_4);
  \node[draw=red, SEQ, xshift=1.000000cm , below of=c'] (e_3) {\ensuremath{\mathtt{e_
  {3}}}};
  \node[draw=none, blue, xshift=2mm, yshift=0mm] at (e_3.east) {\small{\ensuremath{\mathtt{1}}}};
  \node[draw=none, green!60!black, xshift=-2mm, yshift=0mm] at (e_3.west) {\small{\ensuremath{\mathtt{3}}}};
  \draw[solid] (c'.input 2) edge (e_3);
  \node[draw=red, SEQ, xshift=-0.000000cm , below of=e_3] (e_2) {\ensuremath{\mathtt{e_{2}}}};
  \node[draw=none, blue, xshift=2mm, yshift=0mm] at (e_2.east) {\small{\ensuremath{\mathtt{2}}}};
  \node[draw=none, green!60!black, xshift=-2mm, yshift=0mm] at (e_2.west) {\small{\ensuremath{\mathtt{2}}}};
  \draw[solid] (e_3) edge (e_2);
  \node[draw=red, SEQ, xshift=-0.000000cm , below of=e_2] (e_1) {\ensuremath{\mathtt{e_{1}}}};
  \node[draw=none, blue, xshift=2mm, yshift=0mm] at (e_1.east) {\small{\ensuremath{\mathtt{3}}}};
  \node[draw=none, green!60!black, xshift=-2mm, yshift=0mm] at (e_1.west) {\small{\ensuremath{\mathtt{1}}}};
  \draw[solid] (e_2) edge (e_1);
  \node[draw=red, state, xshift=-0.000000cm , below of=e_1] (e') {\ensuremath{\mathtt{e'}}};
  \node[draw=none, blue, xshift=2mm, yshift=0mm] at (e'.east) {\small{\ensuremath{\mathtt{4}}}};
  \node[draw=none, green!60!black, xshift=-2mm, yshift=0mm] at (e'.west) {\small{\ensuremath{\mathtt{0}}}};
  \draw[solid] (e_1) edge (e');
  \node[draw=red, SEQ, xshift=-0.000000cm , below of=d_4] (d_3) {\ensuremath{\mathtt{d_{3}}}};
  \node[draw=none, blue, xshift=2mm, yshift=0mm] at (d_3.east) {\small{\ensuremath{\mathtt{2}}}};
  \node[draw=none, green!60!black, xshift=-2mm, yshift=0mm] at (d_3.west) {\small{\ensuremath{\mathtt{3}}}};
  \draw[solid] (d_4) edge (d_3);
  \node[draw=red, SEQ, xshift=-0.000000cm , below of=d_3] (d_2) {\ensuremath{\mathtt{d_{2}}}};
  \node[draw=none, blue, xshift=2mm, yshift=0mm] at (d_2.east) {\small{\ensuremath{\mathtt{3}}}};
  \node[draw=none, green!60!black, xshift=-2mm, yshift=0mm] at (d_2.west) {\small{\ensuremath{\mathtt{2}}}};
  \draw[solid] (d_3) edge (d_2);
  \node[draw=red, SEQ, xshift=-0.000000cm , below of=d_2] (d_1) {\ensuremath{\mathtt{d_{1}}}};
  \node[draw=none, blue, xshift=2mm, yshift=0mm] at (d_1.east) {\small{\ensuremath{\mathtt{4}}}};
  \node[draw=none, green!60!black, xshift=-2mm, yshift=0mm] at (d_1.west) {\small{\ensuremath{\mathtt{1}}}};
  \draw[solid] (d_2) edge (d_1);
  \node[draw=red, state, xshift=-0.000000cm , below of=d_1] (d') {\ensuremath{\mathtt{d'}}};
  \node[draw=none, blue, xshift=2mm, yshift=0mm] at (d'.east) {\small{\ensuremath{\mathtt{5}}}};
  \node[draw=none, green!60!black, xshift=-2mm, yshift=0mm] at (d'.west) {\small{\ensuremath{\mathtt{0}}}};
  \draw[solid] (d_1) edge (d');
  \node[draw=red, SEQ, xshift=-0.000000cm , below of=b_2] (b_1) {\ensuremath{\mathtt{b_{1}}}};
  \node[draw=none, blue, xshift=2mm, yshift=0mm] at (b_1.east) {\small{\ensuremath{\mathtt{1}}}};
  \node[draw=none, green!60!black, xshift=-2mm, yshift=0mm] at (b_1.west) {\small{\ensuremath{\mathtt{1}}}};
  \draw[solid] (b_2) edge (b_1);
  \node[draw=red, state, xshift=-0.000000cm , below of=b_1] (b') {\ensuremath{\mathtt{b'}}};
  \node[draw=none, blue, xshift=2mm, yshift=0mm] at (b'.east) {\small{\ensuremath{\mathtt{2}}}};
  \node[draw=none, green!60!black, xshift=-2mm, yshift=0mm] at (b'.west) {\small{\ensuremath{\mathtt{0}}}};
  \draw[solid] (b_1) edge (b');
\end{tikzpicture}
    }
    \caption{Preprocessing the \texttt{interrupted} \ADT/}
\end{figure}

\begin{table}[!!htb]
    \centering
    \rowcolors{2}{lightgray!30}{white}

      \begin{tabular}{c|l|l|}
        \diagbox[]{slot}{agent} & 1                         & 2                         \\
        \hline
1&{\ensuremath{\mathtt{d',d_{1}}}}&{\ensuremath{\mathtt{e',e_{1}}}}\\
2&{\ensuremath{\mathtt{d_{2}}}}&{\ensuremath{\mathtt{b',b_{1}}}}\\
3&{\ensuremath{\mathtt{d_{3}}}}&{\ensuremath{\mathtt{e_{2}}}}\\
4&{\ensuremath{\mathtt{d_{4}}}}&{\ensuremath{\mathtt{e_{3}}}}\\
5&{\ensuremath{\mathtt{c',c_{1}}}}&{\ensuremath{\mathtt{a',b_{2}}}}\\
        \hline
      \end{tabular}
    \caption{Assignment for \texttt{interrupted}}
\end{table}

%% file: examples/ADTree_interrupted.tikz

\begin{tikzpicture}
	[every node/.style={ultra thick,draw=red,minimum size=6mm},
	node distance=1.8cm]

	\node[and gate US,point up,logic gate inputs=nn] (a)
		{\rotatebox{-90}{\leaf{a}}};

	\node[and gate US,point up,logic gate inputs=nn,
			below left of = a] (c)
		{\rotatebox{-90}{\leaf{c}}};
	\draw (a.input 2) -- ([yshift=-0.15cm]a.input 2) -| (c.east);

	\node[state, below left of = a] (b) {\leaf{b}};
	\draw (a.input 1) -- ([yshift=-0.15cm]a.input 1) -| (b);

	\node[state, below left of = c] (d) {\leaf{d}};
	\draw (c.input 1) -- ([yshift=-0.15cm]c.input 1) -| (d);

	\node[state, below right of = c] (e) {\leaf{e}};
	\draw (c.input 2) -- ([yshift=-0.15cm]c.input 2) -| (e);
\end{tikzpicture}

%% file: examples/icfem2020_example.tex

\begin{figure}[ht]
  \centering
  \begin{tikzpicture}
    \node at (0,0)   {\scalebox{.5}{\input{examples/ADTree_toy_example.tikz}}};
    \node at (4.5,0) {\scalebox{.555}{\input{examples/tabAttributes_toy_example}}};
  \end{tikzpicture}
  \caption{Scaling example (\texttt{scaling-example})}
  \label{fig:scalingexample}
\end{figure}

\begin{figure}[!htb]
  \centering
    \scalebox{0.55}{
      \begin{tikzpicture}[node distance=1.8cm]
        \tikzstyle{SEQ}=[diamond]
        \tikzstyle{NULL}=[trapezium, trapezium left angle=120, trapezium right angle=120, minimum size=8mm]
        \tikzstyle{AND}=[and gate US, logic gate inputs=nn, rotate=90 ]
        \tikzstyle{OR}=[or gate US, logic gate inputs=nn, rotate=90 ]
        \tikzset{every node/.style={ultra thick, draw=red, minimum size=6mm}}
        \node[draw=red, SEQ, xshift=0.000000cm ] (A7_1) {\ensuremath{\mathtt{A7_{1}}}};
        \node[draw=none, blue, xshift=2mm, yshift=7mm] at (A7_1.east) {\small{\ensuremath{\mathtt{level}}}};\node[draw=none, green!60!black, xshift=-2mm, yshift=7mm] at (A7_1.west) {\small{\ensuremath{\mathtt{depth}}}};\node[draw=none, blue, xshift=2mm, yshift=0mm] at (A7_1.east) {\small{\ensuremath{\mathtt{0}}}};\node[draw=none, green!60!black, xshift=-2mm, yshift=0mm] at (A7_1.west) {\small{\ensuremath{\mathtt{5}}}};  \node[draw=red, AND, xshift=-0.000000cm , yshift=4mm, below = 1.4cm of A7_1.south] (A7') {\rotatebox {-90}{\ensuremath{\mathtt{A7'}}}};
        \node[draw=none, blue, xshift=2mm, yshift=0mm] at (A7'.south) {\small{\ensuremath{\mathtt{1}}}};\node[draw=none, green!60!black, xshift=-2mm, yshift=0mm] at (A7'.north) {\small{\ensuremath{\mathtt{4}}}};  \draw[solid] (A7_1) edge (A7'.east);
        \node[draw=red, SEQ, xshift=-1.250000cm , below of=A7'] (A6_1) {\ensuremath{\mathtt{A6_{1}}}};
        \node[draw=none, blue, xshift=2mm, yshift=0mm] at (A6_1.east) {\small{\ensuremath{\mathtt{1}}}};\node[draw=none, green!60!black, xshift=-2mm, yshift=0mm] at (A6_1.west) {\small{\ensuremath{\mathtt{4}}}};  \draw[solid] (A7'.input 1) edge (A6_1);
        \node[draw=red, SEQ, xshift=1.250000cm , below of=A7'] (l8_1) {\ensuremath{\mathtt{l8_{1}}}};
        \node[draw=none, blue, xshift=2mm, yshift=0mm] at (l8_1.east) {\small{\ensuremath{\mathtt{1}}}};\node[draw=none, green!60!black, xshift=-2mm, yshift=0mm] at (l8_1.west) {\small{\ensuremath{\mathtt{1}}}};  \draw[solid] (A7'.input 2) edge (l8_1);
        \node[draw=red, state, xshift=-0.000000cm , below of=l8_1] (l8') {\ensuremath{\mathtt{l8'}}};
        \node[draw=none, blue, xshift=2mm, yshift=0mm] at (l8'.east) {\small{\ensuremath{\mathtt{2}}}};\node[draw=none, green!60!black, xshift=-2mm, yshift=0mm] at (l8'.west) {\small{\ensuremath{\mathtt{0}}}};  \draw[solid] (l8_1) edge (l8');
        \node[draw=red, AND, xshift=-0.000000cm , yshift=4mm, below = 1.4cm of A6_1.south] (A6') {\rotatebox {-90}{\ensuremath{\mathtt{A6'}}}};
        \node[draw=none, blue, xshift=2mm, yshift=0mm] at (A6'.south) {\small{\ensuremath{\mathtt{2}}}};\node[draw=none, green!60!black, xshift=-2mm, yshift=0mm] at (A6'.north) {\small{\ensuremath{\mathtt{3}}}};  \draw[solid] (A6_1) edge (A6'.east);
        \node[draw=red, SEQ, xshift=-4.500000cm , below of=A6'] (A4_1) {\ensuremath{\mathtt{A4_{1}}}};
        \node[draw=none, blue, xshift=2mm, yshift=0mm] at (A4_1.east) {\small{\ensuremath{\mathtt{2}}}};\node[draw=none, green!60!black, xshift=-2mm, yshift=0mm] at (A4_1.west) {\small{\ensuremath{\mathtt{3}}}};  \draw[solid] (A6'.input 1) edge (A4_1);
        \node[draw=red, SEQ, xshift=4.500000cm , below of=A6'] (A5_1) {\ensuremath{\mathtt{A5_{1}}}};
        \node[draw=none, blue, xshift=2mm, yshift=0mm] at (A5_1.east) {\small{\ensuremath{\mathtt{2}}}};\node[draw=none, green!60!black, xshift=-2mm, yshift=0mm] at (A5_1.west) {\small{\ensuremath{\mathtt{3}}}};  \draw[solid] (A6'.input 2) edge (A5_1);
        \node[draw=red, AND, xshift=-0.000000cm , yshift=4mm, below = 1.4cm of A5_1.south] (A5') {\rotatebox {-90}{\ensuremath{\mathtt{A5'}}}};
        \node[draw=none, blue, xshift=2mm, yshift=0mm] at (A5'.south) {\small{\ensuremath{\mathtt{3}}}};\node[draw=none, green!60!black, xshift=-2mm, yshift=0mm] at (A5'.north) {\small{\ensuremath{\mathtt{2}}}};  \draw[solid] (A5_1) edge (A5'.east);
        \node[draw=red, SEQ, xshift=-1.250000cm , below of=A5'] (A3_1) {\ensuremath{\mathtt{A3_{1}}}};
        \node[draw=none, blue, xshift=2mm, yshift=0mm] at (A3_1.east) {\small{\ensuremath{\mathtt{3}}}};\node[draw=none, green!60!black, xshift=-2mm, yshift=0mm] at (A3_1.west) {\small{\ensuremath{\mathtt{2}}}};  \draw[solid] (A5'.input 1) edge (A3_1);
        \node[draw=red, SEQ, xshift=1.250000cm , below of=A5'] (l7_1) {\ensuremath{\mathtt{l7_{1}}}};
        \node[draw=none, blue, xshift=2mm, yshift=0mm] at (l7_1.east) {\small{\ensuremath{\mathtt{3}}}};\node[draw=none, green!60!black, xshift=-2mm, yshift=0mm] at (l7_1.west) {\small{\ensuremath{\mathtt{1}}}};  \draw[solid] (A5'.input 2) edge (l7_1);
        \node[draw=red, state, xshift=-0.000000cm , below of=l7_1] (l7') {\ensuremath{\mathtt{l7'}}};
        \node[draw=none, blue, xshift=2mm, yshift=0mm] at (l7'.east) {\small{\ensuremath{\mathtt{4}}}};\node[draw=none, green!60!black, xshift=-2mm, yshift=0mm] at (l7'.west) {\small{\ensuremath{\mathtt{0}}}};  \draw[solid] (l7_1) edge (l7');
        \node[draw=red, AND, xshift=-0.000000cm , yshift=4mm, below = 1.4cm of A3_1.south] (A3') {\rotatebox {-90}{\ensuremath{\mathtt{A3'}}}};
        \node[draw=none, blue, xshift=2mm, yshift=0mm] at (A3'.south) {\small{\ensuremath{\mathtt{4}}}};\node[draw=none, green!60!black, xshift=-2mm, yshift=0mm] at (A3'.north) {\small{\ensuremath{\mathtt{1}}}};  \draw[solid] (A3_1) edge (A3'.east);
        \node[draw=red, SEQ, xshift=-1.250000cm , below of=A3'] (l5_1) {\ensuremath{\mathtt{l5_{1}}}};
        \node[draw=none, blue, xshift=2mm, yshift=0mm] at (l5_1.east) {\small{\ensuremath{\mathtt{4}}}};\node[draw=none, green!60!black, xshift=-2mm, yshift=0mm] at (l5_1.west) {\small{\ensuremath{\mathtt{1}}}};  \draw[solid] (A3'.input 1) edge (l5_1);
        \node[draw=red, SEQ, xshift=1.250000cm , below of=A3'] (l6_1) {\ensuremath{\mathtt{l6_{1}}}};
        \node[draw=none, blue, xshift=2mm, yshift=0mm] at (l6_1.east) {\small{\ensuremath{\mathtt{4}}}};\node[draw=none, green!60!black, xshift=-2mm, yshift=0mm] at (l6_1.west) {\small{\ensuremath{\mathtt{1}}}};  \draw[solid] (A3'.input 2) edge (l6_1);
        \node[draw=red, state, xshift=-0.000000cm , below of=l6_1] (l6') {\ensuremath{\mathtt{l6'}}};
        \node[draw=none, blue, xshift=2mm, yshift=0mm] at (l6'.east) {\small{\ensuremath{\mathtt{5}}}};\node[draw=none, green!60!black, xshift=-2mm, yshift=0mm] at (l6'.west) {\small{\ensuremath{\mathtt{0}}}};  \draw[solid] (l6_1) edge (l6');
        \node[draw=red, state, xshift=-0.000000cm , below of=l5_1] (l5') {\ensuremath{\mathtt{l5'}}};
        \node[draw=none, blue, xshift=2mm, yshift=0mm] at (l5'.east) {\small{\ensuremath{\mathtt{5}}}};\node[draw=none, green!60!black, xshift=-2mm, yshift=0mm] at (l5'.west) {\small{\ensuremath{\mathtt{0}}}};  \draw[solid] (l5_1) edge (l5');
        \node[draw=red, AND, xshift=-0.000000cm , yshift=4mm, below = 1.4cm of A4_1.south] (A4') {\rotatebox {-90}{\ensuremath{\mathtt{A4'}}}};
        \node[draw=none, blue, xshift=2mm, yshift=0mm] at (A4'.south) {\small{\ensuremath{\mathtt{3}}}};\node[draw=none, green!60!black, xshift=-2mm, yshift=0mm] at (A4'.north) {\small{\ensuremath{\mathtt{2}}}};  \draw[solid] (A4_1) edge (A4'.east);
        \node[draw=red, SEQ, xshift=-2.500000cm , below of=A4'] (A1_1) {\ensuremath{\mathtt{A1_{1}}}};
        \node[draw=none, blue, xshift=2mm, yshift=0mm] at (A1_1.east) {\small{\ensuremath{\mathtt{3}}}};\node[draw=none, green!60!black, xshift=-2mm, yshift=0mm] at (A1_1.west) {\small{\ensuremath{\mathtt{2}}}};  \draw[solid] (A4'.input 1) edge (A1_1);
        \node[draw=red, SEQ, xshift=2.500000cm , below of=A4'] (A2_1) {\ensuremath{\mathtt{A2_{1}}}};
        \node[draw=none, blue, xshift=2mm, yshift=0mm] at (A2_1.east) {\small{\ensuremath{\mathtt{3}}}};\node[draw=none, green!60!black, xshift=-2mm, yshift=0mm] at (A2_1.west) {\small{\ensuremath{\mathtt{2}}}};  \draw[solid] (A4'.input 2) edge (A2_1);
        \node[draw=red, AND, xshift=-0.000000cm , yshift=4mm, below = 1.4cm of A2_1.south] (A2') {\rotatebox {-90}{\ensuremath{\mathtt{A2'}}}};
        \node[draw=none, blue, xshift=2mm, yshift=0mm] at (A2'.south) {\small{\ensuremath{\mathtt{4}}}};\node[draw=none, green!60!black, xshift=-2mm, yshift=0mm] at (A2'.north) {\small{\ensuremath{\mathtt{1}}}};  \draw[solid] (A2_1) edge (A2'.east);
        \node[draw=red, SEQ, xshift=-1.250000cm , below of=A2'] (l3_1) {\ensuremath{\mathtt{l3_{1}}}};
        \node[draw=none, blue, xshift=2mm, yshift=0mm] at (l3_1.east) {\small{\ensuremath{\mathtt{4}}}};\node[draw=none, green!60!black, xshift=-2mm, yshift=0mm] at (l3_1.west) {\small{\ensuremath{\mathtt{1}}}};  \draw[solid] (A2'.input 1) edge (l3_1);
        \node[draw=red, SEQ, xshift=1.250000cm , below of=A2'] (l4_1) {\ensuremath{\mathtt{l4_{1}}}};
        \node[draw=none, blue, xshift=2mm, yshift=0mm] at (l4_1.east) {\small{\ensuremath{\mathtt{4}}}};\node[draw=none, green!60!black, xshift=-2mm, yshift=0mm] at (l4_1.west) {\small{\ensuremath{\mathtt{1}}}};  \draw[solid] (A2'.input 2) edge (l4_1);
        \node[draw=red, state, xshift=-0.000000cm , below of=l4_1] (l4') {\ensuremath{\mathtt{l4'}}};
        \node[draw=none, blue, xshift=2mm, yshift=0mm] at (l4'.east) {\small{\ensuremath{\mathtt{5}}}};\node[draw=none, green!60!black, xshift=-2mm, yshift=0mm] at (l4'.west) {\small{\ensuremath{\mathtt{0}}}};  \draw[solid] (l4_1) edge (l4');
        \node[draw=red, state, xshift=-0.000000cm , below of=l3_1] (l3') {\ensuremath{\mathtt{l3'}}};
        \node[draw=none, blue, xshift=2mm, yshift=0mm] at (l3'.east) {\small{\ensuremath{\mathtt{5}}}};\node[draw=none, green!60!black, xshift=-2mm, yshift=0mm] at (l3'.west) {\small{\ensuremath{\mathtt{0}}}};  \draw[solid] (l3_1) edge (l3');
        \node[draw=red, AND, xshift=-0.000000cm , yshift=4mm, below = 1.4cm of A1_1.south] (A1') {\rotatebox {-90}{\ensuremath{\mathtt{A1'}}}};
        \node[draw=none, blue, xshift=2mm, yshift=0mm] at (A1'.south) {\small{\ensuremath{\mathtt{4}}}};\node[draw=none, green!60!black, xshift=-2mm, yshift=0mm] at (A1'.north) {\small{\ensuremath{\mathtt{1}}}};  \draw[solid] (A1_1) edge (A1'.east);
        \node[draw=red, SEQ, xshift=-1.250000cm , below of=A1'] (l1_1) {\ensuremath{\mathtt{l1_{1}}}};
        \node[draw=none, blue, xshift=2mm, yshift=0mm] at (l1_1.east) {\small{\ensuremath{\mathtt{4}}}};\node[draw=none, green!60!black, xshift=-2mm, yshift=0mm] at (l1_1.west) {\small{\ensuremath{\mathtt{1}}}};  \draw[solid] (A1'.input 1) edge (l1_1);
        \node[draw=red, SEQ, xshift=1.250000cm , below of=A1'] (l2_1) {\ensuremath{\mathtt{l2_{1}}}};
        \node[draw=none, blue, xshift=2mm, yshift=0mm] at (l2_1.east) {\small{\ensuremath{\mathtt{4}}}};\node[draw=none, green!60!black, xshift=-2mm, yshift=0mm] at (l2_1.west) {\small{\ensuremath{\mathtt{1}}}};  \draw[solid] (A1'.input 2) edge (l2_1);
        \node[draw=red, state, xshift=-0.000000cm , below of=l2_1] (l2') {\ensuremath{\mathtt{l2'}}};
        \node[draw=none, blue, xshift=2mm, yshift=0mm] at (l2'.east) {\small{\ensuremath{\mathtt{5}}}};\node[draw=none, green!60!black, xshift=-2mm, yshift=0mm] at (l2'.west) {\small{\ensuremath{\mathtt{0}}}};  \draw[solid] (l2_1) edge (l2');
        \node[draw=red, state, xshift=-0.000000cm , below of=l1_1] (l1') {\ensuremath{\mathtt{l1'}}};
        \node[draw=none, blue, xshift=2mm, yshift=0mm] at (l1'.east) {\small{\ensuremath{\mathtt{5}}}};\node[draw=none, green!60!black, xshift=-2mm, yshift=0mm] at (l1'.west) {\small{\ensuremath{\mathtt{0}}}};  \draw[solid] (l1_1) edge (l1');
      \end{tikzpicture}
      }
    \caption{Preprocessing the \texttt{scaling-example} \ADT/}
  \end{figure}

\begin{table}[!!htb]
    \centering
    \rowcolors{2}{lightgray!30}{white}
  \hspace{-0.26\linewidth}      \begin{tabular}{c|l|l|l|l|l|l|}
        \diagbox[]{slot}{agent} & 1                         & 2                         & 3                         & 4                         & 5                         & 6                         \\
        \hline
        1                       & \leaf{l1'}, \leaf{l1_{1}} & \leaf{l2'}, \leaf{l2_{1}} & \leaf{l3'}, \leaf{l3_{1}} & \leaf{l4'}, \leaf{l4_{1}} & \leaf{l5'}, \leaf{l5_{1}} & \leaf{l6'}, \leaf{l6_{1}} \\
        2                       & \leaf{A1'}, \leaf{A1_{1}} & \leaf{A2'}, \leaf{A2_{1}} & \leaf{A3'}, \leaf{A3_{1}} & \leaf{l7'}, \leaf{l7_{1}} &                           &                           \\
        3                       & \leaf{A4'}, \leaf{A4_{1}} & \leaf{A5'}, \leaf{A5_{1}} &                           &                           &                           &                           \\
        4                       & \leaf{A6'}, \leaf{A6_{1}} & \leaf{l8'}, \leaf{l8_{1}} &                           &                           &                           &                           \\
        5                       & \leaf{A7'}, \leaf{A7_{1}} &                           &                           &                           &                           &                           \\
        \hline
      \end{tabular}

    \caption{Assignment for \texttt{scaling-example}}
  \end{table}

%% file: examples/ADTree_toy_example.tikz

\begin{tikzpicture}
	[every node/.style={ultra thick,draw=red,minimum size=6mm},
	node distance=1.8cm]

	\node[and gate US,point up,logic gate inputs=nn] (A7)
		{\rotatebox{-90}{\gate{A7}}};

	\node[and gate US,point up,logic gate inputs=nn,
			above left of = A7,yshift=2.8cm] (A6)
		{\rotatebox{-90}{\gate{A6}}};
	\draw (A7.input 1) -- ([yshift=-0.15cm]A7.input 1) -| (A6.east);

	\node[and gate US,point up,logic gate inputs=nn,
			below left of = A6, yshift=-4cm] (A5)
		{\rotatebox{-90}{\gate{A5}}};
	\draw (A6.input 2) -- ([yshift=-0.15cm]A6.input 2) -| (A5.east);

	\node[and gate US,point up,logic gate inputs=nn,
			above left of = A6] (A4)
		{\rotatebox{-90}{\gate{A4}}};
	\draw (A6.input 1) -- ([yshift=-0.15cm]A6.input 1) -| (A4.east);

	\node[and gate US,point up,logic gate inputs=nn,
			above left of = A5] (A3)
		{\rotatebox{-90}{\gate{A3}}};
	\draw (A5.input 1) -- ([yshift=-0.15cm]A5.input 1) -| (A3.east);

	\node[and gate US,point up,logic gate inputs=nn,
			below left of = A4,yshift=-0.5cm] (A2)
		{\rotatebox{-90}{\gate{A2}}};
	\draw (A4.input 2) -- ([yshift=-0.15cm]A4.input 2) -| (A2.east);

	\node[and gate US,point up,logic gate inputs=nn,
			above left of = A4,yshift=0.5cm] (A1)
		{\rotatebox{-90}{\gate{A1}}};
	\draw (A4.input 1) -- ([yshift=-0.15cm]A4.input 1) -| (A1.east);

	\node[state, below right of = A7] (l8) {\leaf{l8}};
	\draw (A7.input 2) -- ([yshift=-0.15cm]A7.input 2) -| (l8);

	\node[state, below right of = A5] (l7) {\leaf{l7}};
	\draw (A5.input 2) -- ([yshift=-0.15cm]A5.input 2) -| (l7);

	\node[state, below right of = A3] (l6) {\leaf{l6}};
	\draw (A3.input 2) -- ([yshift=-0.15cm]A3.input 2) -| (l6);

	\node[state, below left of = A3] (l5) {\leaf{l5}};
	\draw (A3.input 1) -- ([yshift=-0.15cm]A3.input 1) -| (l5);

	\node[state, below right of = A2] (l4) {\leaf{l4}};
	\draw (A2.input 2) -- ([yshift=-0.15cm]A2.input 2) -| (l4);

	\node[state, below left of = A2] (l3) {\leaf{l3}};
	\draw (A2.input 1) -- ([yshift=-0.15cm]A2.input 1) -| (l3);

	\node[state, below right of = A1] (l2) {\leaf{l2}};
	\draw (A1.input 2) -- ([yshift=-0.15cm]A1.input 2) -| (l2);

	\node[state, below left of = A1] (l1) {\leaf{l1}};
	\draw (A1.input 1) -- ([yshift=-0.15cm]A1.input 1) -| (l1);
\end{tikzpicture}

%% file: examples/last.tex

\begin{figure}[ht]
  \centering
  \begin{tikzpicture}
    \node at (0,0)   {\scalebox{.5}{\input{examples/ADTree_last.tikz}}};
    \node at (4.5,0) {\scalebox{.555}{\input{examples/tabAttributes_last}}};
  \end{tikzpicture}
  \caption{Last example (\texttt{last})}
  \label{fig:last}
\end{figure}

\begin{figure}[!htb]
  \centering
    \scalebox{0.55}{
\begin{tikzpicture}[node distance=1.8cm]
  \tikzstyle{SEQ}=[diamond]
  \tikzstyle{NULL}=[trapezium, trapezium left angle=120, trapezium right angle=120, minimum size=8mm]
  \tikzstyle{AND}=[and gate US, rotate=90 ]
  \tikzstyle{OR}=[or gate US, rotate=90 ]
  \tikzset{every node/.style={ultra thick, draw=red, minimum size=6mm}}
  \node[draw=red, AND, logic gate inputs=nn, xshift=0.000000cm ] (a') {\rotatebox {-90}{\ensuremath{\mathtt{a'}}}};
  \node[draw=none, blue, xshift=2mm, yshift=7mm] at (a'.south) {\small{\ensuremath{\mathtt{level}}}};
  \node[draw=none, green!60!black, xshift=-2mm, yshift=7mm] at (a'.north) {\small{\ensuremath{\mathtt{depth}}}};
  \node[draw=none, blue, xshift=2mm, yshift=0mm] at (a'.south) {\small{\ensuremath{\mathtt{0}}}};
  \node[draw=none, green!60!black, xshift=-2mm, yshift=0mm] at (a'.north) {\small{\ensuremath{\mathtt{4}}}};
  \node[draw=red, SEQ, xshift=-2.250000cm , below of=a'] (b_1) {\ensuremath{\mathtt{b_{1}}}};
  \node[draw=none, blue, xshift=2mm, yshift=0mm] at (b_1.east) {\small{\ensuremath{\mathtt{0}}}};
  \node[draw=none, green!60!black, xshift=-2mm, yshift=0mm] at (b_1.west) {\small{\ensuremath{\mathtt{4}}}};
  \draw[solid] (a'.input 1) edge (b_1);
  \node[draw=red, NULL, xshift=2.250000cm , below of=a'] (e') {\ensuremath{\mathtt{e'}}};
  \node[draw=none, blue, xshift=2mm, yshift=0mm] at (e'.east) {\small{\ensuremath{\mathtt{0}}}};
  \node[draw=none, green!60!black, xshift=-2mm, yshift=0mm] at (e'.west) {\small{\ensuremath{\mathtt{1}}}};
  \draw[solid] (a'.input 2) edge (e');
  \node[draw=red, AND, logic gate inputs=nn, xshift=-0.000000cm , yshift=4mm, below = 1.4cm of e'.south] (f') {\rotatebox {-90}{\ensuremath{\mathtt{f'}}}};
  \node[draw=none, blue, xshift=2mm, yshift=0mm] at (f'.south) {\small{\ensuremath{\mathtt{0}}}};
  \node[draw=none, green!60!black, xshift=-2mm, yshift=0mm] at (f'.north) {\small{\ensuremath{\mathtt{1}}}};
  \draw[solid] (e') edge (f'.east);
  \node[draw=red, NULL, xshift=-1.250000cm , below of=f'] (h') {\ensuremath{\mathtt{h'}}};
  \node[draw=none, blue, xshift=2mm, yshift=0mm] at (h'.east) {\small{\ensuremath{\mathtt{0}}}};
  \node[draw=none, green!60!black, xshift=-2mm, yshift=0mm] at (h'.west) {\small{\ensuremath{\mathtt{1}}}};
  \draw[solid] (f'.input 1) edge (h');
  \node[draw=red, NULL, xshift=1.250000cm , below of=f'] (i') {\ensuremath{\mathtt{i'}}};
  \node[draw=none, blue, xshift=2mm, yshift=0mm] at (i'.east) {\small{\ensuremath{\mathtt{0}}}};
  \node[draw=none, green!60!black, xshift=-2mm, yshift=0mm] at (i'.west) {\small{\ensuremath{\mathtt{1}}}};
  \draw[solid] (f'.input 2) edge (i');
  \node[draw=red, SEQ, xshift=-0.000000cm , below of=i'] (l_1) {\ensuremath{\mathtt{l_{1}}}};
  \node[draw=none, blue, xshift=2mm, yshift=0mm] at (l_1.east) {\small{\ensuremath{\mathtt{0}}}};
  \node[draw=none, green!60!black, xshift=-2mm, yshift=0mm] at (l_1.west) {\small{\ensuremath{\mathtt{1}}}};
  \draw[solid] (i') edge (l_1);
  \node[draw=red, state, xshift=-0.000000cm , below of=l_1] (l') {\ensuremath{\mathtt{l'}}};
  \node[draw=none, blue, xshift=2mm, yshift=0mm] at (l'.east) {\small{\ensuremath{\mathtt{1}}}};
  \node[draw=none, green!60!black, xshift=-2mm, yshift=0mm] at (l'.west) {\small{\ensuremath{\mathtt{0}}}};
  \draw[solid] (l_1) edge (l');
  \node[draw=red, SEQ, xshift=-0.000000cm , below of=h'] (j_1) {\ensuremath{\mathtt{j_{1}}}};
  \node[draw=none, blue, xshift=2mm, yshift=0mm] at (j_1.east) {\small{\ensuremath{\mathtt{0}}}};
  \node[draw=none, green!60!black, xshift=-2mm, yshift=0mm] at (j_1.west) {\small{\ensuremath{\mathtt{1}}}};
  \draw[solid] (h') edge (j_1);
  \node[draw=red, state, xshift=-0.000000cm , below of=j_1] (j') {\ensuremath{\mathtt{j'}}};
  \node[draw=none, blue, xshift=2mm, yshift=0mm] at (j'.east) {\small{\ensuremath{\mathtt{1}}}};
  \node[draw=none, green!60!black, xshift=-2mm, yshift=0mm] at (j'.west) {\small{\ensuremath{\mathtt{0}}}};
  \draw[solid] (j_1) edge (j');
  \node[draw=red, AND, logic gate inputs=nn, xshift=-0.000000cm , yshift=4mm, below = 1.4cm of b_1.south] (b') {\rotatebox {-90}{\ensuremath{\mathtt{b'}}}};
  \node[draw=none, blue, xshift=2mm, yshift=0mm] at (b'.south) {\small{\ensuremath{\mathtt{1}}}};
  \node[draw=none, green!60!black, xshift=-2mm, yshift=0mm] at (b'.north) {\small{\ensuremath{\mathtt{3}}}};
  \draw[solid] (b_1) edge (b'.east);
  \node[draw=red, SEQ, xshift=-1.250000cm , below of=b'] (c_1) {\ensuremath{\mathtt{c_{1}}}};
  \node[draw=none, blue, xshift=2mm, yshift=0mm] at (c_1.east) {\small{\ensuremath{\mathtt{1}}}};
  \node[draw=none, green!60!black, xshift=-2mm, yshift=0mm] at (c_1.west) {\small{\ensuremath{\mathtt{1}}}};
  \draw[solid] (b'.input 1) edge (c_1);
  \node[draw=red, SEQ, xshift=1.250000cm , below of=b'] (d_3) {\ensuremath{\mathtt{d_{3}}}};
  \node[draw=none, blue, xshift=2mm, yshift=0mm] at (d_3.east) {\small{\ensuremath{\mathtt{1}}}};
  \node[draw=none, green!60!black, xshift=-2mm, yshift=0mm] at (d_3.west) {\small{\ensuremath{\mathtt{3}}}};
  \draw[solid] (b'.input 2) edge (d_3);
  \node[draw=red, SEQ, xshift=-0.000000cm , below of=d_3] (d_2) {\ensuremath{\mathtt{d_{2}}}};
  \node[draw=none, blue, xshift=2mm, yshift=0mm] at (d_2.east) {\small{\ensuremath{\mathtt{2}}}};
  \node[draw=none, green!60!black, xshift=-2mm, yshift=0mm] at (d_2.west) {\small{\ensuremath{\mathtt{2}}}};
  \draw[solid] (d_3) edge (d_2);
  \node[draw=red, SEQ, xshift=-0.000000cm , below of=d_2] (d_1) {\ensuremath{\mathtt{d_{1}}}};
  \node[draw=none, blue, xshift=2mm, yshift=0mm] at (d_1.east) {\small{\ensuremath{\mathtt{3}}}};
  \node[draw=none, green!60!black, xshift=-2mm, yshift=0mm] at (d_1.west) {\small{\ensuremath{\mathtt{1}}}};
  \draw[solid] (d_2) edge (d_1);
  \node[draw=red, state, xshift=-0.000000cm , below of=d_1] (d') {\ensuremath{\mathtt{d'}}};
  \node[draw=none, blue, xshift=2mm, yshift=0mm] at (d'.east) {\small{\ensuremath{\mathtt{4}}}};
  \node[draw=none, green!60!black, xshift=-2mm, yshift=0mm] at (d'.west) {\small{\ensuremath{\mathtt{0}}}};
  \draw[solid] (d_1) edge (d');
  \node[draw=red, state, xshift=-0.000000cm , below of=c_1] (c') {\ensuremath{\mathtt{c'}}};
  \node[draw=none, blue, xshift=2mm, yshift=0mm] at (c'.east) {\small{\ensuremath{\mathtt{2}}}};
  \node[draw=none, green!60!black, xshift=-2mm, yshift=0mm] at (c'.west) {\small{\ensuremath{\mathtt{0}}}};
  \draw[solid] (c_1) edge (c');
\end{tikzpicture}
    }
    \caption{Preprocessing the \texttt{last} \ADT/}
\end{figure}

\begin{table}[!!htb]
    \centering
    \rowcolors{2}{lightgray!30}{white}

      \begin{tabular}{c|l|l|}
        \diagbox[]{slot}{agent} & 1                         & 2                         \\
        \hline
1&{\ensuremath{\mathtt{d',d_{1}}}}&{\ensuremath{\mathtt{}}}\\
2&{\ensuremath{\mathtt{d_{2}}}}&{\ensuremath{\mathtt{i',l',l_{1}}}}\\
3&{\ensuremath{\mathtt{d_{3}}}}&{\ensuremath{\mathtt{c',c_{1}}}}\\
4&{\ensuremath{\mathtt{a',b',b_{1}}}}&{\ensuremath{\mathtt{e',f',h',j',j_{1}}}}\\
        \hline
      \end{tabular}
    \caption{Assignment for \texttt{last}}
\end{table}

%% file: examples/ADTree_last.tikz

\begin{tikzpicture}
	[every node/.style={ultra thick,draw=red,minimum size=6mm},
	node distance=1.8cm]

	\node[and gate US,point up,logic gate inputs=nn] (a)
		{\rotatebox{-90}{\leaf{a}}};

	\node[and gate US,point up,logic gate inputs=nn,
			above left of = a,yshift=1cm] (b)
		{\rotatebox{-90}{\leaf{b}}};
	\draw (a.input 1) -- ([yshift=-0.15cm]a.input 1) -| (b.east);

	\node[and gate US,point up,logic gate inputs=ni,
			below left of = a, yshift=-1cm] (e)
		{\rotatebox{-90}{\leaf{e}}};
	\draw (a.input 2) -- ([yshift=-0.15cm]a.input 2) -| (e.east);

	\node[and gate US,point up,logic gate inputs=nn,
			above left of = e,yshift=0.5cm] (f)
		{\rotatebox{-90}{\leaf{f}}};
	\draw (e.input 1) -- ([yshift=-0.15cm]e.input 1) -| (f.east);

	\node[and gate US,point up,logic gate inputs=ni,
			above left of = f,yshift=0.5cm] (h)
		{\rotatebox{-90}{\leaf{h}}};
	\draw (f.input 1) -- ([yshift=-0.15cm]f.input 1) -| (h.east);

	\node[and gate US,point up,logic gate inputs=ni,
			below left of = f,yshift=-0.5cm] (i)
		{\rotatebox{-90}{\leaf{i}}};
	\draw (f.input 2) -- ([yshift=-0.15cm]f.input 2) -| (i.east);

	\node[state, below left of = b] (c) {\leaf{c}};
	\draw (b.input 1) -- ([yshift=-0.15cm]b.input 1) -| (c);

	\node[state, below right of = b] (d) {\leaf{d}};
	\draw (b.input 2) -- ([yshift=-0.15cm]b.input 2) -| (d);

	\node[rectangle,draw=Green,minimum size=8mm, below right of = e]
		(g) {\leaf{g}};
	\draw (e.input 2) -- ([yshift=-0.1cm]e.input 2) -| (g);

	\node[state, below left of = h] (j) {\leaf{j}};
	\draw (h.input 1) -- ([yshift=-0.24cm]h.input 1) -| (j);

	\node[rectangle,draw=Green,minimum size=8mm, below right of = h]
		(k) {\leaf{k}};
	\draw (h.input 2) -- ([yshift=-0.1cm]h.input 2) -| (k);

	\node[state, below left of = i] (l) {\leaf{l}};
	\draw (i.input 1) -- ([yshift=-0.24cm]i.input 1) -| (l);

	\node[rectangle,draw=Green,minimum size=8mm, below right of = i]
		(m) {\leaf{m}};
	\draw (i.input 2) -- ([yshift=-0.1cm]i.input 2) -| (m);

\end{tikzpicture}

%% file: examples/scaling_example.tex

\begin{figure}[ht]
  \centering
  \begin{tikzpicture}
    \node at (0,0)   {\scalebox{.5}{\input{examples/ADTree_scaling.tikz}}};
    \node at (4.5,0) {\scalebox{.555}{\input{examples/tabAttributes_scaling}}};
  \end{tikzpicture}
  \caption{Scaling example (\texttt{scaling})}
  \label{fig:scaling}
\end{figure}

\begin{figure}[!htb]
  \centering
    \scalebox{0.55}{
\begin{tikzpicture}[node distance=1.8cm]
  \tikzstyle{SEQ}=[diamond]
  \tikzstyle{NULL}=[trapezium, trapezium left angle=120, trapezium right angle=120, minimum size=8mm]
  \tikzstyle{AND}=[and gate US, rotate=90 ]
  \tikzstyle{OR}=[or gate US, rotate=90 ]
  \tikzset{every node/.style={ultra thick, draw=red, minimum size=6mm}}
  \node[draw=red, SEQ, xshift=0.000000cm ] (a_1) {\ensuremath{\mathtt{a_{1}}}};
  \node[draw=none, blue, xshift=2mm, yshift=7mm] at (a_1.east) {\small{\ensuremath{\mathtt{level}}}};
  \node[draw=none, green!60!black, xshift=-2mm, yshift=7mm] at (a_1.west) {\small{\ensuremath{\mathtt{depth}}}};
  \node[draw=none, blue, xshift=2mm, yshift=0mm] at (a_1.east) {\small{\ensuremath{\mathtt{0}}}};
  \node[draw=none, green!60!black, xshift=-2mm, yshift=0mm] at (a_1.west) {\small{\ensuremath{\mathtt{5}}}};
  \node[draw=red, AND, logic gate inputs=nn, xshift=-0.000000cm , yshift=4mm, below = 1.4cm of a_1.south] (a') {\rotatebox {-90}{\ensuremath{\mathtt{a'}}}};
  \draw[solid] (a_1) edge (a'.east);
  \node[draw=none, blue, xshift=2mm, yshift=0mm] at (a'.south) {\small{\ensuremath{
  \mathtt{1}}}};
  \node[draw=none, green!60!black, xshift=-2mm, yshift=0mm] at (a'.north) {\small{\ensuremath{\mathtt{4}}}};
  \node[draw=red, SEQ, xshift=-2.250000cm , below of=a'] (b_1) {\ensuremath{\mathtt{b_{1}}}};
  \node[draw=none, blue, xshift=2mm, yshift=0mm] at (b_1.east) {\small{\ensuremath{
  \mathtt{1}}}};
  \node[draw=none, green!60!black, xshift=-2mm, yshift=0mm] at (b_1.west) {\small{\ensuremath{\mathtt{4}}}};
  \draw[solid] (a'.input 1) edge (b_1);
  \node[draw=red, SEQ, xshift=2.250000cm , below of=a'] (c_1) {\ensuremath{\mathtt{c_{1}}}};
  \node[draw=none, blue, xshift=2mm, yshift=0mm] at (c_1.east) {\small{\ensuremath{\mathtt{1}}}};
  \node[draw=none, green!60!black, xshift=-2mm, yshift=0mm] at (c_1.west) {\small{\ensuremath{\mathtt{2}}}};
  \draw[solid] (a'.input 2) edge (c_1);
  \node[draw=red, AND, logic gate inputs=nn, xshift=-0.000000cm , yshift=4mm, below = 1.4cm of c_1.south] (c') {\rotatebox {-90}{\ensuremath{\mathtt{c'}}}};
  \node[draw=none, blue, xshift=2mm, yshift=0mm] at (c'.south) {\small{\ensuremath{\mathtt{2}}}};
  \node[draw=none, green!60!black, xshift=-2mm, yshift=0mm] at (c'.north) {\small{\ensuremath{\mathtt{1}}}};
  \draw[solid] (c_1) edge (c'.east);
  \node[draw=red, SEQ, xshift=-1.250000cm , below of=c'] (f_1) {\ensuremath{\mathtt{f_{1}}}};
  \node[draw=none, blue, xshift=2mm, yshift=0mm] at (f_1.east) {\small{\ensuremath{\mathtt{2}}}};
  \node[draw=none, green!60!black, xshift=-2mm, yshift=0mm] at (f_1.west) {\small{
  \ensuremath{\mathtt{1}}}};
  \draw[solid] (c'.input 1) edge (f_1);
  \node[draw=red, SEQ, xshift=1.250000cm , below of=c'] (g_1) {\ensuremath{\mathtt{g_{1}}}};
  \node[draw=none, blue, xshift=2mm, yshift=0mm] at (g_1.east) {\small{\ensuremath{\mathtt{2}}}};
  \node[draw=none, green!60!black, xshift=-2mm, yshift=0mm] at (g_1.west) {\small{\ensuremath{\mathtt{1}}}};
  \draw[solid] (c'.input 2) edge (g_1);
  \node[draw=red, state, xshift=-0.000000cm , below of=g_1] (g') {\ensuremath{\mathtt{g'}}};
  \node[draw=none, blue, xshift=2mm, yshift=0mm] at (g'.east) {\small{\ensuremath{\mathtt{3}}}};
  \node[draw=none, green!60!black, xshift=-2mm, yshift=0mm] at (g'.west) {\small{\ensuremath{\mathtt{0}}}};
  \draw[solid] (g_1) edge (g');
  \node[draw=red, state, xshift=-0.000000cm , below of=f_1] (f') {\ensuremath{\mathtt{f'}}};
  \node[draw=none, blue, xshift=2mm, yshift=0mm] at (f'.east) {\small{\ensuremath{\mathtt{3}}}};
  \node[draw=none, green!60!black, xshift=-2mm, yshift=0mm] at (f'.west) {\small{\ensuremath{\mathtt{0}}}};
  \draw[solid] (f_1) edge (f');
  \node[draw=red, AND, logic gate inputs=nn, xshift=-0.000000cm , yshift=4mm, below = 1.4cm of b_1.south] (b') {\rotatebox {-90}{\ensuremath{\mathtt{b'}}}};
  \node[draw=none, blue, xshift=2mm, yshift=0mm] at (b'.south) {\small{\ensuremath{\mathtt{2}}}};
  \node[draw=none, green!60!black, xshift=-2mm, yshift=0mm] at (b'.north) {\small{\ensuremath{\mathtt{3}}}};
  \draw[solid] (b_1) edge (b'.east);
  \node[draw=red, SEQ, xshift=-1.250000cm , below of=b'] (d_1) {\ensuremath{\mathtt{d_{1}}}};
  \node[draw=none, blue, xshift=2mm, yshift=0mm] at (d_1.east) {\small{\ensuremath{\mathtt{2}}}};
  \node[draw=none, green!60!black, xshift=-2mm, yshift=0mm] at (d_1.west) {\small{\ensuremath{\mathtt{1}}}};
  \draw[solid] (b'.input 1) edge (d_1);
  \node[draw=red, SEQ, xshift=1.250000cm , below of=b'] (e_3) {\ensuremath{\mathtt{e_{3}}}};
  \node[draw=none, blue, xshift=2mm, yshift=0mm] at (e_3.east) {\small{\ensuremath{\mathtt{2}}}};
  \node[draw=none, green!60!black, xshift=-2mm, yshift=0mm] at (e_3.west) {\small{\ensuremath{\mathtt{3}}}};
  \draw[solid] (b'.input 2) edge (e_3);
  \node[draw=red, SEQ, xshift=-0.000000cm , below of=e_3] (e_2) {\ensuremath{\mathtt{e_{2}}}};
  \node[draw=none, blue, xshift=2mm, yshift=0mm] at (e_2.east) {\small{\ensuremath{\mathtt{3}}}};
  \node[draw=none, green!60!black, xshift=-2mm, yshift=0mm] at (e_2.west) {\small{\ensuremath{\mathtt{2}}}};
  \draw[solid] (e_3) edge (e_2);
  \node[draw=red, SEQ, xshift=-0.000000cm , below of=e_2] (e_1) {\ensuremath{\mathtt{e_{1}}}};
  \node[draw=none, blue, xshift=2mm, yshift=0mm] at (e_1.east) {\small{\ensuremath{\mathtt{4}}}};
  \node[draw=none, green!60!black, xshift=-2mm, yshift=0mm] at (e_1.west) {\small{\ensuremath{\mathtt{1}}}};
  \draw[solid] (e_2) edge (e_1);
  \node[draw=red, state, xshift=-0.000000cm , below of=e_1] (e') {\ensuremath{\mathtt{e'}}};
  \node[draw=none, blue, xshift=2mm, yshift=0mm] at (e'.east) {\small{\ensuremath{\mathtt{5}}}};
  \node[draw=none, green!60!black, xshift=-2mm, yshift=0mm] at (e'.west) {\small{\ensuremath{\mathtt{0}}}};
  \draw[solid] (e_1) edge (e');
  \node[draw=red, state, xshift=-0.000000cm , below of=d_1] (d') {\ensuremath{\mathtt{d'}}};
  \node[draw=none, blue, xshift=2mm, yshift=0mm] at (d'.east) {\small{\ensuremath{\mathtt{3}}}};
  \node[draw=none, green!60!black, xshift=-2mm, yshift=0mm] at (d'.west) {\small{\ensuremath{\mathtt{0}}}};
  \draw[solid] (d_1) edge (d');
\end{tikzpicture}
    }
    \caption{Preprocessing the \texttt{scaling} \ADT/}
\end{figure}

\begin{table}[!!htb]
    \centering
    \rowcolors{2}{lightgray!30}{white}

      \begin{tabular}{c|l|l|}
        \diagbox[]{slot}{agent} & 1                         & 2                         \\
        \hline
1&{\ensuremath{\mathtt{e',e_{1}}}}&{\ensuremath{\mathtt{g',g_{1}}}}\\
2&{\ensuremath{\mathtt{e_{2}}}}&{\ensuremath{\mathtt{f',f_{1}}}}\\
3&{\ensuremath{\mathtt{e_{3}}}}&{\ensuremath{\mathtt{d',d_{1}}}}\\
4&{\ensuremath{\mathtt{b',b_{1}}}}&{\ensuremath{\mathtt{c',c_{1}}}}\\
5&{\ensuremath{\mathtt{a',a_{1}}}}&\\
        \hline
      \end{tabular}
    \caption{Assignment for \texttt{scaling}}
\end{table}

%% file: examples/ADTree_scaling.tikz

\begin{tikzpicture}
	[every node/.style={ultra thick,draw=red,minimum size=6mm},
	node distance=1.8cm]

	\node[and gate US,point up,logic gate inputs=nn] (a)
		{\rotatebox{-90}{\leaf{a}}};

	\node[and gate US,point up,logic gate inputs=nn,
			above left of = a,yshift=0.5cm] (b)
		{\rotatebox{-90}{\leaf{b}}};
	\draw (a.input 1) -- ([yshift=-0.15cm]a.input 1) -| (b.east);

	\node[and gate US,point up,logic gate inputs=nn,
			below left of = a, yshift=-0.5cm] (c)
		{\rotatebox{-90}{\leaf{c}}};
	\draw (a.input 2) -- ([yshift=-0.15cm]a.input 2) -| (c.east);

	\node[state, below left of = b] (d) {\leaf{d}};
	\draw (b.input 1) -- ([yshift=-0.15cm]b.input 1) -| (d);

	\node[state, below right of = b] (e) {\leaf{e}};
	\draw (b.input 2) -- ([yshift=-0.15cm]b.input 2) -| (e);

	\node[state, below left of = c] (f) {\leaf{f}};
	\draw (c.input 1) -- ([yshift=-0.15cm]c.input 1) -| (f);

	\node[state, below right of = c] (g) {\leaf{g}};
	\draw (c.input 2) -- ([yshift=-0.15cm]c.input 2) -| (g);
\end{tikzpicture}

%% file: examples/scaling_table.tex
\begin{table}
\centering
\caption{Scalability of different agent configurations}
\label{scalability}
\begin{tabular}{|r|r|r|r||r|r|}
\hline
\textbf{depth} & \textbf{width} & \textbf{\# children} & \textbf{|ADTree|} & \textbf{\# agents} & \textbf{\# slots} \\
\hline\hline
    2 &     2 &          2 &        5 &        2 &       3 \\
    2 &     4 &          2 &        7 &        3 &       4 \\
    2 &     4 &          4 &        9 &        4 &       3 \\
    2 &     6 &          6 &       13 &        6 &       3 \\
    2 &     8 &          4 &       13 &        3 &       6 \\
    2 &     8 &          8 &       17 &        8 &       3 \\
    2 &    10 &         10 &       21 &       10 &       3 \\
    2 &    12 &          4 &       17 &        3 &      10 \\
    2 &    12 &          6 &       19 &        4 &       8 \\
    2 &    16 &          8 &       25 &        4 &      10 \\
    3 &     2 &          2 &        7 &        2 &       4 \\
    3 &     4 &          2 &        9 &        3 &       5 \\
    3 &     4 &          4 &       13 &        4 &       4 \\
    3 &     6 &          2 &       13 &        3 &       7 \\
    3 &     6 &          6 &       19 &        6 &       4 \\
    3 &     8 &          2 &       15 &        3 &       9 \\
    3 &     8 &          4 &       17 &        4 &       7 \\
    3 &     8 &          8 &       25 &        8 &       4 \\
    3 &    10 &         10 &       31 &       10 &       4 \\
    3 &    12 &          4 &       21 &        3 &      11 \\
    3 &    12 &          6 &       25 &        4 &       9 \\
    3 &    16 &          8 &       33 &        4 &      11 \\
    4 &     2 &          2 &        9 &        2 &       5 \\
    4 &     4 &          2 &       11 &        3 &       6 \\
    4 &     4 &          4 &       17 &        4 &       5 \\
    4 &     6 &          2 &       15 &        3 &       8 \\
    4 &     6 &          6 &       25 &        6 &       5 \\
    4 &     8 &          2 &       17 &        3 &      10 \\
    4 &     8 &          4 &       21 &        4 &       8 \\
    4 &     8 &          8 &       33 &        8 &       5 \\
    4 &    10 &          2 &       23 &        3 &      12 \\
    4 &    10 &         10 &       41 &       10 &       5 \\
    4 &    12 &          4 &       25 &        3 &      12 \\
    4 &    12 &          6 &       31 &        4 &      10 \\
    4 &    16 &          8 &       41 &        5 &      12 \\
    5 &     2 &          2 &       11 &        2 &       6 \\
    5 &     4 &          2 &       13 &        3 &       7 \\
    5 &     4 &          4 &       21 &        4 &       6 \\
    5 &     6 &          2 &       17 &        3 &       9 \\
    5 &     6 &          6 &       31 &        6 &       6 \\
    5 &     8 &          2 &       19 &        3 &      11 \\
    5 &     8 &          4 &       25 &        4 &       9 \\
    5 &     8 &          8 &       41 &        8 &       6 \\
    5 &    10 &          2 &       25 &        3 &      13 \\
    5 &    10 &         10 &       51 &       10 &       6 \\
    5 &    12 &          4 &       29 &        3 &      13 \\
    5 &    12 &          6 &       37 &        5 &      11 \\
    5 &    16 &          8 &       49 &        5 &      13 \\
\hline
\end{tabular}
\end{table}

%% file: main-ICECCS.bbl
\begin{thebibliography}{10}
\providecommand{\url}[1]{#1}
\csname url@samestyle\endcsname
\providecommand{\newblock}{\relax}
\providecommand{\bibinfo}[2]{#2}
\providecommand{\BIBentrySTDinterwordspacing}{\spaceskip=0pt\relax}
\providecommand{\BIBentryALTinterwordstretchfactor}{4}
\providecommand{\BIBentryALTinterwordspacing}{\spaceskip=\fontdimen2\font plus
\BIBentryALTinterwordstretchfactor\fontdimen3\font minus
  \fontdimen4\font\relax}
\providecommand{\BIBforeignlanguage}[2]{{%
\expandafter\ifx\csname l@#1\endcsname\relax
\typeout{** WARNING: IEEEtran.bst: No hyphenation pattern has been}%
\typeout{** loaded for the language `#1'. Using the pattern for}%
\typeout{** the default language instead.}%
\else
\language=\csname l@#1\endcsname
\fi
#2}}
\providecommand{\BIBdecl}{\relax}
\BIBdecl

\bibitem{Wooldridge02a}
M.~J. Wooldridge, \emph{An Introduction to Multiagent Systems}.\hskip 1em plus
  0.5em minus 0.4em\relax John Wiley \& Sons, 2002.

\bibitem{KordyMRS10}
B.~Kordy, S.~Mauw, S.~Radomirovi{\'{c}}, and P.~Schweitzer, ``Foundations of
  {A}ttack-{D}efense {T}rees,'' in \emph{Proceedings of {FAST}~2010}.\hskip 1em
  plus 0.5em minus 0.4em\relax Springer, 2011, pp. 80--95.

\bibitem{Zaruhi_pareto_2015}
Z.~Aslanyan and F.~Nielson, ``Pareto {Efficient} {Solutions} of
  {Attack}-{Defence} {Trees},'' in \emph{Proceedings of {POST} 2015}.\hskip 1em
  plus 0.5em minus 0.4em\relax Springer, 2015, pp. 95--114.

\bibitem{ICECCS2019}
L.~Petrucci, M.~Knapik, W.~Penczek, and T.~Sidoruk, ``Squeezing {S}tate
  {S}paces of ({A}ttack-{D}efence) {T}rees,'' in \emph{Proceedings of {ICECCS}
  2019}.\hskip 1em plus 0.5em minus 0.4em\relax {IEEE}, 2019, pp. 71--80.

\bibitem{ICFEM2020}
J.~Arias, C.~E. Budde, W.~Penczek, L.~Petrucci, T.~Sidoruk, and M.~Stoelinga,
  ``Hackers vs. {S}ecurity: {A}ttack-{D}efence {T}rees as {A}synchronous
  {M}ulti-agent {S}ystems,'' in \emph{Proceedings of {ICFEM} 2020}.\hskip 1em
  plus 0.5em minus 0.4em\relax Springer, 2020, pp. 3--19.

\bibitem{KMRS14}
B.~Kordy, S.~Mauw, S.~Radomirovi{\'{c}}, and P.~Schweitzer, ``Attack-{D}efense
  {T}rees,'' \emph{Journal of {L}ogic and {C}omputation}, vol.~24, no.~1, pp.
  55--87, 2014.

\bibitem{Gadyatskaya_2016}
O.~Gadyatskaya, R.~R. Hansen, K.~G. Larsen, A.~Legay, M.~C. Olesen, and D.~B.
  Poulsen, ``Modelling {Attack}-{D}efense {Trees} {Using} {Timed} {Automata},''
  in \emph{Formal {Modeling} and {Analysis} of {Timed} {Systems}}.\hskip 1em
  plus 0.5em minus 0.4em\relax Springer, 2016, pp. 35--50.

\bibitem{AGKS15}
F.~Arnold, D.~Guck, R.~Kumar, and M.~Stoelinga, ``Sequential and {P}arallel
  {A}ttack {T}ree {M}odelling,'' in \emph{Proceedings of {SAFECOMP}
  2015}.\hskip 1em plus 0.5em minus 0.4em\relax Springer, 2015, pp. 291--299.

\bibitem{GIM15}
M.~Gribaudo, M.~Iacono, and S.~Marrone, ``Exploiting {B}ayesian {N}etworks for
  the {A}nalysis of {C}ombined {A}ttack {T}rees,'' \emph{Electronic Notes in
  Theoretical Computer Science}, vol. 310, pp. 91--111, 2015.

\bibitem{ANP16}
Z.~Aslanyan, F.~Nielson, and D.~Parker, ``Quantitative {V}erification and
  {S}ynthesis of {A}ttack-{D}efence {S}cenarios,'' in \emph{Proceedings of
  {CSF} 2016}.\hskip 1em plus 0.5em minus 0.4em\relax {IEEE}, 2016, pp.
  105--119.

\bibitem{KPCS13}
B.~Kordy, L.~Pi{\`{e}}tre{-}Cambac{\'{e}}d{\`{e}}s, and P.~Schweitzer,
  ``{DAG}-based {A}ttack and {D}efense {M}odeling: {D}on't {M}iss the {F}orest
  for the {A}ttack {T}rees,'' \emph{Computer Science Review}, vol. 13-14, pp.
  1--38, 2014.

\bibitem{WAFP19}
W.~Widel, M.~Audinot, B.~Fila, and S.~Pinchinat, ``Beyond 2014: {F}ormal
  {M}ethods for {A}ttack {T}ree-based {S}ecurity {M}odeling,'' \emph{{{ACM}
  Computing Surveys}}, vol.~52, no.~4, pp. 75:1--75:36, 2019.

\bibitem{FilaW20}
B.~Fila and W.~Widel, ``Exploiting {A}ttack-{D}efense {T}rees to {F}ind an
  {O}ptimal {S}et of {C}ountermeasures,'' in \emph{Proceedings of {CSF}
  2020}.\hskip 1em plus 0.5em minus 0.4em\relax {IEEE}, 2020, pp. 395--410.

\bibitem{AAMASWJWPPDAM2018a}
W.~Jamroga, W.~Penczek, P.~Dembinski, and A.~Mazurkiewicz, ``Towards {P}artial
  {O}rder {R}eductions for {S}trategic {A}bility,'' in \emph{Proceedings of
  {AAMAS} '18}.\hskip 1em plus 0.5em minus 0.4em\relax {ACM}, 2018, pp.
  156--165.

\bibitem{CompSched}
T.~L. Adam, K.~M. Chandy, and J.~R. Dickson, ``A {C}omparison of {L}ist
  {S}chedules for {P}arallel {P}rocessing {S}ystems,'' \emph{Communications of
  the {ACM}}, vol.~17, no.~12, p. 685–690, Dec. 1974.

\bibitem{KwokA99}
Y.~Kwok and I.~Ahmad, ``Static {S}cheduling {A}lgorithms for {A}llocating
  {D}irected {T}ask {G}raphs to {M}ultiprocessors,'' \emph{{{ACM} Computing
  Surveys}}, vol.~31, no.~4, pp. 406--471, 1999.

\bibitem{Hu1961}
T.~C. Hu, ``Parallel {S}equencing and {A}ssembly {L}ine {P}roblems,''
  \emph{Operations Research}, vol.~9, no.~6, pp. 841--848, 1961.

\bibitem{PapaYan1979}
C.~H. Papadimitriou and M.~Yannakakis, ``Scheduling {I}nterval-{O}rdered
  {T}asks,'' \emph{{SIAM} Journal on Computing}, vol.~8, no.~3, pp. 405--409,
  1979.

\bibitem{Towsley86}
D.~F. Towsley, ``Allocating {P}rograms {C}ontaining {B}ranches and {L}oops
  {W}ithin a {M}ultiple {P}rocessor {S}ystem,'' \emph{{IEEE} Transactions on
  Software Engineering}, vol.~12, no.~10, pp. 1018--1024, 1986.

\bibitem{ElRewiniA95}
H.~El{-}Rewini and H.~H. Ali, ``Static {S}cheduling of {C}onditional {B}ranches
  in {P}arallel {P}rograms,'' \emph{{Journal of Parallel and Distributed
  Computing}}, vol.~24, no.~1, pp. 41--54, 1995.

\bibitem{NunesL14}
I.~Nunes and M.~Luck, ``Softgoal-based {P}lan {S}election in {M}odel-driven
  {BDI} {A}gents,'' in \emph{Proceedings of {AAMAS} '14}.\hskip 1em plus 0.5em
  minus 0.4em\relax IFAAMAS, 2014, pp. 749--756.

\bibitem{DannT0L20}
M.~Dann, J.~Thangarajah, Y.~Yao, and B.~Logan, ``Intention-{A}ware {M}ultiagent
  {S}cheduling,'' in \emph{Proceedings of {AAMAS} '20}.\hskip 1em plus 0.5em
  minus 0.4em\relax IFAAMAS, 2020, pp. 285--293.

\bibitem{ADT2AMASDemoPaper}
J.~Arias, W.~Penczek, L.~Petrucci, and T.~Sidoruk, ``{ADT2AMAS:} {M}anaging
  {A}gents in {A}ttack-{D}efence {S}cenarios,'' in \emph{Proceedings of {AAMAS}
  '21}.\hskip 1em plus 0.5em minus 0.4em\relax {ACM}, 2021, pp. 1749--1751.

\bibitem{adt2amas}
``{ADT2AMAS}: {T}ool for {M}anaging {A}gents in {A}ttack-{D}efense
  {S}cenarios,''
  \url{https://depot.lipn.univ-paris13.fr/parties/tools/adt2amas}.

\bibitem{KnapikAPJP19}
M.~Knapik, {\'{E}}.~Andr{\'{e}}, L.~Petrucci, W.~Jamroga, and W.~Penczek,
  ``Timed {ATL:} {F}orget {M}emory, {J}ust {C}ount,'' \emph{Journal of
  Artificial Intelligence Research}, vol.~66, pp. 197--223, 2019.

\end{thebibliography}
